%% file: redux.tex
\PassOptionsToPackage{nosumlimits,nonamelimits}{amsmath}
\PassOptionsToPackage{colorlinks,linkcolor={blue},citecolor={blue},urlcolor={red},breaklinks=true,final}{hyperref}

\documentclass[acmsmall,nonacm]{acmart}

\usepackage{ifthen}
\newboolean{arxiv}
\setboolean{arxiv}{true} %
\newcommand{\ifarx}[2]{\ifthenelse{\boolean{arxiv}}{#1}{#2}}

\usepackage{stel-common}

\input{catprog.tex}

\input{additional-macros.tex}

\let\oldcheckmark\checkmark
\renewcommand{\checkmark}{\raisebox{-4pt}{\scalebox{1.2}[.65]{$\oldcheckmark$}}}

\usepackage[utf8]{inputenc}

\usepackage{hypcap}
\usepackage{proof}
\usepackage{microtype}
\usepackage[strict]{changepage}

\usepackage{enumitem}
\setlist[enumerate,1]{label=(\arabic*),font=\normalfont,align=left,leftmargin=0pt,labelindent=0pt,listparindent=\parindent,labelwidth=0pt,itemindent=!,topsep=2pt,parsep=0pt,itemsep=2pt,start=1}
\setlist[enumerate,2]{label=(\alph*),font=\normalfont,labelindent=*,leftmargin=*,start=1}
\setlist[itemize]{labelindent=*,leftmargin=*}
\setlist[description]{labelindent=*,leftmargin=*,itemindent=-1 em}

\usepackage{ifdraft}

\ifdraft{
  \usepackage[draft]{showlabels}

  \usepackage[footnote,marginclue,nomargin,draft]{fixme}

  \usepackage[displaymath,mathlines,switch]{lineno}
	\linenumbers

	\usepackage{etoolbox} %

	\newcommand*\linenomathpatch[1]{%
		\cspreto{#1}{\linenomath}%
		\cspreto{#1*}{\linenomath}%
		\csappto{end#1}{\endlinenomath}%
		\csappto{end#1*}{\endlinenomath}%
	}

	\linenomathpatch{equation}
	\linenomathpatch{gather}
	\linenomathpatch{multline}
	\linenomathpatch{align}
	\linenomathpatch{alignat}
	\linenomathpatch{flalign}
}{
 \usepackage[layout=footnote,final]{fixme}
}

\FXRegisterAuthor{hu}{ahu}{HU} %
\FXRegisterAuthor{sm}{asm}{SM} %
\FXRegisterAuthor{st}{ast}{ST} %
\FXRegisterAuthor{ls}{als}{LS} %
\FXRegisterAuthor{sg}{asg}{SG} %
\FXRegisterAuthor{as}{aas}{AS} %

\renewcommand{\comp}{\cdot}

\newcommand{\seqcomp}{\mathop{;}}
\newcommand{\ret}{\mathsf{ret}}

\let\cedilla\c
\renewcommand{\c}{\colon}
\newcommand{\term}{1} %

\makeatletter
        \newcommand{\pushright}[1]{\ifmeasuring@#1\else\omit\hfill$\displaystyle#1$\fi\ignorespaces}
        \newcommand{\pushleft}[1]{\ifmeasuring@#1\else\omit$\displaystyle#1$\hfill\fi\ignorespaces}
\makeatother

\usepackage{seqsplit}
\usepackage{xstring}
\usepackage{xcolor}

\tikzstyle{shiftarr}=[
rounded corners,%
to path={--([#1]\tikztostart.center)
  -- ([#1]\tikztotarget.center) \tikztonodes
  -- (\tikztotarget)},
]

\tikzset{
  commutative diagrams/.cd,
  arrow style=tikz,
  diagrams={>=stealth},
  row sep=large,
  column sep = huge
}

\tikzset{cong/.style={draw=none,edge node={node [sloped, allow upside down, auto=false]{$\cong$}}},
  iso/.style={draw=none,every to/.append style={edge node={node [sloped, allow upside down, auto=false]{$\cong$}}}}}

\usetikzlibrary{decorations.pathmorphing}

\tikzcdset{scale cd/.style={every label/.append style={scale=#1},
    cells={nodes={scale=#1}}}}

\usepackage{stackengine}
\stackMath
\newcommand\tsup[2][2]{%
  \def\useanchorwidth{T}%
  \ifnum#1>1%
    \stackon[-1.05ex]{\tsup[\numexpr#1-1\relax]{#2}}{\scalebox{2}[1]{$\mathchar"307E$}\kern-.5pt}%
  \else%
    \stackon[-.9ex]{#2}{\scalebox{2}[1]{$\mathchar"307E$}\kern-.5pt}%
  \fi%
}

\usepackage{xspace}
\newcommand{\isos}{stateful SOS\xspace} %

\newcommand{\goes}[2]{\ensuremath{#1 \rightarrow #2}}
 \newcommand{\rets}[2]{\ensuremath{#1 \downarrow #2}}

\newcommand{\bN}{\mathbb{N}}
\newcommand{\bZ}{\mathbb{Z}}
\newcommand{\expr}{\mathsf{Ex}}
\newcommand{\Progs}{\mathsf{P}}

\DeclareFontFamily{U}{mathc}{}
\DeclareFontShape{U}{mathc}{m}{it}{<->s*[1.03] mathc10}{}
\DeclareMathAlphabet{\morph}{U}{mathc}{m}{it}

\newcommand{\Tr}{\mathsf{R}}
\newcommand{\Co}{\mathsf{W}}

\renewcommand{\O}{\mathcal{O}}
\newcommand{\Trs}{{\textsf{r}}}
\newcommand{\Cos}{{\textsf{w}}}

\newcommand{\while}{\ensuremath{\mathsf{Imp}}\xspace}
\newcommand{\whiletwo}{\ensuremath{\mathsf{Imp}^{\mathsf{2}}}\xspace}
\newcommand{\reflang}{\ensuremath{\mathsf{Ref}^{\mathsf{2}}}\xspace}

\newcommand{\store}{{\mathcal{S}}}
\newcommand{\impcat}{\Set^2}  %
\newcommand{\twoc}{2} %

\newcommand{\Rel}{\mathbf{Rel}}

\newcommand{\Sig}{\Theta}
\newcommand{\Sigs}{\Theta^\star}
\newcommand{\mS}{\mu\Theta}

\newcommand{\run}[2]{[#2]_{#1}}

\newcommand{\ctx}{\mathrm{ctx}}
\newcommand{\res}{\mathsf{res}}
\renewcommand{\trc}{\mathsf{trc}}
\newcommand{\scn}{\mathsf{cst}}

\newcommand{\trm}{\mathsf{ter}}

\newcommand{\act}{\mathsf{act}}
\newcommand{\pas}{\mathsf{pas}}

\newcommand{\Var}{\mathcal{V}}

\newcommand{\ass}{\mathrel{\coloneqq}}

\def\monoto{\rightarrowtail}

\usepackage{etoolbox} %

\theoremstyle{plain}

\theoremstyle{definition}
\newtheorem{defn}[theorem]{Definition} %
\newtheorem{rem}[theorem]{Remark} %

\theoremstyle{definition}
\newtheorem{notation}[theorem]{Notation}

\usepackage{scalerel,stackengine,graphicx}

\newcommand{\mybar}[3]{%
  \mathrlap{\hspace{#2}\overline{\scalebox{#1}[1]{\phantom{\ensuremath{#3}}}}}\ensuremath{#3}
}

\newcommand{\barF}{\mybar{0.6}{1.65pt}{F}}
\newcommand{\barf}{\ol{\f}}
\newcommand{\barB}{\mybar{0.6}{1.65pt}{B}}
\newcommand{\barD}{\mybar{0.6}{1.65pt}{D}}

\newcommand{\barPow}{\mybar{0.7}{1.5pt}{\Pow}}
\newcommand{\barSigmas}{\mybar{0.9}{0pt}{\Sigmas}}

\newcommand{\fset}{{\mathbb{F}}}

\setpremisesend{2pt}
\begin{document}\allowdisplaybreaks
\title{Bialgebraic Reasoning on Stateful Languages}

\author{Sergey Goncharov}
\orcid{0000-0001-6924-8766}             %
\affiliation{
  \department{School of Computer Science}              %
  \institution{University of Birmingham} %
  \city{Birmingham}
  \country{UK}                    %
}
\email{s.goncharov@bham.ac.uk}          %

\author{Stefan Milius}
\orcid{0000-0002-2021-1644}             %
\affiliation{
  \institution{Friedrich-Alexander-Universität Erlangen-Nürnberg}            %
  \country{Germany}                    %
}
\email{stefan.milius@fau.de}          %

\author{Lutz Schröder}
\orcid{0000-0002-3146-5906}             %
\affiliation{
  \institution{Friedrich-Alexander-Universität Erlangen-Nürnberg}            %
  \country{Germany}                    %
}
\email{lutz.schroeder@fau.de}          %

\author{Stelios Tsampas}
\orcid{0000-0001-8981-2328}             %
\affiliation{
  \institution{University of Southern Denmark (SDU)}            %
  \city{Odense}
  \country{Denmark}                    %
}
\email{stelios@imada.sdu.dk}          %

\author{Henning Urbat}
\orcid{0000-0002-3265-7168}             %
\affiliation{
  \institution{Friedrich-Alexander-Universität Erlangen-Nürnberg}            %
  \city{Erlangen}
  \country{Germany}                    %
}
\email{henning.urbat@fau.de}          %

\begin{abstract}
  Reasoning about program equivalence in imperative languages is
  notoriously challenging, as the presence of states (in the form of variable stores)
  fundamentally increases the observational power of program
  terms. The key desideratum for any notion of equivalence is
  \emph{compositionality}, guaranteeing that subprograms can be
  safely replaced by equivalent subprograms regardless of the
  context. To facilitate compositionality proofs and avoid boilerplate
  work, one would hope to employ the abstract bialgebraic methods
  provided by Turi and Plotkin's powerful theory of \emph{mathematical
  operational semantics} (a.k.a.~\emph{abstract GSOS}) or its recent extension by
  Goncharov et al.\ to higher-order languages. However, multiple
  attempts to apply abstract GSOS to stateful languages have
  thus failed. We propose a novel approach to the operational
  semantics of stateful languages based on the formal distinction
  between \emph{readers} (terms that expect an initial input
  store before being executed), and \emph{writers} (running terms that have already
  been provided with a store). In contrast to earlier work, this
  style of semantics is fully compatible with abstract GSOS, and we
  can thus leverage the existing theory to obtain coinductive reasoning
  techniques. We demonstrate that our approach generates
  non-trivial compositionality results for stateful languages with
  first-order and higher-order store and that it flexibly applies to
  program equivalences at different levels of granularity, such as
  trace, cost, and natural equivalence.
\end{abstract}

\begin{CCSXML}
  <ccs2012>
  <concept>
  <concept_id>10003752.10010124.10010131.10010137</concept_id>
  <concept_desc>Theory of computation~Categorical semantics</concept_desc>
  <concept_significance>500</concept_significance>
  </concept>
  <concept>
  <concept_id>10003752.10010124.10010131.10010134</concept_id>
  <concept_desc>Theory of computation~Operational semantics</concept_desc>
  <concept_significance>500</concept_significance>
  </concept>
  </ccs2012>
\end{CCSXML}

\ccsdesc[500]{Theory of computation~Categorical semantics}
\ccsdesc[500]{Theory of computation~Operational semantics}

\keywords{Abstract GSOS, Imperative languages, Stateful languages}

\maketitle %

\section{Introduction}\label{sec:intro}

Reasoning about program equivalence in stateful languages is a
long-standing topic of interest in formal
semantics with a sizeable body of work dedicated to it, e.g.~\cite{TennentGhica00,Ahmed04,BirkedalReusEtAl11,TuronThamsborgEtAl13,SterlingGratzerEtAl22,JungKrebbersEtAl16,AhmanFournetEtAl17,BentonHur09,AhmedBentonEtAl08,MasonTalcott91,KoutavasWand06,AhmedBentonEtAl10,OHearnReynoldsEtAl01,Reynolds02,Vafeiadis11,Harper94,Levy01,PittsStark98,YoshidaEA08}.
The level of attention the problem has received is due on the one hand
to the fact that almost every popular programming language has some
features involving mutable state, and on the other hand to the
technical challenges posed by the design of formal reasoning methods
for stateful languages, in particular in presence of higher-order
store (e.g.~\cite{AbelAllaisEtAl19,TimanyStefanescoEtAl18}).
Efficient reasoning about program equivalence is largely a question of
\emph{compositionality}, which allows exchanging subprograms of a
given program for equivalent ones without affecting the overall meaning
of the whole program. The well-known framework of \emph{mathematical
  operational semantics} by~\citet{TuriPlotkin97}, also known as
\emph{abstract GSOS}, derives a generic setup for the operational semantics of programming and process languages from a form of
categorical distributive law of syntax over behaviour (dubbed
a \emph{GSOS law}). It produces in particular a syntax-based
\emph{operational} model and a more abstract \emph{denotational}
model, and guarantees compositionality of the syntax with respect to the
denotational semantics, thereby obviating the need for laborious
ad-hoc compositionality proofs. Consequently, one would hope for abstract GSOS to provide a general and principled approach to establishing compositionality results for stateful languages.

However, there are known problems with applying abstract GSOS to a
stateful setting: A naive modelling of stateful languages in the framework
yields a denotational semantics in a final coalgebra for a functor
involving the state monad. The character of such a semantics is that
of a \emph{resumption semantics}, which assumes arbitrary interference
by the environment between program steps and as such is too
fine-grained for most purposes (cf.\ \Cref{sec:challenge}). There have been two major attempts to
address this issue.
\citet{Abou-SalehPattinson11,Abou-SalehPattinson13} consider abstract
GSOS laws over Kleisli categories for effect monads (such as the state
monad). While the format of GSOS law is the same as in the original
framework, the denotational semantics and the associated
compositionality result need to be completely reworked for the
purpose, and need complex technical assumptions. The coarsest form of
semantics covered in their framework is \emph{cost
  semantics}, which counts steps until termination. Contrastingly, the
framework of \emph{stateful SOS} introduced by
\citet{GoncharovMiliusEtAl22} does also cover \emph{trace semantics}, which considers the sequence of states computed by a program, and most notably \emph{termination semantics} (also known as~\emph{natural semantics}), the standard
end-to-end net execution semantics that only records the eventual
result of a program. However, stateful SOS deviates from abstract
GSOS even in the underlying abstraction of the rule format, replacing
GSOS laws with \emph{stateful SOS laws}, thus again requiring a
dedicated redevelopment of the framework to obtain compositionality results. Moreover, the compositionality proofs in \emph{op.~cit.} are not inherently categorical and
are based on syntactic reasoning at the level of rule formats.

In the present work, we propose a refined perspective on stateful languages %
that fully reconciles stateful operational semantics with abstract GSOS. The key novelty of our approach lies in an explicit division of program terms into \emph{readers} (programs that expect an initial store to start their execution) and \emph{writers} (programs that have already been given a store). For example, a standard
imperative while-language with variable store can be extended to a reader-writer language where every sequential composition $p\seqcomp q$ is regarded as a reader term and comes accompanied with a writer term $[p]_s\seqcomp q$ for every store $s$, with $[p]_s$ representing the program $p$ executed on $s$.

The reader-writer separation, along with the explicit representation of states in the writer syntax, enables the instantiation of abstract GSOS theory: We demonstrate that reader-writer languages can be modelled in abstract GSOS over two-sorted sets,
obtaining various flavours of stateful semantics (resumption, trace, cost, and termination semantics) as
instances. In fact, our approach fits smoothly into a
recent generalization of abstract GSOS due to~\citet{UrbatTsampasEtAl23b} that allows coverage of weak
equivalences such as weak bisimilarity
using a laxification of the core concept of
bialgebra by~\citet{BonchiPetrisanEtAl15}. It
subsumes, inter alia, the previous compositionality results for stateful SOS~\cite{GoncharovMiliusEtAl22} in the sense that
every language specification in stateful SOS can be transformed into
a two-sorted GSOS law in a semantics-preserving way.

 We use a standard
while-language as our running example.
In a further, more advanced case study, we adapt the idea of reader-writer semantics to the
recently introduced framework of \emph{higher-order mathematical
  operational
  semantics}~\cite{GoncharovMiliusEtAl23,UrbatTsampasEtAl23b}, which
extends abstract GSOS to higher-order languages. We then apply this abstract
framework to a first-order language with higher-order store, similar
to languages considered by~\citet{ReusStreicher05} and by
\citet{Pierce02}, obtaining compositionality of higher-order termination semantics.

\section{Bialgebraic Modelling of Stateful Languages}\label{sec:imperative-naive}

In this preliminary section we review the bialgebraic methods that the abstract GSOS framework~\cite{TuriPlotkin97} provides for reasoning about compositionality properties of programming languages. Moreover, we illustrate that stateful languages do not fit well into the framework when modelled naively, which motivates the refined approach presented subsequently in \Cref{sec:imperative}. As our running example we consider a simple imperative language, called $\while$~\cite{Winskel93}, with integer variables, assignments, sequential composition, and while loops.

\subsection{The language \while}\label{sec:while}
We fix a countably infinite set  $\mathcal{A}$ of \emph{(program)
    variables}, and a set $\expr$ of arithmetic expressions that are formed using
  standard operations such as $+,-,*$, constants $n\in \bZ$, and
  variables. A \emph{variable store} (or \emph{state}) is a map
  $s\colon \mathcal{A}\to\bZ$ assigning integer values to
  variables; we write $s[x\ass  n]$ for the store that maps $x$ to $n$ and otherwise equals $s$. We denote by $\store$ the set of
  states, and by ${\oname{eev}\c \expr\times \store \to \bZ}$ the \emph{expression evaluation}
  map; for example, $\oname{eev}(x*(y+3)-y,s) = s(x)*(s(y)+3)-s(y)$.

The set $\Progs$ of program terms of \while is given by the grammar
\[    \Progs\owns p, q \Coloneqq\;\mathsf{skip}\mid x\ass e \mid \mathsf{while}\;e\;p %
      \mid p \seqcomp q
\]
where $x\in \mathcal{A}$ and $e\in \expr$. The (small-step) operational semantics of \while are given by the inductive rules in \Cref{fig:rules-while}. The rules specify transitions of the form
\[ p,s\to p',s'  \qqand p,s\downarrow s'\qquad (p,p'\in\Progs,\, s,s'\in \S)\]
stating that when the program $p$ is executed on store $s$, it transforms $s$ into $s'$ and then behaves like~$p'$ (first case) or terminates (second case).
\begin{figure}
  \begin{gather*}
    \inference{p,s\to p',s'}{p \seqcomp q,s
      \to p' \seqcomp q,s'}
    \qquad
    \inference{p,s\downarrow s'}{p \seqcomp q, s
      \to q,s'}
    \qquad
    \inference{}{\mathsf{skip},s \downarrow s}
    \\[1ex]
    \inference{\oname{eev}(e,s) =
      n}{x \ass  e,s \downarrow s[x \ass  n]}%
    \qquad
    \inference{\oname{eev}(e,s) = 0}{\mathsf{while}~e~p,s \downarrow s}
    \qquad
    \inference{\oname{eev}(e,s) \not = 0}
    {\mathsf{while}~e~p,s \to
      p\seqcomp \mathsf{while}~e~p,s
    }
  \end{gather*}
  \caption{Small-step operational semantics of \while.}
  \label{fig:rules-while}
\end{figure}
These transitions are deterministic, hence induce a map
\begin{equation}\label{eq:while-opmodel} \gamma\colon \Progs \to (\Progs\times \store+\store)^\store
\qquad\text{given by}\qquad \gamma(p)(s)=
\begin{cases}
(p',s') & \text{if } p,s\to p',s',\\
s' & \text{if } p,s\downarrow s'.
\end{cases}
 \end{equation}
There are several natural notions of program behaviour and program equivalence
that can be associated to the transition system~\eqref{eq:while-opmodel}. Specifically, we consider the following four notions, presented in decreasing order w.r.t.~the power of an observer (environment) that executes programs:

\begin{enumerate}
\item \emph{Resumption semantics} corresponds to a low-level observer that has unrestricted access to the current program state and can both {read} and {modify} it at any stage of the computation. Under this semantics, two programs are equivalent if they are indistinguishable regardless of how the observer interferes. This idea is formally captured by \emph{resumption bisimilarity}:
\begin{defn}\label{def:res-bis}
A \emph{resumption bisimulation} is a relation $R\seq \Progs\times \Progs$ such that for $R(p,q)$ and $s\in \store$,
\begin{enumerate}\item if $p,s\to p',s'$ then $\exists q'.\, q,s\to q',s' \wedge R(p',q')$;
\item if $p,s\downarrow s'$ then $q,s\downarrow s'$;
\item if $q,s\to q',s'$ then $\exists p'.\, p,s\to p',s' \wedge R(p',q')$;
\item if $q,s\downarrow s'$ then $p,s\downarrow s'$.
\end{enumerate}
\emph{Resumption bisimilarity $\sim_\res$} is the greatest resumption bisimulation (given by the union of all resumption bisimulations). It forms an equivalence relation, with the corresponding quotient map \[\res\colon \Progs\to \Progs/{\sim_\res}.\]
\end{defn}
Resumption bisimilarity yields a very fine-grained notion of program equivalence.
For instance,
\begin{equation}\label{eq:pq} p=(x\ass 1 \seqcomp x\ass 2) \qqand q= (x\ass 1\seqcomp x\ass x+1)\end{equation}
 are not bisimilar: if the observer sets the program state to any state $s$ with $s(x)\neq 1$ after the first step, then after the second step the states become different, hence distinguishable.

 \item \emph{Trace semantics} corresponds to an observer that can read but not modify intermediate states. Formally, let $\store^{+}$ and $\store^\omega$ denote the sets of non-empty finite sequences and infinite sequences of elements of $\store$, respectively, and put $\store^\infty = \store^{+}\cup \store^\omega$. The \emph{trace map}
$\trc\colon \Progs\to (\store^\infty)^\S$
associates to $p\in \Progs$ and $s\in \store$ the sequence of states produced by running $p$ on the initial state $s$:
\[ \trc(p)(s)=
\begin{cases}
(s_1, \ldots, s_n,s')\in \store^{+} & \text{if}\;\;\exists p_i\; (1\leq i\leq n).\; p,s\to p_1,s_1 \to \cdots \to p_{n},s_{n}\downarrow s', \\
(s_1,s_2,s_3,\cdots)\in \store^\omega  & \text{if}\;\; \exists p_i\; (i\geq 1).\;\;\;\;\;\;\; p,s\to p_1,s_1 \to p_2,s_2\to p_3,s_3\to \cdots.
\end{cases}
\]
\item \emph{Cost semantics} corresponds to an observer which can observe the fact that a computation step has happened (hence count the total number of steps) but can only access the program state after termination. This is captured by the following map $\scn\colon \Progs\to (\bN\times \S + 1)^\S$ , where $1=\{*\}$:
\[
\scn(p)(s) = \begin{cases}
(n,s') & \text{if } \trc(p)(s) = (s_1,\ldots, s_n,s')\in \S^{+},\\
\ast & \text{if } \trc(p)(s) \in \S^\omega.
\end{cases}
\]
\item \emph{Termination semantics}, a.k.a.~\emph{natural semantics}, corresponds to an observer that can only view the terminating state (if any) of a program, while having no access to the actual computation. This is the standard form of semantics for imperative languages~\cite{Winskel93}, and is captured by
$\trm\colon \Progs\to (\S + 1)^\S$:
\[
\trm(p)(s) = \begin{cases}
s' & \text{if } \trc(p)(s) = (s_1,\ldots, s_n,s')\in \S^{+},\\
\ast & \text{if } \trc(p)(s) \in \S^\omega.
\end{cases}
\]
\end{enumerate}
Two programs $p$ and $q$ are \emph{trace equivalent} if $\trc(p)=\trc(q)$, \emph{cost equivalent} if $\scn(p)=\scn(q)$ and \emph{termination equivalent} if $\trm(p)=\trm(q)$. Note that
\[ \text{resumption bisimilarity} \To \text{trace equivalence} \To \text{cost equivalence} \To \text{termination equivalence}. \]
For example, $p$ and $q$ of \eqref{eq:pq} are trace equivalent, hence also cost and termination equivalent.

\begin{remark}
The definitions of $\res$, $\trc$, $\scn$ and $\trm$ and the ensuing notions of semantic equivalence generalize to arbitrary transition systems of type
$\gamma\colon X\to (X\times \S + \S)^\S$ where, as above, we put
\[p,s\to p',s' \quad\text{if}\quad \gamma(p)(s)=(p',s')\qqand p,s\downarrow s' \quad\text{if}\quad \gamma(p)(s)=s'.\]
\end{remark}
The four forms of semantics are \emph{compositional}, that is, they are respected by all \while-constructors:

\begin{theorem}[Compositionality of \while]\label{thm:comp-while}
For all $p,p',q,q'\in \Progs$, $e\in \mathsf{Ex}$ and $\sem{-}\in \{ \res,\trc,\scn,\trm \}$,
\[
\sem{p} = \sem{p'} \wedge \sem{q} = \sem{q'} \implies \sem{p \seqcomp q}=\sem{p' \seqcomp q'}\qand \sem{p} = \sem{p'} \implies \sem{\mathsf{while}~e~p} = \sem{\mathsf{while}~e~p'}.
\]
\end{theorem}
The importance of this theorem is that it enables modular, inductive reasoning about complex programs: compositionality guarantees that the observable behaviour of a program with respect to either of the different semantics is fully determined by the behaviour of its subterms. Standardly, compositionality results such as \Cref{thm:comp-while} are proved from scratch, using inductive or coinductive methods tailored to the features of the language at hand. This approach has several drawbacks: compositionality proofs tend to be tedious and error-prone, while at the same time containing many largely generic and repetitive parts. Additionally, every extension or modification of the language requires a meticulous adaptation of the proof. Our goal is to derive compositionality of stateful languages in a more principled manner by employing categorical techniques.

\subsection{Categorical Background}
\label{sec:catback}
The bialgebraic approach to compositional operational semantics, outlined in detail in \Cref{sec:abstract-gsos,sec:generalcomp}, rests on three fundamental categorical abstractions:
\begin{enumerate}
\item Model the set of programs of a language as an (initial) algebra for a suitable syntax functor.
\item Model the operational model of a language as a coalgebra for a suitable behaviour functor.
\item Model program equivalences via coalgebraic (bi)simulations, parametric in
  a relation lifting.
\end{enumerate}

\begin{remark}
  Relation lifting~\cite{KurzVelebil16} can be thought of as a systematic approach to
  \emph{relational reasoning} over algebraic structures. The precise notion is
  given ahead in the text.
\end{remark}

In the following we review the necessary terminology from category theory.
 Readers should be familiar with basic notions such as functors, natural transformations, (co)limits, and monads~\cite{Mac-Lane78}.

\paragraph*{Notation} Given
 objects $X_1, X_2$ in a category~$\C$, we write $X_1\times X_2$ for their product, and $\brks{f_1,f_2}\colon Y \to X_1 \times X_2$ for the pairing of
morphisms $f_i\colon Y \to X_i$, $i = 1,2$. Dually, $X_1+X_2$ denotes the coproduct, $[g_1,g_2]\c X_1+X_2\to Y$ the copairing of $g_i\colon X_i\to Y$, $i=1,2$, and we put $\nabla=[\id_X,\id_X]\colon X+X\to X$.

\paragraph*{Relations} A \emph{relation} on an object $X$ of~$\C$ is a subobject (represented by a monomorphism) $R\monoto X\times X$. We write $\outl_R,\outr_R\colon R \to X$ for the left and right projections, and usually drop the subscript.  A \emph{morphism} from $R\monoto X\times X$ to a relation $S\monoto Y\times Y$ is a morphism $f\colon X\to Y$ of~$\C$ such that there exists a (necessarily unique) dashed morphism making the first diagram in \eqref{eq:rel-diags} below commute. Relations $R,S\monoto X\times X$ are ordered by $R\leq S$ iff $\id_X$ is a morphism from $R$ of~$S$. We let $\Rel(\C)$ denote the category of relations and their morphisms.
\begin{equation}\label{eq:rel-diags}
\begin{tikzcd}
R \ar[dashed]{r} \ar[tail]{d} & S \ar[tail]{d} \\
X\times X \ar{r}{f\times f} & Y\times Y
\end{tikzcd} \qquad
\begin{tikzcd}
\Rel(\D) \ar{r}{\barF } \ar{d} & \Rel(\C) \ar{d} \\
\D \ar{r}{F} & \C
\end{tikzcd}
\end{equation}

In our applications we consider relations in the product category $\Set^T$ for a
set $T$. Its objects are $T$-sorted sets $X=(X_t)_{t\in T}$, and a morphism
$f\colon X\to Y$ is a $T$-sorted map $(f_t\colon X_t\to Y_t)_{t\in T}$.

\begin{example}
  A relation $R\monoto X\times X$ in $\Set^T$ is given by a family $(R_t\seq X_t\times X_t)_{t\in T}$.
  We write $\Delta_X\monoto X\times X$ for the \emph{identity relation}, i.e.
  $(\Delta_X)_t = \{ (x,x) \mid x\in X_t\}$, and $R\bullet S$ for the
  \emph{composite} of two relations $R,S\monoto X\times X$, i.e.\ $(R\bullet S)_t
  = R_t\bullet S_t$ where the right-hand $\bullet$ denotes the usual left-to-right composition
  of relations in $\Set$. The \emph{product} of $R\monoto X\times X$ and $S\monoto
  Y\times Y$ is the relation $R\times S \monoto (X\times Y)\times (X\times Y)$
  where $(R\times S)_t((x,y),(x',y'))$ iff $R_t(x,x')$ and $S_t(y,y')$. The
  \emph{power} $R^B\monoto X^B\times X^B$ for a set $B$ is defined analogously.
  The \emph{coproduct} $R+S\monoto (X+Y)\times (X+Y)$ is the disjoint union of $R$
  and $S$.
\end{example}

\begin{example}
  For the case of $T = 1$, where $1$ is the singleton set, thus $\C = \Set^{1}
  \cong \Set$, one recovers the standard notion of a relation on sets. For
  the case of $T = \twoc$, where $\twoc$ is the two-element set, relations on
  $\Set^{\twoc}$ are \emph{pairs} of set-theoretic relations etc.
\end{example}

A \emph{relation lifting} of a functor $F\colon \D\to\C$ is a functor $\barF $ making the second diagram in \eqref{eq:rel-diags} commute, where the vertical arrows are the respective forgetful functors given by $(R\monoto X\times X)\mapsto X$.
Every endofunctor $F$ on $\Set^T$ has a \emph{canonical} relation lifting, which sends a relation $R\monoto X\times X$ to the relation $\barF R\monoto FX\times FX$ obtained via the (surjective, injective)-factorization shown below:
\[
\begin{tikzcd}
FR \ar[shiftarr = {yshift=1.5em}]{rr}{\langle F\outl_R,\, F\outr_R\rangle} \ar[two heads]{r} & \barF R \ar[tail]{r} & FX\times FX
\end{tikzcd}
\]

\begin{example}\label{ex:can-lift-endofunctor}
\begin{enumerate}
\item Recall that a \emph{$T$-sorted algebraic signature}, for a fixed set $T$ of \emph{sorts}, consists of a set~$\Sigma$
of \emph{operation symbols} and a map $\ar\colon \Sigma\to T^{\star}\times T$
associating to every $\f\in \Sigma$ its \emph{arity}. We write $\f\colon t_1\times\cdots\times t_n\to t$ if $\ar(\f)=(t_1,\ldots,t_n,t)$, and $\f\colon t$ if $n=0$ (then $\f$ is called a \emph{constant}). In the case of a single-sorted signature ($T=1$) the arity of $\f$ is determined by the number $n$ of its arguments. Every
signature~$\Sigma$ induces an endofunctor on the category $\Set^T$, denoted by
the same letter $\Sigma$ and defined by $(\Sigma X)_t = \coprod_{\f\colon
  t_1\cdots t_n\to t} \prod_{i=1}^n X_{t_i}$ for $X\in \Set^T$ and $t\in T$.
Endofunctors of this form are called \emph{polynomial}. The canonical relation
lifting of $\Sigma$ maps a relation $R\monoto X\times X$ to the relation
$\ol{\Sigma}R\monoto \Sigma X\times \Sigma X$ where $(\ol{\Sigma} R)_t$ contains
all pairs $(\f(x_1,\ldots,x_n),\f(y_1,\ldots,y_n))$ such that $\f\colon
t_1\times \cdots \times t_n \to t$ is in $\Sigma$ and $R_{t_i}(x_i,y_i)$ for $i=1,\ldots,n$.
\item The canonical relation lifting of the power set functor $\Pow\colon \Set\to \Set$ maps a relation $R\seq X\times X$ to the \emph{Egli-Milner relation} $\barPow R\seq \Pow X\times \Pow X$ where $\barPow R(A,B)$ iff
\[ \text{(i) } \forall a\in A.\,\exists b\in B.\,R(a,b)\;\wedge\; \text{(ii) } \forall b\in B.\,\exists a\in A.\, R(a,b).  \]
Requiring only (i) yields the non-canonical lifting $\vec{\Pow}$. Note that the relation $\vec{\Pow}R$ is \emph{up-closed}, that is, $\vec{\Pow}R(A,B)$ and $B\seq B'$ implies  $\vec{\Pow}R(A,B')$.
\end{enumerate}
\end{example}
The up-closure of $\vec{\Pow}R$ turns out to be critical for the general compositionality result presented in \Cref{sec:generalcomp}. In abstract terms, this property is captured by imposing a suitable order structure:

\begin{definition}[\citet{UrbatTsampasEtAl23b}]
\begin{enumerate}
\item A functor $B\colon \D\to\C$ is \emph{ordered} if each hom-set $\C(Z,BX)$ ($Z\in \C, X\in \D$) is equipped with a preorder (a reflexive transitive relation) $\preceq$ such that
\[ \forall q,q'\colon Z\to BX,\,p\colon Z'\to Z.\qquad  q\preceq q' \implies q\comp p\preceq q'\comp p.\]
\item A relation lifting $\barB $ of an ordered functor $B\colon \D\to \C$ is \emph{up-closed} if for every relation $R\monoto X\times X$ in $\D$, every span $BX\xleftarrow{f} Z
  \xrightarrow{g} BX$ in $\C$ and every morphism $Z\to \barB R$ in $\C$ such that the
  left-hand triangle in the first diagram below commutes, and the right-hand triangle
  commutes laxly as indicated, there exists a morphism $Z\to \barB R$ in $\C$ such that the second
  diagram commutes.
\[
  \begin{tikzcd}
    & \ar[bend right=2em]{dl}[swap]{f} \ar[bend left=2em]{dr}{g} \ar{d} Z & {~}\\
    BX & \ar[phantom]{ur}[description, pos=.35, xshift=-10]{\dleq{45}} \ar{l}[swap]{\outl} \barB R \ar{r}{\outr} & BX
  \end{tikzcd}\qquad
  \begin{tikzcd}
    & \ar[bend right=2em]{dl}[swap]{f} \ar[bend left=2em]{dr}{g} \ar[dashed]{d} Z & {~}\\
    BX & \ar{l}[swap]{\outl} \barB R \ar{r}{\outr} & BX
  \end{tikzcd}
\]
\end{enumerate}
\end{definition}

For instance, the lifting $\vec{\Pow}$ is up-closed if $\Pow$ is ordered by $q\preceq q'$ iff $\forall z\in Z.\, q(z)\seq q'(z)$.

\paragraph*{Algebras}
Given an endofunctor $\Sigma$ on a category $\C$,
a \emph{$\Sigma$-algebra} is a pair $(A,a)$ of an object~$A$
 and a morphism $a\colon \Sigma A\to A$. A \emph{morphism} from
$(A,a)$ to a $\Sigma$-algebra $(B,b)$ is a morphism $h\colon A\to B$
of~$\C$ with $h\circ a = b\circ \Sigma h$. Algebras for $\Sigma$ and
their morphisms form a category, and an \emph{initial}
$\Sigma$-algebra is simply an initial object in that category. We denote the initial algebra by $(\mu\Sigma,\ini)$ if it exists.

More generally, a \emph{free $\Sigma$-algebra} on an object $X$ of $\C$ is a
$\Sigma$-algebra $(\Sigma^{\star}X,\iota_X)$ together with a morphism
$\eta_X\c X\to \Sigma^{\star}X$ of~$\C$ such that for every algebra $(A,a)$
and every morphism $h\colon X\to A$ of $\C$, there exists a unique
$\Sigma$-algebra morphism $h^\star\colon (\Sigma^{\star}X,\iota_X)\to (A,a)$
such that $h=h^\star\circ \eta_X$. If $\C$ has an initial object $0$, then $\Sigmas 0=\mu\Sigma$. If free algebras
exist on every object, then their formation induces a monad
$\Sigma^{\star}\colon \C\to \C$, the \emph{free monad} generated by $\Sigma$~\cite{Barr70}. Every $\Sigma$-algebra $(A,a)$ induces an
Eilenberg-Moore algebra $\hat a \colon \Sigma^{\star} A \to A$, viz.\ $\hat a =\id_A^\star$ for the identity morphism $\id_A\c A\to A$.

Given a relation lifting $\ol{\Sigma} $ of $\Sigma$, a \emph{$\ol{\Sigma}$-congruence} on an algebra $(A,a)$ is a relation $R\monoto A\times A$ such that the algebra structure $a\colon \Sigma A\to A$ is a morphism $a\colon \bar \Sigma R\to R$ of $\Rel(\C)$.

\begin{example} The prototypical instance of the above categorical concepts are algebras for a signature. For every $T$-sorted signature $\Sigma$, an algebra for the associated polynomial functor $\Sigma$ is
precisely an algebra for the signature $\Sigma$ in the usual sense, that is, an $T$-sorted set $A=(A_t)_{t\in T}$ equipped with an operation $\f^A\colon \prod_{i=1}^n A_{t_i}\to A_t$ for every $\f\colon t_1\times\cdots \times t_n\to t$ in $\Sigma$. Morphisms of $\Sigma$-algebras are $T$-sorted
maps respecting the algebraic structure. Given a $T$-sorted set $X$ of
variables, the free algebra $\Sigmas X$ is the $\Sigma$-algebra of
$\Sigma$-terms with variables from~$X$; more precisely, $(\Sigmas X)_t$ is inductively defined by $X_t\seq (\Sigmas X)_t$ and $\f(u_1,\ldots,u_n)\in (\Sigmas X)_t$ for all ${\f\colon t_1\times\cdots \times t_n\to t}$ and $u_i\in (\Sigmas X)_{t_i}$.
The free
algebra on the empty set is the initial algebra~$\mu \Sigma$; it is
formed by all \emph{closed} $\Sigma$-terms. For every
$\Sigma$-algebra $(A,a)$, the corresponding Eilenberg-Moore algebra
$\hat{a}\colon \Sigmas A \to A$ is given by the map that evaluates terms in~$A$.

For the canonical relation lifting $\ol{\Sigma}$, a $\ol{\Sigma}$-congruence on a $\Sigma$-algebra $A$ is the usual concept from universal algebra, namely a relation $R\monoto A\times A$ compatible with all operations: for $\f\colon t_1\times\cdots \times t_n\to t$ and elements $x_i,y_i\in A_{t_i}$ such that $R_{t_i}(x_i,y_i)$ ($i=1,\ldots,n$), one has $R_t(\f^A(x_1,\ldots,x_n), \f^A(y_1,\ldots,y_n))$. Note that we do not require congruences to be equivalence relations.
\end{example}

\begin{example}[Program terms]\label{ex:while-syntax}
The set $\Progs$ of $\while$-terms forms the initial algebra for the polynomial functor $\Sigma\colon \Set\to \Set$ corresponding the single-sorted signature of \while:
\begin{equation}\label{eq:sig-while} \Sigma X
    = \underbrace{\term}_{\mathsf{skip}}
       +\, \underbrace{\mathcal{A} \times \expr}_{\ass }
       \,+\, \underbrace{\expr \times X}_{\mathsf{while}}
       \,+\, \underbrace{X \times X}_{-;-}.
\end{equation}
 Compositionality of \while (\Cref{thm:comp-while}) states that for $\sem{-}\in \{ \res,\trc,\scn,\trm \}$ the equivalence relation $p\sim p' \iff \sem{p} = \sem{p'}$ forms a congruence on the algebra $\Progs=\mu\Sigma$ of program terms.
\end{example}
\smallskip
\paragraph*{Coalgebras}~Dually to the notion of algebra, a \emph{coalgebra} for an
endofunctor $B$ on $\C$ is a pair $(C,c)$ of an object $C$ and a morphism $c\colon C\to BC$. A \emph{morphism} from
$(C,c)$ to a $B$-coalgebra $(D,d)$ is a morphism
$h\colon C\to D$ of $\C$ such that $Bh\circ c = d\circ h$.
Coalgebras for $B$ and their morphisms form a category, and a
\emph{final} $B$-coalgebra, denoted $(\nu B,\tau)$, is a final object in that category.

\begin{example}\label{ex:while-behaviour}
Coalgebras form an abstraction of transition systems. For instance, the operational model $\gamma\colon \Progs \to (\Progs\times \store + \store)^\store$ of $\while$ is a $D^\S$-coalgebra where $DX=X\times \S + \S$ on $\Set$.
\end{example}
The theory of coalgebras allows modelling notions of strong or weak (bi)simulation for transition systems at a convenient level of generality, using relation liftings of behaviour functors.

\begin{definition}[\citet{UrbatTsampasEtAl23b}]\label{def:coalg-sim}
 Given a relation lifting $\barB $ of $B$, a \emph{weakening}~\cite{UrbatTsampasEtAl23} of a $B$-coalgebra $(C,c)$ is a $B$-coalgebra $(C,\tilde{c})$ such that for every relation $R\monoto C\times C$, there exists a morphism $R\to \barB R$ in $\C$ making the first diagram below commute iff there exists a morphism $R\to \barB R$ in $\C$ making the second diagram commute. A relation $R$ satisfying these two equivalent properties is then called a \emph{$\barB$-simulation on $(C,c,\tilde{c})$}. The greatest $\barB$-simulation (if it exists) is called \emph{$\barB$-similarity}.
\begin{equation}\label{eq:weakening}
\begin{tikzcd}
R \ar[dashed]{r} \ar[tail]{d} & \barB R \ar[tail]{d} \\
C\times C \ar{r}{c\times \tilde{c}} & BC\times BC
\end{tikzcd}
\qquad\qquad
\begin{tikzcd}
R \ar[dashed]{r} \ar[tail]{d} & \barB R \ar[tail]{d} \\
C\times C \ar{r}{\tilde{c}\times \tilde{c}} & BC\times BC
\end{tikzcd}
\end{equation}
\end{definition}

\begin{rem} Thanks to being parametric in the lifting $\barB$ and the weakening $\tilde{c}$, the above notion of $\barB$-simulation is quite flexible. In fact, despite the terminology, it will in some cases amount to \emph{bi}simulations in the usual sense. There are two typical situations to which $\barB$-simulations instantiate:
\begin{enumerate}
\item For the trivial weakening $\tilde{c}=c$, the notion of $\barB$-simulation first appeared in the work of \citet{HermidaJacobs98} (under the name $\barB$-\emph{bi}simulation). For instance, for the power set functor $\Pow$ with its lifting $\vec{\Pow}$ of \Cref{ex:can-lift-endofunctor}, a $\vec{\Pow}$-simulation is the usual notion of strong simulation of graphs, while the `symmetric' lifting $\barPow$ yields strong bisimulations.
\item If $\tilde{c}$ is some form of reflexive transitive closure of $c$, then $\barB $-simulations amount to a notion of \emph{weak} simulation. For instance, let $c\colon C\to \Pow C$ and let $\tilde{c}\colon C\to \Pow C$  be given by $y\in \tilde{c}(x)$ iff there exist $n\geq 0$ and $x=x_0,\ldots,x_n=y\in C$ such that $x_{i+1}\in c(x_i)$ for $0\leq i<n$. Then a $\vec{\Pow}$-simulation on $(C,c,\tilde{c})$ corresponds to a weak simulation: given any two related states $x,y$, every strong transition from $x$ is matched by a weak transition from $y$. The fact that $\tilde{c}$ is a weakening expresses that this property is equivalent to matching weak transitions with weak transitions.
\end{enumerate}
\end{rem}

\begin{rem}\label{rem:bisim-vs-behaveq} For endofunctors $B$ on $\Set^T$, we note that:
\begin{enumerate}
\item The $\barB $-similarity relation is given by the (sortwise) union of all $\barB$-simulations on $(C,c,\tilde{c})$.
\item  If $\barB$ is the canonical relation lifting, the final coalgebra $\nu B$ exists, the functor $B$ preserves weak pullbacks, and $\tilde{c}=c$, then two states of $(C,c)$ are $\barB$-similar iff they are \emph{behaviourally equivalent}, i.e.\ merged by the unique coalgebra morphism from $(C,c)$ to $\nu B$~\cite[Thm.~4.2.4]{Jacobs16}.
\end{enumerate}
\end{rem}

\begin{example}\label{ex:while-behaviour-2}
The behaviour functor $D^\store = (X\times \store+\store)^\store$ for \while has a canonical lifting $\barD^\store$ that sends  $R\seq X\times X$ to $\barD^\store  R\seq (X\times \S+\S)^\S \times (X\times \S+\S)^\S$ where $\barD^\store (f,g)$ iff, for all $s\in \S$, either $f(s)=(x,s')$ and $g(s)=(x',s')$ where $R(x,x')$ and $s'\in \S$, or $f(s)=g(s)\in \S$. Then
$\barD^\store$-similarity on the operational model \eqref{eq:while-opmodel} coincides with resumption bisimilarity (\Cref{def:res-bis}). Since $D^\store$ preserves weak pullbacks, it also coincides with behavioural equivalence (\Cref{rem:bisim-vs-behaveq}).
\end{example}

\subsection{Abstract Operational Semantics}\label{sec:abstract-gsos}
The abstract GSOS framework \cite{TuriPlotkin97} yields an elegant categorical approach to operational semantics. It is parametric in two endofunctors $\Sigma, B\colon \C\to \C$ on a category $\C$ with products, where $\Sigma$ has an initial algebra $\mu\Sigma$ and generates a free monad $\Sigmas$. The functors $\Sigma$ and $B$ represent the \emph{syntax} and \emph{behaviour} of a programming language, with $\mu\Sigma$ thought of as the object of program terms. The operational semantics of a language are modelled by a \emph{GSOS law of $\Sigma$ over $B$}: a natural transformation
\begin{equation}\label{eq:rho} \rho_X\colon \Sigma(X\times BX)\to B\Sigmas X \quad (X\in \C).
\end{equation}
Informally, $\rho$ encodes the inductive operational rules of the language at hand: for every program constructor~$\f$, it specifies the one-step behaviour of programs $\f(p_1,\ldots,p_n)$, i.e.\ the $\Sigma$-terms they transition into next, depending on the one-step behaviours of the operands $p_1,\ldots,p_n$.

\begin{example}\label{ex:while-gsos-law}
The operational rules of \while (\Cref{fig:rules-while}) can be translated into a GSOS law $\rho$ of the signature functor $\Sigma$ given in \Cref{ex:while-syntax} over the behaviour functor $D^\S X=(X\times \S+\S)^\S$. The component $\rho_X\colon \Sigma(X\times (X\times \S+\S)^\S)\to (\Sigmas X\times \S+\S)^\S$ encodes, for instance, the rules for sequential composition into the following assignment for $p,q\in X$ and $f,g \in (X\times \S+\S)^\S$:
\[ \rho_X((p,f) \seqcomp\, (q,g)) \;=\; \lambda s.
\begin{cases}
((p'\seqcomp q),s') & \text{if } f(s)=(p',s')\in X\times \S,\\
(q,s') & \text{if } f(s)=s'\in \S.
\end{cases} \]

\end{example}

The universal property of the initial algebra $(\mu\Sigma,\iota)$ entails that there
exists a unique $B$-coalgebra structure $\gamma\c \mu\Sigma\to B(\mu\Sigma)$ such that the diagram below commutes: \begin{equation}\label{eq:gamma-def}
\begin{tikzcd}
\Sigma(\mu\Sigma) \ar{r}{\ini} \ar{d}[swap]{\Sigma\langle \id,\,\gamma\rangle} & \mu\Sigma \ar[dashed]{r}{\gamma} & B(\mu\Sigma) \\
\Sigma(\mu\Sigma\times B(\mu\Sigma)) \ar{rr}{\rho_{\mu\Sigma}} & & B\Sigmas (\mu\Sigma) \ar{u}[swap]{B\hat \ini}
\end{tikzcd}
\end{equation}
The coalgebra $(\mu\Sigma,\gamma)$ is called the \emph{operational model} of the GSOS law $\rho$. Informally, this is the transition system that runs programs according the operational rules represented by $\rho$. For instance, the operational model of the GSOS law for \while (\Cref{ex:while-gsos-law}) is precisely the transition system~\eqref{eq:while-opmodel} on the set $\Progs=\mu\Sigma$ of program terms induced by the rules of \while.

\subsection{Compositionality in Abstract GSOS}
\label{sec:generalcomp}
The key feature of abstract GSOS is that it allows for general compositionality results applying uniformly to languages modelled in the framework. The compositionality theorem presented below (\Cref{thm:congruence-abstract-gsos}) is substantially more powerful than the original one by \citet{TuriPlotkin97}: it not only applies to behavioural equivalence in the final coalgebra, but to similarity w.r.t.~a choice of relation lifting of the behaviour functor and a choice of weakening of the operational model. Let us first fix the required data. We restrict to the base category $\C=\Set^T$ and polynomial syntax functors, which simplifies some technical conditions and is sufficient to our purposes:

\begin{definition}
  \label{def:aos}
An \emph{abstract operational setting} (\emph{AOS}) $\O=(\Sigma,\ol{\Sigma},B,\barB ,\rho,\gamma,\tilde{\gamma})$
is given by the following data:
\begin{itemize}
\item a polynomial functor $\Sigma\colon \Set^T\to \Set^T$ with its {canonical} relation lifting $\ol{\Sigma}$;
\item an ordered functor $(B,\preceq)\colon \Set^T\to \Set^T$ with a (not necessarily canonical) relation lifting $\barB$;
\item a GSOS law $\rho$ of $\Sigma$ over $B$;
\item the operational model $(\mu\Sigma,\gamma)$ of $\rho$ with a weakening $\tilde{\gamma}$.
\end{itemize}
\end{definition}

The point of the definition of AOS is to collect (in a neat, abstract form) the
\emph{problem setting} of relational reasoning about some ``weak'' notion of program
equivalence on a (first-order) programming language. The polynomial functor
$\Sigma$ models the algebraic signature of the language and $B$ models the
behaviour, for instance $B = \mathcal{P}(L \times \Id)$ for a nondeterministic
LTS. The relation lifting $\ol{B}$ eventually determines the type of program
equivalence considered. The GSOS law $\rho$ corresponds to the collection of
operational rules and the weakening $\ol{\gamma}$ to the chosen notion of a
\emph{weak} transition system, e.g. a saturation $\To$.

Given an AOS we can study the congruence of $\barB $-similarity on $(\mu\Sigma,\gamma,\tilde{\gamma})$. It turns out that the congruence property boils down to three natural conditions.

The first one is that the lifting $\barB$ is up-closed and laxly preserves composition and identities:
\[\barB R\bullet \barB S\seq \barB(R\bullet S) \text{ for all $R,S\monoto X\times X$} \qqand
\Delta_{BX} \seq \barB\Delta_X \text{ for all $X$},
\]
where $\seq$ means sortwise inclusion. This ensures that $\barB$-similarity is a preorder~\cite[Lemma~VI.3]{UrbatTsampasEtAl23b}.

The second condition is \emph{liftability} of the law $\rho$:
\begin{definition}
The GSOS law $\rho$ of an AOS is \emph{liftable} if for every relation $R\monoto X\times X$ in $\Set^T$, the component $\rho_X$ is a $\Rel(\Set^T)$-morphism from $\ol{\Sigma}(R\times \barB R)$ to $\barB \,\ol{\Sigma}^{*}{R}$.
\end{definition}
Note that the free monad $\ol{\Sigma}^{*}$ is given by the canonical lifting $\ol{\Sigma^{*}}$ of the free monad $\Sigma^{*}$~\cite[Prop.~V.4]{UrbatTsampasEtAl23b}. Intuitively, liftability expresses in abstract terms that the operational rules encoded by $\rho$ are parametrically polymorphic, i.e.\ do not inspect the structure of their arguments.
 Hence, every relation between the arguments should be preserved by $\rho$. Liftability is closely related to dinaturality of $\rho$; in fact, when using canonical liftings, dinaturality implies liftability (\autoref{rem:congruence-abstract-gsos}).

The final condition ensures that weak transitions interact well with the operational rules:
\begin{definition}\label{def:lax-bialg}
The triple $(\mu\Sigma,\ini,\tilde{\gamma})$ is a \emph{lax $\rho$-bialgebra} if the diagram below commutes laxly:
\begin{equation}\label{eq:lax-bialgebra-fo}
\begin{tikzcd}
\Sigma(\mu\Sigma) \ar{r}{\ini} \ar{d}[swap]{\Sigma \langle \id, \tilde{\gamma}\rangle} & \mu\Sigma \ar{r}{\tilde{\gamma}} & B(\mu\Sigma) \\
\Sigma(\mu\Sigma\times B(\mu\Sigma)) \ar{rr}{\rho_{\mu\Sigma}} & ~ \ar[phantom]{u}[description, yshift=5]{\dgeq{-90}} & B\Sigmas(\mu\Sigma) \ar{u}[swap]{B\hat\ini}
\end{tikzcd}
\end{equation}
\end{definition}
The lax bialgebra condition expresses that the operational rules corresponding to $\rho$ remain sound in the operational model when strong transitions (given by $\gamma$) are replaced by weak transitions (given by $\tilde{\gamma}$). Lax bialgebras were introduced (in a slightly less general form than \Cref{def:lax-bialg}) by \citet{BonchiPetrisanEtAl15} and originally used to analyse sound up-to techniques for weak (bi)similarity.

The following general compositionality result for abstract GSOS is a special case of \cite[Cor.~VIII.7]{UrbatTsampasEtAl23b}. A higher-order version will appear in \Cref{sec:higher-order}.

\begin{theorem}[Compositionality]\label{thm:congruence-abstract-gsos}
Suppose that $\O=(\Sigma,\ol{\Sigma},B,\barB ,\rho,\gamma,\tilde{\gamma})$ is an AOS such that
\begin{enumerate}
\item\label{comp-cond1} $\barB$ is up-closed and laxly preserves composition and identities;
\item\label{comp-cond2} $\rho$ is liftable;
\item\label{comp-cond3} $(\mu\Sigma,\ini,\tilde{\gamma})$ is a lax
  $\rho$-bialgebra.
\end{enumerate}
Then the $\barB $-similarity relation on $(\mu\Sigma,\gamma,\tilde{\gamma})$ is a $\ol{\Sigma}$-congruence on the initial algebra $(\mu\Sigma,\ini)$.
\end{theorem}

\begin{remark}\label{rem:congruence-abstract-gsos}
In practice, some of the conditions of the theorem may come for free:
\begin{enumerate}
\item holds if $\barB$ is canonical, and $B$ is ordered by equality and preserves weak pullbacks~\cite[Prop.~C.9]{UrbatTsampasEtAl23}.
\item holds if $\barB$ is the canonical lifting~\cite[Constr.~D.5]{UrbatTsampasEtAl23}.
\item holds if $B$ is ordered by equality and $\tilde{\gamma}=\gamma$; in this case, the lax bialgebra condition \eqref{eq:lax-bialgebra-fo} is just \eqref{eq:gamma-def}.
\end{enumerate}
For functors $B$ preserving weak pullbacks and having a final coalgebra, \Cref{thm:congruence-abstract-gsos} thus specializes to the original congruence result by~\citet{TuriPlotkin97}: for every GSOS law, behavioural equivalence in the final coalgebra is a congruence on the operational model.
\end{remark}

The general compositionality theorem simplifies proofs of congruence results for programming languages by reducing their congruence properties to elementary and usually easily verifiable conditions on the lifting $\barB $ of the behaviour functor and the rules of the language given by $\rho$.

\subsection{The Fundamental Challenge}\label{sec:challenge}

Can we deduce the compositionality of \while (\Cref{thm:comp-while}) from the compositionality theorem for abstract GSOS? At this stage, only the case of resumption semantics -- the most fine-grained and arguably least useful of the four notions of program semantics in \Cref{sec:while} -- is immediate:

\begin{example}
We apply \Cref{thm:congruence-abstract-gsos} to the AOS $\O=(\Sigma,\ol{\Sigma},D^\store,\barD^\store ,\rho,\gamma,{\gamma})$ where $\Sigma$ is the signature of \while, $D^\store$  and $\barD^\store$ are as in \Cref{ex:while-behaviour-2} (with $D^\store$ ordered by equality), and $\rho$ is as in \Cref{ex:while-gsos-law}. Recall that $\barD^\S$-similarity is resumption bisimilarity. By \Cref{rem:congruence-abstract-gsos}, all conditions of \Cref{thm:congruence-abstract-gsos} are satisfied, whence resumption bisimilarity is a congruence on the algebra $\Progs=\mu\Sigma$ of \while-terms. In fact, since $B$ has a final coalgebra, this result already follows from Turi and Plotkin's original theory of congruence and does not require the full generality of \Cref{thm:congruence-abstract-gsos}.
\end{example}

To capture the more interesting parts of \Cref{thm:comp-while} in our abstract
setting, we would need to come up with liftings $\barD^\S$ and weakenings
$\tilde{\gamma}$ such that the ensuing notions of $\barD^S$-similarity coincide
with trace, cost, and termination equivalence, respectively. This
appears to be impossible under the present choice of behaviour functor
$D^\store X=(X\times\store +\store)^\S$; intuitively, the exponent $(-)^\S$
forces a resumption-like semantics. Our solution lies in switching to a
\emph{two-sorted} refinement of \while.

\section{Readers and Writers}
\label{sec:imperative}

Trace semantics in \while interprets a term $p$ as the
function mapping an input store $s$ to the (possibly infinite) trace of $(p,s)$.
This suggests that program-store pairs $(p,s)$ are essentially treated as the states of a transition system, i.e. $(p,s) \xrightarrow{s'} (p',s')$. On the other hand, the operational semantics treats terms as \emph{transformers} acting on an input store: for instance, in $p;q$, the subterm $q$ is thought of as an 'open' computation, which needs to be given an input $s$ to start. Afterwards, however, trace semantics treats the pair $(q,s)$ as a \emph{running} computation.

The reader-writer approach to stateful semantics turns this implicit divide
between ``transformers'' and
``started transformers/computations'' into explicit syntactic sorts, namely
\emph{readers} and \emph{writers} respectively.
Readers represent programs that need to be provided with an input state in
order to run. Contrastingly,
writers represent programs that are already running, producing a trace of states, and for
instance arise by passing an input state to a reader. We
illustrate this idea via a language called \whiletwo, which extends the language \while.

\subsection{The Language $\boldsymbol{\whiletwo}$}\label{sec:whiletwo}
We inherit the sets $\mathcal{A}$ (program variables), $\S$ (states) and $\expr$ (arithmetic expressions) from \while.
The program terms of $\whiletwo$ are split into a set~$\Tr$ of
readers $p,q,\dots$ and a set~$\Co$ of writers $c,\dots$ and are defined by the
 grammar
\begin{equation}
  \label{eq:transsyn}
  \Tr\owns p, q \Coloneqq \mathsf{skip}\mid x \ass  e \mid \mathsf{while}~e~p %
  \mid p \seqcomp q \qqand
  \Co\owns  c 		\Coloneqq \run{s}{p} \mid s.c \mid \ret_s \mid c \seqcomp q
\end{equation}
where $s \in \store$, $x \in \mathcal{A}$ and $e \in \expr$. Equivalently, $(\Tr,\Co)$ is the initial algebra for the polynomial endofunctor $\Sigma \c \impcat \to
\impcat$ corresponding to the two-sorted signature of $\whiletwo$, with sorts $\Trs$ and $\Cos$:
\[
  \Sigma_{\Trs} X
  = \underbrace{\term}_{\mathsf{skip}}
  +\, \underbrace{\mathcal{A} \times \expr}_{\ass }
  \,+\, \underbrace{\expr \times X_{\Trs}}_{\mathsf{while}}
  \,+\, \underbrace{X_{\Trs} \times X_{\Trs}}_{-;-},\quad
  \Sigma_{\Cos} X
  = \underbrace{X_{\Trs}\times\store}_{[-]_s} \,+\, \underbrace{\store \times X_{\Cos}}_{s.-}\,+\, \underbrace{\store}_{\ret_s}
  \,+\,
  \underbrace{X_{\Cos} \times X_{\Trs}}_{-;-}.
\]
Note that readers are precisely \while-terms ($\Progs = \Tr$). Additionally, there are four
types of writers: For every reader $p$ and state $s$, the
writer~$\run{s}{p}$ represents the encapsulated computation of~$p$ at input $s$. This bears some resemblance with the \emph{fine-grain call-by-value} paradigm~\cite{LevyPowerEtAl03} with its operator $[-]$ casting values as computations.
 The writer $s.c$ produces the state $s$ and then behaves like $c$. The writer $\ret_s$ terminates in the state $s$. Finally, the writer $c\seqcomp q$ represents a sequential composition that currently evaluates its first argument. One can think of  writers as programs with
`explicit' store, in analogy with explicit substitutions: the store
  is gradually moved towards the computation point, e.g.~one may
  derive $[p\seqcomp q]_s \to [p]_s \seqcomp q$ or $\ret_s \seqcomp q \xto{s} [q]_s,s$, where a labelled transition denotes
  `store output', signaling to the outer
  world that a computation step in \while has been completed.

Formally, the operational semantics of readers are specified by the rules in the top part of \Cref{fig:comp-rules}.
\begin{figure}
  \begin{gather*}
    \inference{}{p\seqcomp q,s
      \to \run{s}{p}\seqcomp q}
    \qquad
    \inference{}{\mathsf{skip},s \to \ret_s}
    \qquad
    \inference{\oname{eev}(e,s) =
      n}{x \ass  e,s \to \ret_{s[x \ass  n]}}%
    \\[1ex]
    \inference{\oname{eev}(e,s) = 0}{\mathsf{while}~e~p,s \to \ret_s}
    \qquad
    \inference{\oname{eev}(e,s) \not = 0}
    {\mathsf{while}~e~p,s \to
      s.\run{s}{p\seqcomp \mathsf{while}~e~p}
    }
  \end{gather*}

\vspace{1.5ex}
~\dotfill~
\vspace{1ex}
  \begin{gather*}
    \inference{p,s\to c}{\run{s}{p} \to c}
    \qquad
    \inference{}{\ret_s \downarrow s}
    \qquad
    \inference{}{s.c \xto{s} c}
    \qquad
    \inference{c \xto{s} d}{c\seqcomp q
      \xto{s} d\seqcomp q}
    \qquad
    \inference{c \to d}{c\seqcomp q
      \to d\seqcomp q}
    \qquad
    \inference{c \downarrow s'}{c\seqcomp q
      \xto{s'} \run{s'}{q}}
  \end{gather*}
  \caption{Reader semantics of \whiletwo (top) and writer semantics of \whiletwo (bottom).}
  \label{fig:comp-rules}
\end{figure}
They inductively determine transitions of the form
\begin{equation}
  \begin{aligned}
    & p,s \to c
    & & \text{where}
        \quad p \in \Tr,\, c \in \Co, s\in \S,
  \end{aligned}
\end{equation}
to be read as `on input state $s$,
  reader~$p$ continues as writer~$c$'. The reader semantics induces a map
\begin{equation}
  \label{eq:trsem}
  \gamma_0^{\Trs} \c\Tr  \to
  \Co^{\store}
  \qquad\text{given by}\qquad \gamma_0^\Trs(p)(s)= c \qquad \text{if} \qquad p,s\to c.
\end{equation}
The operational semantics of the writers are given in the bottom part of \Cref{fig:comp-rules}. These rules determine transitions of type
\[ c \to d,\qquad c\xto{s} d, \qquad c\downarrow s, \qquad \text{where}\qquad c,d\in \Co,\, s\in \store, \]
 to be read as `$c$ progresses to $d$', `$c$ progresses to $d$ and generates the state $s$', and `$c$ terminates in the state $s$', respectively. The writer semantics thus yields a function
\begin{equation}
  \label{eq:cosem}
    \gamma_0^{\Cos} \colon \Co \to \Co\times \store + \Co + \store
\quad\text{given by}\quad \gamma_0^{\Cos}(c) = (d,s) \, / \, d\, /\, s \quad\text{if}\quad c\xto{s} d \, / \, c\to d \,/\, c\downarrow s.
\end{equation}

\begin{notation}
 We denote objects of $\Set^2$ as a pairs $X=(X_\Trs,X_\Cos)$, morphisms as pairs $f=(f^\Trs,f^\Cos)$\footnote{We use superscripts for the sorts of morphisms to reserve subscripts for components of natural transformations.}, and we write $F_\Trs$ and $F_\Cos$ for the two components of a functor $F\colon \C\to \Set^2$.
\end{notation}

The two maps~\eqref{eq:trsem} and~\eqref{eq:cosem} combine into a coalgebra
\begin{equation}\label{eq:while2opmodel}
\gamma_0=(\gamma_0^\Trs,\gamma_0^\Cos)\colon (\Tr,\Co) \to B_0(\Tr,\Co)
\end{equation}
for the endofunctor $B_0$ on $\Set^2$ defined by
\begin{equation}\label{eq:beh-whiletwo} (B_0)_\Trs X = X_\Cos^\store \qqand (B_0)_\Cos X = X_\Cos\times \store + X_\Cos + \store. \end{equation}
The coalgebra \eqref{eq:while2opmodel} is the operational model of the GSOS law $\rho_0$ of $\Sigma$ over $B_0$ that encodes the operational rules of \Cref{fig:comp-rules}. The component $\rho_{0,X}=(\rho_0)_X$ at $X\in \Set^2$ is given by the maps
\[
\rho^\Trs_{0,X} \colon \Sigma_\Trs(X\times B_0X)\to (\Sigmas_\Cos X)^\store \qqand
\rho^\Cos_{0,X} \colon \Sigma_\Cos(X\times B_0X)\to \Sigmas_\Cos X\times \S + \Sigmas_\Cos X +\store
\]
where, for instance,
\[
\rho^{\Trs}_{0,X}(x\ass  e) =
  \lambda s.\, \ret_{s[x \ass  \oname{eev}(e,s)]}\\\qqand
 \rho^{\Cos}_{0,X}((c,u)\seqcomp\, (q,f)) =
  \begin{cases}
d\seqcomp q & \text{if } $u\,=\,d$;\\
((d \seqcomp q),s) & \text{if } $u\,=\,(d,\,s)$;\\
([q]_{s},s) & \text{if } $u\,=\,s$.
\end{cases}
\]
See \ifarx{Appendix~\ref{app:gsos-laws}}{{\cite[App.~A]{stateful25-arxiv}}} for
  the full definition of $\rho_0$.
We will next introduce trace, cost, and termination
semantics for  $B_0$-coalgebras (including the operational model of \whiletwo) and demonstrate that unlike the case of $D^\S$-coalgebras in the previous section, these notions can be directly captured by coalgebraic similarity for suitable choices of relation liftings and weakenings. In this way, we are able to derive the respective congruence properties for \whiletwo using abstract GSOS theory.

\begin{remark}\label{rem:gsos-law-nondet}
To model the various forms of semantics coalgebraically, we extend the behaviour functor $B_0$ to the `nondeterministic' functor $B\colon \Set^2\to \Set^2$ given by
\begin{equation}\label{eq:beh-whiletwo-nondet} B_\Trs X = X_\Cos^\store \qqand B_\Cos X =\Pow (X_\Cos \times \store + X_\Cos + \store).  \end{equation}
Note that $\Pow$ only appears in the writer sort. This allows us to work with weak transitions of writers. We tacitly identify a $B_0$-coalgebra $(X,\chi_0)$ with the deterministic $B$-coalgebra $(X,\chi)$ given by
\[ \chi^\Trs = \chi_0^\Trs \qqand \chi^\Cos(c)=\{\chi_0^\Cos(c)\}. \]
The GSOS law $\rho_0$ of $\Sigma$ over $B_0$ extends to a GSOS law $\rho$ of $\Sigma$ over $B$ by using $\rho_0$ elementwise, e.g.,
\begin{equation}\label{eq:rho-x-example}
 \rho^{\Cos}_X((c,U)\seqcomp\, (q,f)) = \{ d\seqcomp\, q \mid d\in U \} \cup \{ ((d \seqcomp\, q),s) \mid (d,s)\in U \} \cup
\{ ([q]_{s},s) \mid s\in U \}.
\end{equation}
Since $\rho$ carries the same information as $\rho_0$, the operational models of $\rho_0$ and $\rho$ coincide modulo the above identification of $B_0$- and $B$-coalgebras.
\end{remark}

\subsection{Trace Semantics for $\boldsymbol{\whiletwo}$}\label{sec:whiletwo-trace}
To define trace semantics for $B_0$-coalgebras, we need to introduce suitable weak transitions:

\begin{notation}\label{not:weaktrans-1}
Let $\chi\colon X\to BX$ be a $B$-coalgebra. For $p\in X_\Trs$, $c,d\in X_\Cos$ and $s\in \S$ we write
\begin{itemize}
\item $p,s\to c$ if $\chi^\Trs(p)(s)=c$, and $c\to d$ / $c\xto{s} d$ / $c\downarrow s$ if $d$ / $(d,s)$ / $s \in \chi^\Cos(c)$, respectively.
\item $c\xTo{1} d$ if there exist $n\geq 0$ and $c_0,\ldots,c_n\in X_\Cos$ with $c=c_0 \to c_1\to \cdots \to c_n=d$.
\item $c\xTo{1,s} d$ if there exists $c'\in X_\Cos$ with $c\xTo{1} c'\to d,s$;
\item $c\Downarrow^1 s$ if there exists $c'\in X_\Cos$ with $c\xTo{1} c'\downarrow s$.
\end{itemize}
The \emph{trace map} of a $B_0$-coalgebra, i.e.~a deterministic $B$-coalgebra $(X,\chi)$, is the two-sorted map
 $\trc\colon X\to ((\store^\infty)^\store,\, \store^\infty)$
defined as follows:
\begin{itemize}
\item $\trc^\Trs(p)(s)=\trc^\Cos(c)$ if $p,s\to c$;
\item $\trc^\Cos(c)=(s_1,\ldots, s_n,s)$ if there exist $c=c_0,\ldots,c_n\in X_\Cos$ with $c_i\xTo{1} c_{i+1},s_{i+1}$ ($i<n$) and $c_n\Downarrow^1 s$;
\item $\trc^\Cos(c)=(s_1,s_2,s_3,\ldots)$ if there exist $c=c_0,c_1,c_2,\ldots \in X_\Cos$ with $c_i\xTo{1} c_{i+1},s_{i+1}$ for all $i\in \bN$.
\end{itemize}
\end{notation}
We say that $p,q\in X_\Trs$ are \emph{trace equivalent} if $\trc^\Trs(p)=\trc^\Trs(q)$; similarly for $c,d\in X_\Cos$. To capture trace equivalence for $B_0$-coalgebras, we work with the extended functor $B$ \eqref{eq:beh-whiletwo-nondet} and choose a relation lifting of $B$ and a notion of weakening for $B$-coalgebras as follows:

\paragraph*{Relation lifting} We take the relation lifting $\barB_\trc\colon \Rel(\Set^2)\to \Rel(\Set^2)$ of $B$ defined by
\begin{equation}\label{eq:barB-trc} (\barB_\trc R)_\Trs = R_\Cos^\S \qqand (\barB_\trc R)_\Cos = \vec{\Pow}(R_\Cos\times \Delta_\store + R_\Cos + \Delta_\store). \end{equation}
Here $\vec{\Pow}$ is the asymmetric Egli-Milner lifting (\Cref{ex:can-lift-endofunctor}), $\Delta_\S$ is the identity relation of $\S$, and $(-)^\S$, $\times$, $+$ are the $\S$-fold power, the product and the coproduct of single-sorted relations (\Cref{sec:catback}).

\paragraph*{Weakening} The weakening $\tilde{\chi}_\trc\colon X\to BX$ of a coalgebra $\chi\colon X\to BX$
is defined by
\begin{equation}\label{eq:weakening-trc} \tilde{\chi}^\Trs_\trc = \chi^\Trs \qqand \tilde{\chi}_\trc^\Cos(c) = \{ (d,s) \mid c\xTo{1} d,s \} \cup \{ d \mid c\xTo{1} d \} \cup \{ s \mid c\Downarrow^1 s\}.  \end{equation}
\ifarx{We verify in the appendix (\Cref{lem:weakening}) that $\tilde{\chi}_\trc$
    is a weakening w.r.t.~$\barB_\trc$}{See \cite[App.~B]{stateful25-arxiv} for
    details as to why $\tilde{\chi}_\trc$ is a weakening}.

For every $B$-coalgebra $(X,\chi)$, we thus obtain from \Cref{def:coalg-sim} the generic notions of $\barB_\trc$-simulation and $\barB_\trc$-similarity on $(X,\chi,\tilde{\chi}_\trc)$. They spell out as follows:

\begin{definition}\label{def:trace-sim}
A \emph{trace simulation} on a $B$-coalgebra $(X,\chi)$ is a relation $R\monoto X\times X$ such that the following holds for all $p,q \in X_\Trs$, $c,d\in X_\Cos$ and $s \in
    \store$:
    \begin{enumerate}
    \item If $R_\Trs(p,q)$ and $p,s\to c$ then there exists $d$ such that $q,s \to d$ and $R_{\Cos}(c,d)$.
    \item If $R_\Cos(c,d)$ and $c\to c'$ then there exists $d'$ such that $d\xTo{1} d'$ and $R_\Cos(c',d')$.
    \item\label{trc-sim-cond} If $R_\Cos(c,d)$ and $c\to c',s$ then there exists $d'$ such that $d\xTo{1,s} d'$ and $R_\Cos(c',d')$.
    \item If $R_\Cos(c,d)$ and $c\downarrow s$ then $d\Downarrow^1 s$.
    \end{enumerate}
\emph{Trace similarity} is the greatest trace simulation.
\end{definition}
For $B_0$-coalgebras, this notion of simulation captures precisely trace equivalence:
\begin{proposition}\label{prop:trc-sim}
For every $B_0$-coalgebra, trace similarity coincides with trace equivalence.
\end{proposition}

\subsection{Cost Semantics for $\boldsymbol{\whiletwo}$}\label{sec:whiletwo-stepcounting}
Given a $B_0$-coalgebra $\chi\colon X\to B_0X$, we define a two-sorted map
$\scn\colon X\to ((\bN\times \store+1)^\store,\,\bN\times \store+1)$:
\[
\scn^{\Trs}(p)(s) = \scn^\Cos(c) \quad\text{if}\quad p,s\to c,\qquad
\scn^{\Cos}(c) =
\begin{cases}
(n,s) & \text{if } \trc^{\Cos}(c) = (s_1,\ldots, s_n,s)\in \store^{+}, \\
* & \text{if } \trc^{\Cos}(c) \in \store^\omega.
\end{cases}
\]
We say that $p,q\in X_\Trs$ are \emph{cost equivalent} if $\scn^\Trs(p)=\scn^\Trs(q)$; similarly for $c,d\in X_\Cos$. To capture cost equivalence for $B_0$-coalgebras, we work with a relation lifting and weakening chosen as follows:

\paragraph*{Relation lifting} We take the relation lifting $\barB_\scn$ of $B$ defined below, where $\top_\store = \store\times\store$:
\begin{equation}\label{eq:barB-scn} (\barB_\scn R)_\Trs = R_\Cos^\S \qqand (\barB_\scn R)_\Cos = \vec{\Pow}(R_\Cos\times \top_\store + R_\Cos + \Delta_\store). \end{equation}
Thus, in comparison to $\barB_\trc$, the first occurrence of $\Delta_\S$ in the writer sort is replaced with $\top_\store$.

\paragraph*{Weakening} We take the same weakening as for trace semantics:
$\tilde{\chi}_\scn=\tilde{\chi}_\trc$.
\ifarx{We verify in the appendix (\Cref{lem:weakening-scn}) that
    $\tilde{\chi}_\scn$ is a weakening with respect to $\barB_\scn$}{See \cite[App.~B]{stateful25-arxiv} for
    details as to why $\tilde{\chi}_\scn$ is a weakening with respect to $\barB_\scn$}.

Given a $B$-coalgebra $(X,\chi)$, the notion of $\barB_\scn$-simulation on $(X,\chi,\tilde{\chi}_\scn)$ now gives:
\begin{definition}\label{def:cost-sim}
A \emph{cost simulation} on a $B$-coalgebra $(X,\chi)$ is a relation $R\monoto X\times X$ such that the following holds for all $p,q \in X_\Trs$, $c,d\in X_\Cos$ and $s \in
    \store$:
    \begin{enumerate}
    \item If $R_\Trs(p,q)$ and $p,s\to c$ then there exists $d$ such that $q,s \to d$ and $R_{\Cos}(c,d)$.
    \item If $R_\Cos(c,d)$ and $c\to c'$ then there exists $d'$ such that $d\xTo{1} d'$ and $R_\Cos(c',d')$.
    \item\label{scn-sim-cond} If $R_\Cos(c,d)$ and $c\xto{s} c'$ then there exist $d',s'$ such that $d\xTo{1,s'} d'$ and $R_\Cos(c',d')$.
    \item If $R_\Cos(c,d)$ and $c\downarrow s$ then $d\Downarrow^1 s$.
    \end{enumerate}
\emph{Cost similarity} is the greatest cost simulation.
\end{definition}
The only difference to trace simulations lies in condition \ref{scn-sim-cond}: due the more permissive lifting, it suffices to match $c\xto{s} c'$ with a weak transition $d\xTo{1,s'} d'$ where not necessarily $s=s'$. However, cost similar terms still produce the same \emph{number} of states, so we get:

\begin{proposition}\label{prop:scn-sim}
For every $B_0$-coalgebra, cost similarity coincides with cost equivalence.
\end{proposition}

\subsection{Termination Semantics for $\boldsymbol{\whiletwo}$}
\label{sec:terminationsemantics}

Given a $B_0$-coalgebra $\chi\colon X\to B_0X$, we define the two-sorted map
$\trm\colon X\to ((\store+1)^\store,\,\store+1)$
by
\[
\trm^{\Trs}(p)(s) = \trm^\Cos(c) \quad\text{if}\quad p,s\to c,\qquad
 \trm^{\Cos}(c) =
\begin{cases}
s & \text{if } \trc^{\Cos}(c) = (s_1,\ldots,s_n,s), \\
* & \text{if } \trc^\Cos(c)\in \S^\omega.
\end{cases}
\]
We say that $p,q\in X_\Trs$ are \emph{termination equivalent} if $\trm^\Trs(p)=\trm^\Trs(q)$; similarly for $c,d\in X_\Cos$. We capture termination equivalence via relation liftings and weakenings as follows:

\paragraph*{Relation lifting} We take the same relation lifting as for cost semantics:
\begin{equation}\label{eq:barB-trm} \barB_\trm = \barB_\scn.\end{equation}

\paragraph*{Weakening} We extend the weak transitions $\xTo{1}$, $\Downarrow^1$ of \Cref{not:weaktrans-1} to a more liberal version:
\begin{notation}\label{not:weak-trans-2}
Given a coalgebra $\chi\colon X\to BX$ and $c,d\in X_\Cos$, $s\in \store$, we write
\begin{itemize}
\item $c\xTo{2} d$ and $c\xTo{2,s} d$ if there exist $n\geq 0$ and $c=c_0,\ldots,c_n=d\in X_\Cos$ such that, for each $i<n$,
\[ c_i\to c_{i+1} \qquad \text{or}\qquad  \exists s_{i+1}.\, c_i\xto{s_{i+1}} c_{i+1}.\]
\item $c\Downarrow^2 s$ if there exists $c'\in X_\Cos$ such that $c\xTo{2} c' \downarrow s$.
\end{itemize}
\end{notation}
The weak transitions $\xTo{2}$, $\Downarrow^2$ thus completely ignore the intermediate states produced in a sequence of strong transitions. We define the weakening $\tilde{\chi}_\trm\colon X\to BX$ of $\chi$ by
\begin{equation}\label{eq:weakening-trm} \tilde{\chi}_\trm^\Trs = \chi^\Trs \qqand \tilde{\chi}_\trm^\Cos(c) = \{ (d,s) \mid c\xTo{2,s} d\} \cup \{ d \mid c\xTo{2} d\} \cup \{ s \mid c\Downarrow^2 s\}.  \end{equation}
\ifarx{We verify in the appendix (\Cref{lem:weakening-trm}) that
    $\tilde{\chi}_\trm$ is a weakening with respect to $\barB_\trm$}{See \cite[App.~B]{stateful25-arxiv} for
    details as to why $\tilde{\chi}_\trm$ is a weakening with respect to $\barB_\trm$}.

For every $B$-coalgebra $(X,\chi)$, the notion of $\barB_\trm$-simulation on $(X,\chi,\tilde{\chi}_\trm)$ then corresponds to:
\begin{definition}\label{def:term-sim}
A \emph{termination simulation} on a $B$-coalgebra $(X,\chi)$ is a relation $R\monoto X\times X$ such that the following holds for all $p,q \in X_\Trs$, $c,d\in X_\Cos$ and $s \in
    \store$:
    \begin{enumerate}
    \item If $R_\Trs(p,q)$ and $p,s\to c$ then there exists $d$ such that $q,s \to d$ and $R_{\Cos}(c,d)$.
    \item If $R_\Cos(c,d)$ and $c\xto{s} c'$ or $c\to c'$, then there exists $d'$ such that $d\xTo{2} d'$ and $R_\Cos(c',d')$.
    \item If $R_\Cos(c,d)$ and $c\downarrow s$ then $d\Downarrow^2 s$.
    \end{enumerate}
\emph{Termination similarity} is the greatest termination simulation and is denoted by $\preceq_\trm$.
\end{definition}
Since termination similarity only observes terminating states, we get:
\begin{proposition}\label{prop:trm-sim}
For every $B_0$-coalgebra, the relation $\preceq_\trm\cap \succeq_\trm$ equals termination equivalence.
\end{proposition}

\subsection{Compositionality of $\boldsymbol{\whiletwo}$ and $\boldsymbol{\while}$}
We have shown that trace, cost, and termination equivalence for $B_0$-coalgebras correspond to coalgebraic similarity for suitable choices of relations liftings and weakenings. We now instantiate the general compositionality result for abstract GSOS (\Cref{thm:congruence-abstract-gsos}) to the AOS
\[\O_\star = (\Sigma,\ol{\Sigma},B,\barB_\star,\rho,\gamma,\tilde{\gamma}_\star)\qquad (\star\in \{\trc,\scn,\trm\}),\]
where the functor $B$ is ordered by $f\preceq g$ iff $f^\Trs=g^\Trs$ and $\forall z\in Z_\Cos.\, f^\Cos(z)\seq g^\Cos(z)$ for $f,g\colon Z\to BX$, and moreover $\rho$ is the extended GSOS law \eqref{eq:rho-x-example} for \whiletwo and $\gamma$ is its operational model. This yields:
\begin{theorem}[Compositionality of \whiletwo]\label{thm:whiletwo-cong}
Trace equivalence, cost equivalence and termination equivalence form a congruence on the operational model of \whiletwo.
\end{theorem}

\begin{proof}[Proof sketch]
 We only need to verify the conditions (1)--(3) of \Cref{thm:congruence-abstract-gsos}. We illustrate the case of trace semantics, the other two being almost identical. Put $\barB=\barB_\trc$ and $\tilde{\gamma}=\tilde{\gamma}_\trc$.
\begin{enumerate}
\item Up-closure and preservation of composition and identities by $\barB$ follow by an easy computation.
\item We need to show that the GSOS law $\rho$ is liftable, that is, for each $R\monoto X\times X$ in $\Set^2$ the map
$\rho_X\colon \Sigma(X\times BX)\to B\Sigmas X$ is a relation morphism from $\ol{\Sigma}(R\times \barB R)$ to $\barB\ol{\Sigma}^{\star} R$. This boils down to observing that the rules of \whiletwo (\Cref{fig:comp-rules}) are parametrically polymorphic. Let us give the argument for the operator~`$\seqcomp$' of arity $\Cos\times \Trs\to \Trs$. Consider two elements $(c,U) \seqcomp\, (q,f)$ and $(c',U') \seqcomp\, (q',f')$ of $\Sigma_\Cos(X\times BX)$ that are related in $(\ol{\Sigma}(R\times \barB R))_\Cos$; in particular, $(U,U')\in \vec{\Pow}(R_\Cos\times \Delta_\S + R_\Cos + \Delta_\S)$ and $(q,q')\in R_\Trs$. Since
\begin{align*}
\rho^{\Cos}_X((c,U)\seqcomp\, (q,f)) &= \{ d\seqcomp\, q \mid d\in U \} \cup \{ ((d \seqcomp\, q),s) \mid (d,s)\in U \} \cup
\{ ([q]_{s},s) \mid s\in U \},\\
\rho^{\Cos}_X((c',U')\seqcomp\, (q',f')) &= \{ d'\seqcomp\, q' \mid d'\in U' \} \cup \{ ((d' \seqcomp\, q'),s) \mid (d',s)\in U' \} \cup
\{ ([q']_{s},s) \mid s\in U' \}
\end{align*}
by \eqref{eq:rho-x-example}, it immediately follows that
\begin{equation}\label{eq:proof-goal-sketch} (\rho_X^\Cos( (c,U)\seqcomp\, (q,f) ),\, \rho_X^\Cos((c',U')\seqcomp\, (q',f') )) \in \vec{\Pow}((\ol{\Sigma}^{*}R)_\Cos\times \Delta_\S + (\ol{\Sigma}^{*}R)_\Cos + \Delta_\S)=(\barB \ol{\Sigma}^\star R)_\Cos,
\end{equation}
as required. The other operators are treated analogously.
\item Finally, we show that $(\mu\Sigma,\gamma,\tilde{\gamma})$ is a lax $\rho$-bialgebra. By \eqref{eq:lax-bialgebra-fo} and the definition of $\tilde{\gamma}$, this means that the writer rules
of \Cref{fig:comp-rules} remain sound in the operational model $\gamma\colon (\Tr,\Co)\to B(\Tr,\Co)$ when strong transitions $\to$, $\downarrow$ are replaced with corresponding weak transitions $\xTo{1}$, $\Downarrow^1$. For illustration, let us consider two of the writer rules for sequential composition and their weak versions:
\[
    \inference{c \to d}{c\seqcomp q
      \to d\seqcomp q}
\qquad
  \inference{c \xto{s} d}{c\seqcomp q
      \xto{s} d\seqcomp q}
\qquad
    \inference{c \xTo{1} d}{c\seqcomp q
      \xTo{1} d\seqcomp q}
\qquad
 \inference{c \xTo{1,s} d}{c\seqcomp q
      \xTo{1,s} d\seqcomp q}
 \]
Soundness of the third and fourth rule means that if the premise holds, then the conclusion holds. The third rule is sound because it follows by repeated application of the first one. The fourth rule is sound because it emerges from the third rule followed by one application of the second rule.\qedhere
\end{enumerate}
\end{proof}
It remains to relate the compositionality of \whiletwo to that of the original language \while. To this end, we make the key observation that the embedding of \while into \whiletwo is \emph{semantics-preserving}: the trace of an \while-term is the same regardless of whether it is executed as a program of \while or as a reader of \whiletwo. In the following we denote the trace maps of \while and \whiletwo by $\trc_\while\colon \Progs\to (\S^\infty)^\S$ and $\trc_{\whiletwo}\colon (\Tr,\Co)\to ((\S^\infty)^\S, \S^\infty)$; recall that  $\Progs=\Tr$. Similarly for the cost and termination maps.

\begin{theorem}\label{thm:semantics-pres}
The trace, cost, and termination semantics for \while and \whiletwo coincide:
\[\trc_\while = \trc^\Trs_{\whiletwo},\qquad  \scn_\while = \scn^\Trs_{\whiletwo},\qquad \trm_\while = \trm^\Trs_{\whiletwo}.\]
\end{theorem}

\begin{proof}[Proof sketch]
It suffices to prove the first equality; the other two are then immediate because cost and termination semantics in both \while and \whiletwo are derived from trace semantics. The key to the proof lies in relating transitions in the operational model \eqref{eq:while-opmodel} of \while to (weak) transitions $\to$, $\xTo{1}$, $\Downarrow^1$ (\Cref{not:weaktrans-1}) in the operational model \eqref{eq:while2opmodel} of \whiletwo. We write $\To$, $\Downarrow$ for $\xTo{1}$, $\Downarrow^1$. By structural induction one can show that the following holds for all $p,p'\in \mS$ and $s,s'\in \store$:

\begin{enumerate}
\item\label{statement-1-sketch}
If $p,s\to p',s'$ in \while, then there exist $c_0,c_1,c_2,c_3\in \Co$ and transitions in \whiletwo as depicted below on the left. (The upper arrow is an \while-transition, the remaining arrows are \whiletwo-transitions.)
\item\label{statement-2-sketch} If $p,s\downarrow s'$ in \while, then there exists $c\in \Co$ such that $p,s\to c$ and $c\Downarrow s'$ in \whiletwo.
\end{enumerate}
The statement of the theorem now easily follows. Suppose that $\trc_\while(p)(s)=(s_1,\ldots,s_n,s')$, witnessed by \while-transitions
$p,s\to  p_1,s_1 \to p_2,s_2 \to \cdots \to p_n,s_n\downarrow s'$. Applying \ref{statement-1-sketch} to the first $n$ transitions and \ref{statement-2-sketch} to the last one yields the situation below on the right, where the $\bullet$'s are elements of $\Co$, the transitions in the upper row are in \while, and the remaining transitions are in \whiletwo.
\[
\begin{tikzcd}[column sep=15, row sep=9]
p,s \ar{rr} \ar{dd} && p',s' \ar{d} \\
&& c_3 \ar[Rightarrow]{d} \\
c_0 \ar[Rightarrow]{r}{s'} & c_1 \ar[Rightarrow]{r} & c_2
\end{tikzcd}
\quad
\begin{tikzcd}[column sep=10, row sep=10]
p,s \ar{rr} \ar{dd} && p_1,s_1 \ar{d} \ar{rr} && p_2,s_2 \ar{d} & \cdots\quad \ar{r} & p_n,s_n \ar{d} &[-1.5em] \downarrow  s' \\
&& \bullet \ar[Rightarrow]{d} && \bullet \ar[Rightarrow]{d} && \bullet \ar[Rightarrow]{d}&  \\
\bullet \ar[Rightarrow]{r}{s_1} & \bullet \ar[Rightarrow]{r} & \bullet \ar[Rightarrow]{r}{s_2} & \bullet \ar[Rightarrow]{r} & \bullet & \quad \bullet \ar[Rightarrow]{r} & \bullet & \Downarrow  s'
\end{tikzcd}
\]
Thus $\trc_{\whiletwo}^\Trs(p)(s)=(s_1,\ldots,s_n,s')=\trc_{\while}(p)(s)$. The case where $\trc_\while(p)(s)$ is an infinite trace is treated analogously. We conclude that $\trc_\while = \trc_{\whiletwo}^\Trs$ as claimed.
\end{proof}

In particular, two \while-programs are trace/cost/termination equivalent iff they are trace/cost/termination equivalent as \whiletwo-readers, and so we recover the essential parts of \Cref{thm:comp-while}:
\begin{theorem}[Compositionality of \while]
Trace equivalence, cost equivalence, and termination equivalence form congruences on the operational model of \while.
\end{theorem}
It turns out that our above results on \whiletwo and its tight relation to \while are not an accident: they extend to a whole family of stateful languages emerging from the general setting of \emph{stateful SOS}, and in fact are arguably best understood at that level of abstraction.

\section{From Stateful SOS to Abstract GSOS}
\label{sec:translations}
In this section we revisit the developments for \while and \whiletwo at the level of \emph{stateful SOS}, a general rule format for modelling stateful languages introduced by \citet{GoncharovMiliusEtAl22}. In particular, we will recover their main compositionality result by employing the methodology of abstract GSOS.

\subsection{Stateful SOS Specifications}
We start by recalling the definition of the stateful SOS format. Let us first settle some notation:
\begin{notation}  \label{not:stateful-sos}
\begin{enumerate}
\item We fix a single-sorted algebraic signature
  $\Sig$; as before, we also denote the induced polynomial functor by $\Sig \c \Set \to \Set$. Moreover, we fix a set $\S$ of \emph{states}. We think of $\Theta$ as the signature of a single-sorted stateful language such as \while and of $\store$ as its set of variable stores.
\item We continue to work with the functors
\begin{align*}
& D\colon \Set\to \Set,&& DX=X\times \store + \store,\\
& B\colon \Set^2\to \Set^2, && B_\Trs X = X_\Cos^\store, && B_\Cos X = \Pow(X_\Cos\times \store + X_\Cos + \store).
\end{align*}
\item Moreover, we fix a countably infinite set of variables
\[\Var=\{ x_n \mid n\in \bN\}\cup \{ y_n\mid n\in \bN \}.\]
\end{enumerate}
\end{notation}

\begin{definition}
  \label{def:isosrule}
\begin{enumerate}
\item
  A \emph{\isos} rule for an $n$-ary operator $\f \in \Sig$ is an expression of either form
  \begin{equation}\label{eq:sos-rule-progressing}
      \inference{x_i, s\to y_i,s_i \; (i\in W) \quad x_j,s\downarrow s_j\; (j\in\ol{W})}{\f(x_1,\ldots,x_n),s \to t,s'}
\end{equation}
  \begin{equation}\label{eq:sos-rule-terminating}
 \inference{x_i, s\to y_i,s_i \; (i\in W) \quad x_j,s\downarrow s_j\; (j\in\ol{W})}{\f(x_1,\ldots,x_n),s\downarrow s'}
\end{equation}
where $W\seq \{1,\ldots,n\}$, $\ol{W}=\{1,\ldots,n\}\setminus W$, $s,s',s_i\in \store$ for $i=1,\ldots,n$ and $t\in \Sigmas\Var$ is a term in the variables $x_1,\ldots,x_n$ and $y_i$ ($i\in W$). The tuple $(W,s,s_1,\ldots,s_n)$ is the \emph{trigger} of the rule.
\item A \emph{stateful SOS specification} is a set $\L$ of stateful SOS rules such that for each $n$-ary $\f\in \Theta$ and each $W\seq \{1,\ldots,n\}$, $s,s_1,\ldots,s_n\in \store$, there is a unique rule for $\f$ with trigger $(W,s,s_1,\ldots,s_n)$ in $\L$.
\end{enumerate}
\end{definition}

\begin{remark}
  Stateful SOS specifications are in a bijective correspondence
  with \emph{stateful SOS laws}, which are natural transformations of the form
  \[
    \delta_{X} \c \Sig(X \times  DX)\times\store\to  D(\Sigs X) \qquad (X\in \Set).
  \]
Similar to a GSOS law, a stateful SOS law simply encodes the rules of its corresponding stateful SOS specification into a natural family of functions; see \citet{GoncharovMiliusEtAl22} for more details.
\end{remark}

\begin{notation}\label{not:sos-conventions}
Unused premises can be dropped as expected: given an $n$-ary operator $\f$ and disjoint subsets $W_0,W_1\seq \{1,\ldots.n\}$, the rule on the left below represents the set of all rules on the right where $W\seq \{1,\ldots,n\}$ satisfies $W_0\seq W$ and $W_1\seq \ol{W}$ and $s_i\in \store$ for $i\in\{1,\ldots,n\}\setminus (W_0\cup W_1)$.
  \begin{equation*}
      \inference{x_i, s\to y_i,s_i \; (i\in W_0) \quad x_j,s\downarrow s_j\; (j\in W_1)}{\f(x_1,\ldots,x_n),s \to t,s'} \qquad  \inference{x_i, s\to y_i,s_i \; (i\in W) \quad x_j,s\downarrow s_j\; (j\in\ol{W})}{\f(x_1,\ldots,x_n),s\to t,s'}
\end{equation*}
Note that $t$ is then a term in the variables $x_1,\ldots,x_n$ and $y_i$ for $i\in W_0$ and that $t$ and $s'$ only depend on $s$ and $s_i$ for $i\in W_0\cup W_1$. Analogously for rules of the form \eqref{eq:sos-rule-terminating}. The special case $W_0=W_1=\emptyset$ corresponds to premise-free rules, where the behaviour of $\f$ is fully determined by the current state:
  \begin{equation}\label{eq:passive}
      \inference{}{\f(x_1,\ldots,x_n),s \to t,s'} \qquad  \inference{}{\f(x_1,\ldots,x_n),s \downarrow s'}
\end{equation}
We say that an operator $\f$ is \emph{passive} if all rules for $\f$ are premise-free. An \emph{active} operator is one that is not passive, i.e.\ its behaviour actually depends on the behaviour of some of its operands.
\end{notation}

A stateful SOS specification completely and deterministically defines the behaviour of each closed $\Theta$-term $\f(p_1,\ldots,p_n)$ depending on the current state and the possible behaviours of its subterms $p_i$ on that state. Its operational model is the coalgebra that runs $\Theta$-terms according to the rules:

\begin{definition}\label{def:opmodel-stateful-sos}
\begin{enumerate}
\item
Every stateful SOS specification $\L$ induces a GSOS law $\rho^\L$ of $\Theta$ over $D^\store$ with
\[ \rho^\L_X\colon \Theta(X\times (X\times \S+\S)^\S) \to (\Sigmas X \times \S + \S)^\S \]
defined as follows: given an $n$-ary $\f\in \Theta$, $p_1,\ldots,p_n\in X$, $u_1,\ldots,u_n\in (X\times \S+\S)^\S$, and $s\in \S$, put
\[W=\{ i\in \{1,\ldots,n\} \mid u_i(s)\in X\times \S \},\qquad (t_i,s_i)=u_i(s),\;\, i\in W,\qquad s_j=u_j(s),\;\, j\in \ol{W}.
\]
 Consider the unique rule in $\L$ for $\f$ with trigger $(W,s,s_1,\ldots,s_n)$. If the rule is given by \eqref{eq:sos-rule-progressing}, put
\[\rho^\L_X(\f((p_1,u_1),\ldots,(p_n,u_n)))(s)= (t[p_1/x_1,\ldots, p_n/x_n, t_{i_1}/y_{i_1},\ldots, t_{i_k}/y_{i_k} ],s')\]
where $W=\{i_1,\ldots,i_k\}$ and $-/-$ denotes substitution. If the rule is given by \eqref{eq:sos-rule-terminating}, put
\[ \rho^\L_X(\f((p_1,u_1),\ldots,(p_n,u_n)))(s) = s'. \]
\item The \emph{operational model} of $\L$ is the operational model of the GSOS law $\rho^\L$, denoted by
\begin{equation}\label{eq:l-opmodel} \gamma^\L\colon \mS\to (\mS\times \S+\S)^\S.\end{equation}
\end{enumerate}
\end{definition}

\begin{example}
The rules of \while (\Cref{fig:rules-while}) form a stateful SOS specification $\L$ for the signature $\Theta$ of \while, modulo renaming of variables and dropping unused premises (\Cref{not:sos-conventions}). For instance, the rule for sequential composition shown on the left represents the rules on the right for all $s_2\in \store$.
\[
    \inference{p,s\downarrow s'}{p \seqcomp q, s
      \to q,s'}
\qquad
\inference{x_1,s\downarrow s'\quad x_2,s\to y_2,s_2}{x_1 \seqcomp x_2, s
      \to x_2,s'}
\qquad
\inference{x_1,s\downarrow s'\quad x_2,s\downarrow s_2}{x_1 \seqcomp x_2, s
      \to x_2,s'}
\]
The induced GSOS law $\rho^\L$ is that of \Cref{ex:while-gsos-law}, and the operational model of $\L$ is the coalgebra~\eqref{eq:while-opmodel} that runs \while-terms according the rules of the language.
\end{example}

\subsection{Reader-Writer Semantics for Stateful SOS}
We will now generalize the transformation of \while into \whiletwo to the level of stateful SOS specifications. This requires a syntactic restriction of the rules:

\begin{defn}
  A \isos{} specification is \emph{cool} if for every $n$-ary active
  operator $\f$ there exists $j\in \{1,\ldots, n\}$ (the \emph{receiving position of $\f$}) such that all rules for $\f$ are of one of the following forms:
     \begin{equation}\label{eq:cool-1}
        \inference{\goes{x_{j},s}{y_j,s'}}{\goes{\f(x_{1},\ldots,x_{j-1},x_j,x_{j+1},\ldots,x_n),s}{\f(x_{1},\ldots,x_{j-1},y_j,x_{j+1},\ldots,x_n),s'}}
      \end{equation}
      \begin{equation}\label{eq:cool-2}
        \inference{\rets{x_{j},s}{s'}}{\goes{\f(x_{1},\dots,x_{n}),s}{t,s''}} \qquad \inference{\rets{x_{j},s}{s'}}{\rets{\f(x_{1},\dots,x_{n}),s}{s''}}
      \end{equation}
     where in \eqref{eq:cool-2} we have $t\in \Sigmas(\{x_{1},\dots,x_{n}\} \smin \{x_{j}\})$, and $t$ and $s''$ depend only on $s'$ but not on $s$.

  \end{defn}

The cool format asserts that an active operator $\f$ runs its $j$-th
subterm until termination and then discards it, proceeding to a state
derivable from the terminating state of the subterm. The terminology alludes to the
  \emph{cool} congruence formats for labelled transition systems by \citet{Bloom95} and
  \citet{Glabbeek11}. In particular, rules of type~\eqref{eq:cool-1} are a stateful version of
\emph{patience rules}~\cite{Glabbeek11}.

\begin{example}
The language \while (\Cref{fig:rules-while}) is cool. Its only active operator is sequential composition $-\seqcomp-$, which is receiving in the first position.
\end{example}

\begin{notation}\label{not:sig-l2} For the remainder of this section, we fix a cool stateful SOS specification $\L$ over the signature $\Theta$. We use the following notation for its passive and active operators:
\begin{itemize}
\item We let $\Theta_{\pas,n}$ denote the set of $n$-ary passive operators. For $\f\in \Theta_{\pas,n}$ and $s\in \store$ we put $o(\f,s)=t,s'$ or $o(\f,s)=s'$ if $\L$ contains the first or second rule \eqref{eq:passive}, respectively.
\item We let $\Theta_{\act,n,j}$ denote the set of $n$-ary active operators with receiving position $j$. For $\f\in \Theta_{\act,n,j}$ and $s,s'\in \store$, we put $o(\f,s,s')=t,s''$ or $o(\f,s,s')=s''$ if $\L$ contains the first or second rule \eqref{eq:cool-2}.
\end{itemize}
The signature $\Theta$ extends to the two-sorted signature $\Sigma$ (with sorts $\Trs$ and $\Cos$) containing the operators
\begin{itemize}
\item $\f\colon \Trs\times \cdots \times \Trs \to \Trs$ (with $n$ inputs) for every $n$-ary $\f\in \Theta$;
\item $\ol{\f}\colon \Trs\times \cdots\times \Trs \times \Cos \times \Trs \times\cdots\times \Trs\to \Cos$ (with $n$ inputs and $\Cos$ in the $j$-th position) for every $\f\in \Theta_{\act,n,j}$.
\item a constant $\ret_s\colon w$ and unary operators $[-]_s\colon \Trs\to \Cos$ and $s.-\colon \Cos\to \Cos$ for every $s\in \store$.
\end{itemize}
The signature $\Sigma$ corresponds to the polynomial functor $\Sigma\colon \Set^2\to \Set^2$ given by
\[ \Sigma_\Trs X = \Theta X_\Trs\qand   \Sigma_{\Cos} X
  = \underbrace{X_{\Trs}\times\store}_{[-]_s} \,+\, \underbrace{\store}_{\ret_s}
  \,+\, \underbrace{\store \times X_{\Cos}}_{s.-}\,+\,
 \coprod_{n,j}\coprod_{\f\in \Theta_{\act,n,j}} \underbrace{X_\Trs^{j-1}\times X_\Cos \times X_\Trs^{n-j}}_{\ol{\f}}.  \]
Elements of $(\mu\Sigma)_\Trs$ and $(\mu\Sigma)_\Cos$ are called \emph{readers} and \emph{writers}, respectively. Note that
\[ (\mu\Sigma)_\Trs = \mS. \]
\end{notation}
For the case where $\Theta$ is the single-sorted signature of \while, the induced $\Sigma$ is the two-sorted signature of \whiletwo. The refinement of the operational rules of \while to the reader-writer rules of \whiletwo (\Cref{fig:comp-rules}) then generalizes to the present setting as follows:

\begin{figure}
  \begin{gather*}
   \inference{\f\in \Theta_{\pas,n}\quad o(\f,s)=t,s'}{\f(x_1,\ldots,x_n),s \to s'.[t]_{s'}}[($\f$1)]
\qquad
  \inference{\f\in \Theta_{\pas,n}\quad o(\f,s)=s'}{\f(x_1,\ldots,x_n),s \to \ret_{s'}}[($\f$2)]\\
\qquad
     \inference{\f\in \Theta_{\act,n,j}}{\f(x_1,\ldots,x_j,\ldots,x_n),s \to \barf(x_1,\ldots, [x_j]_s,\ldots, x_n)} [($\f$3)]
  \end{gather*}

\vspace{1ex}
~\dotfill~
\vspace{1ex}
  \begin{gather*}
    \inference{\f\in \Theta_{\act,n,j} \quad v_j\to w_j}{\barf(x_1,\ldots,v_j,\ldots,x_n) \to \barf(x_1,\ldots, w_j,\ldots, x_n)}[($\barf$1)]
   \;
   \inference{\f\in \Theta_{\act,n,j}\quad o(\f,s,s')=t,s''\quad v_j\downarrow s'}{\barf(x_1,\ldots,v_j,\ldots,x_n) \xto{s''} [t]_{s''}}[($\barf$3)]  \\
 \inference{\f\in \Theta_{\act,n,j} \quad v_j\xto{s} w_j}{\barf(x_1,\ldots,v_j,\ldots,x_n) \xto{s} \barf(x_1,\ldots, w_j,\ldots, x_n)}[($\barf$2)]   \;
\inference{\f\in \Theta_{\act,n,j}\quad o(\f,s,s')=s'' \quad  v_j\downarrow s'}{\barf(x_1,\ldots,v_j,\ldots,x_n) \downarrow s''}[($\barf$4)]  \\
\inference{x_1,s\to v_1}{\run{s}{x_1} \to v_1}[($\text{[}-\text{]}$)]
    \qquad
    \inference{}{\ret_s \downarrow s}[($\ret$)]
    \qquad
    \inference{}{s.v_1 \xto{s} v_1}[($s.-$)]
  \end{gather*}
  \caption{Reader semantics of $\L^2$ (top) and writer semantics of $\L^2$ (bottom).}
  \label{fig:comp-rules-l2}
\end{figure}

\begin{definition}
We fix a new set of variables
\[ \Var^2 = \{ x_n \mid n\in \bN \} \uplus \{ v_n \mid n\in \bN \} \uplus \{ w_n \mid n\in \bN \}  \]
with $x_n$ representing readers and $v_n,w_n$ representing writers. (We use two types of writer variables for convenience, so that writer transitions in premises of rules can be written as $v_j\to w_j$ or $v_j\xto{s} w_j$.) The \emph{reader-writer extension} $\L^2$ of $\L$ is given by the operational rules of \Cref{fig:comp-rules-l2}. Its \emph{operational model} is the $B$-coalgebra
\begin{equation}\label{eq:l2-opmodel} \gamma^{\L^2}\colon \mu\Sigma\to ((\mu\Sigma)_\Cos^\store,\, \Pow((\mu\Sigma)_\Cos\times \S + (\mu\Sigma)_\Cos + \S)) \end{equation}
that runs terms according to the above rules.
\end{definition}
Analogous to the special case \whiletwo, the operational model of $\L^2$ is
actually deterministic (that is, it can be identified with a $B_0$-coalgebra),
and it forms the operational model of a GSOS law $\rho^{\L^2}$ of $\Sigma$ over
$B$ that encodes the rules of $\L^2$. \ifarx{The law $\rho^{\L^2}$ is given in
    Appendix~\ref{app:gsos-laws}}{For the full definition of $\rho^{\L^2}$, see
    \cite[App.~A]{stateful25-arxiv}}. As an alternative to the present rule-based definition, it is also possible to construct the law $\rho^{\L^2}$ diagrammatically from~$\rho^{\L}$ using functorial strength.

\subsection{Compositionality of $\boldsymbol{\L^2}$ and $\boldsymbol{\L}$}
We now derive compositionality of the reader-writer extension $\L^2$ of $\L$ with respect to trace, cost, and termination semantics as an application of the general compositionality result for abstract GSOS (\Cref{thm:congruence-abstract-gsos}). To this end, we instantiate the theorem to the AOS
\[\O_\star = (\Sigma,\ol{\Sigma},B,\barB_\star,\rho^{\L^2},\gamma^{\L^2},\tilde{\gamma}_\star^{\L^2})\qquad (\star\in \{\trc,\scn,\trm\}),\]
with the relation liftings $\barB_\star$ of $B$ given by \eqref{eq:barB-trc}, \eqref{eq:barB-scn}, \eqref{eq:barB-trm} and the weakenings $\tilde{\gamma}_\star^{\L^2}$ of the operational model $\gamma^{\L^2}$ given by \eqref{eq:weakening-trc}, \eqref{eq:barB-scn}, \eqref{eq:barB-trm}. Recall that $\barB_\trc$-similarity on $(\mu\Sigma,\gamma^{\L^2},\tilde{\gamma}^{\L^2})$ is trace equivalence (\Cref{prop:trc-sim}), $\barB_\scn$-similarity is cost equivalence (\Cref{prop:scn-sim}), and the relation $\preceq_\trm\cap \succeq_\trm$ (where $\preceq_\trm$ denotes $\barB_\trm$-similarity) is termination equivalence (\Cref{prop:trm-sim}). Consequently:
\begin{theorem}[Compositionality of $\L^2$]\label{thm:l2-comp}
Trace equivalence, cost equivalence and termination equivalence form congruences on the operational model of $\L^2$.
\end{theorem}
To prove the theorem, we only need to show that the conditions of \Cref{thm:congruence-abstract-gsos} hold. The arguments are similar to the proof sketch of \Cref{thm:whiletwo-cong}, with the coolness restriction on the given stateful SOS specification $\L$ being the key to the lax bialgebra condition.

Lastly, we show that congruence properties of the original language $\L$ can be reduced to those of $\L^2$. We denote the trace, cost, and termination maps on the operational models of $\L$ and $\L^2$ by
\begin{align*}
&\trc_\L\colon \mS \to (\S^\infty)^\S, && \trc_{\L^2}\colon \mu\Sigma\to ((\store^\infty)^\store,\\
&\scn_\L\colon \mS \to (\bN\times\S + 1)^\S, && \scn_{\L^2}\colon \mu\Sigma\to ((\bN\times \store+1)^\store, \\
&\trm_\L\colon \mS \to (\S + 1)^\S, && \trm_{\L^2}\colon \mu\Sigma\to ((\store+1)^\store,\,\store+1).
\end{align*}
Generalizing \Cref{thm:semantics-pres}, the extension of $\L$ to $\L^2$ leaves the semantics unchanged:
\begin{theorem}\label{prop:trace-pres} The trace, cost, and termination semantics for $\L$ and $\L^2$ coincide:
\[\trc_\L = \trc^\Trs_{\L^2},\qquad  \scn_\L = \scn^\Trs_{\L^2},\qquad \trm_\L = \trm^\Trs_{\L^2}.\]
\end{theorem}
From this theorem and the compositionality of $\L^2$, the following is immediate:

\begin{theorem}[Compositionality of $\L$]\label{thm:cong-l}
Trace equivalence, cost equivalence and termination equivalence are congruences on the operational model of~$\L$.
\end{theorem}

The congruence of trace and termination equivalence for cool stateful SOS specifications has been shown previously by \citet{GoncharovMiliusEtAl22}. However, in contrast to the latter work, our present approach uses abstract GSOS rather than ad hoc reasoning on the cool rule format, and is based on the principled extension of $\L$ to $\L^2$. In fact, our approach provides a deeper explanation of why the cool format works: it gives rise to a lax bialgebra structure in the extended language $\L^2$.

\section{Higher-Order Store}
\label{sec:higher-order}

So far we have shown how to cover stateful languages with
first-order store in abstract GSOS. We shall now demonstrate that the reader-writer approach to the
operational semantics of stateful languages is also well-suited for more complex
C-style and ML-style languages with higher-order store, where programs cannot only store basic values like integers, but also (pointers to) other programs. For that purpose, we
introduce \reflang, a C-style, untyped imperative language
with function pointers that follows the reader-writer paradigm. The presence of higher-order store yields semantics that
correspond to \emph{higher-order} abstract GSOS~\cite{GoncharovMiliusEtAl23}.
The inherent (bifunctorial) coalgebraic notions of program equivalence are thus
fundamentally of the higher-order variety; we focus on one such notion,
\emph{termination simulation}, and employ the higher-order abstract GSOS theory to build a sound proof method for the
standard \emph{contextual equivalence} of \reflang.

\subsection{An Imperative Language with Higher-Order Store}

The language \reflang is a basic, untyped imperative language with
references and higher-order store, following mainly \citet{ReusStreicher05}, as well as
\citet{Pierce02}. Like Reus and Streicher, we
work without variable binding, e.g.~$\lambda$-abstractions, allowing us to focus on the core phenomena arising from
higher-order store. We outline in \Cref{sec:higher-higher-order} how such additional features are incorporated.

We write $Loc$ for the abstract set of
\emph{locations} and use the metavariable $l$ to denote a generic location. Locations represent regions in memory that can be allocated.
The expressions of \reflang are generated by the grammar below (of course, more arithmetic operations may be added at wish):
\begin{equation}
\expr\owns e,r \Coloneqq\;
l \mid n \mid {!}e \mid e \oplus r \mid e \ominus r \quad (l \in Loc, n \in \mathbb{Z}).
\end{equation}
The set~$\store(X)$ of \emph{stores} in \reflang is parameterized by a set $X$, which
abstracts the set of readers being stored, and is defined to be the set of
\emph{partial} maps
\[
  \store(X) = Loc \rightharpoonup V(X) \qquad \text{where} \qquad V(X) = Loc +
  \mathbb{Z} + X.
\]
Given $s\in \S(X)$ we write $s[l\mapsto v]$ for the store that maps $l$ to $v$ and otherwise equals $s$. The partial evaluation of expressions is given by the map
$  \oname{eev}_{X} \c \store(X) \times \expr \rightharpoonup V(X)$, natural in $X$,
defined by
\[
  \begin{aligned}
    \oname{eev}_{X}(s,l)
    =\;& l, \qquad
    \oname{eev}_{X}(s,n)
    = n,\\
    \oname{eev}_{X}(s,{!}e)
    =\;& s(l) \text{\quad if\quad }~\oname{eev}_{X}(s,e) = l,\\
    \oname{eev}_{X}(s,e_{1} \oplus e_{2})
    =\;& n_{1} + n_{2}\text{\quad if\quad }~\oname{eev}_{X}(s,e_{i}) = n_{i},\\
    \oname{eev}_{X}(s,e_{1} \ominus e_{2})
    =\;& n_{1} - n_{2}\text{\quad if\quad }~ \oname{eev}_{X}(s,e_{i}) = n_{i}.
  \end{aligned}
\]
According to the above, expressions may evaluate to either a location, an
integer, or an element of~$X$. Stores are accessed by
invoking the dereferencing operator~$!$ on an expression that
evaluates to a valid location, meaning one in the domain of the store.

Analogous to the language \whiletwo, the terms of \reflang are divided into two sorts, \emph{readers}
and \emph{writers}, and are inductively generated by the
following grammar:
\begin{equation*}
  \label{eq:csyn}
  \begin{aligned}
    \Tr\owns p, q \Coloneqq\;
    & \mathsf{skip} \mid \mathsf{while}~e~p \mid e \ass  p
      \mid  \mathsf{if}~e~\mathsf{then}~p~\mathsf{else}~q \mid p \seqcomp q \mid
    \&p \mid \mathsf{expr}~e \mid \mathsf{proc}~p \qquad (e \in \expr)\\
    \Co\owns  c \Coloneqq\;& e \ass c \mid c \seqcomp q \mid \& c \mid s.c \mid \run{s}{p} \mid \ret_{v,s} \mid \ret_{s}
                             \hspace{9em} (s\in\store(\Tr), v \in V(\Tr))
  \end{aligned}
\end{equation*}
We highlight the important added features in comparison to \whiletwo. The reader
$\& p$ first evaluates $p$, then allocates a new region in memory containing
the value of~$p$ and returns the location of this region; $\mathsf{proc}\ p$
returns the reader~$p$ as a value; and $\mathsf{expr}\ e$
evaluates~$e$, returning the result. The behaviour of writers also changes as
they no longer only return an output store, but potentially also a \emph{value} in $V(\Tr)$, that is, an
integer, a location, or a reader term.
Differences aside, \reflang is another instance of the two-sorted
reader-writer approach to stateful languages that was earlier exemplified by
\whiletwo. In particular, it also features writers $\run{s}{p}$,
$\ret$ and $s.c$, and is implicitly an extension of a single-sorted language $\mathsf{Ref}$ whose program terms correspond to the readers of \reflang.

The full operational semantics of \reflang are specified by the inductive rules in
\Cref{fig:comp-rules-ho}. The rules apply to
stores $s\in \store(\Tr)$ and specify transitions of type $p,s \to c$ where $p\in \Tr$ and $c\in \Co$ (`given an input store $s \in \store(\Tr)$, the
reader $p$ continues as $c$'), and transitions $c\to d$ / $c\xto{s} d$ / $c\downarrow s$ / $c \downarrow v,s$ where $c,d\in \Co$ and $v\in V(\Tr)$. The first three types of writer transitions are as in \whiletwo, the last one means `$c$ terminates with value $v$ and output store $s$'. The third rule for the allocation operator \& allows the choice of an \emph{arbitrary} fresh location $l$, making writer transitions {non\-de\-ter\-mi\-nis\-tic}. Reader transitions are (partial) deterministic: there might be no transition $p,s\to c$ for given $p$ and $s$ if expression evaluation fails or gives the wrong type of value, e.g.~in the rule for $\mathsf{while}$.

\begin{figure}
  \begin{gather*}
    \inference{}{p\seqcomp q,s
      \to \run{s}{p}\seqcomp q}
    \qquad
    \inference{}{\mathsf{proc}~p,s \to \ret_{p,s}}
    \qquad
    \inference{}{\&p,s \to \&\run{s}{p}}
    \\[1ex]
    \inference{}{\mathsf{skip},s \to \ret_{s}}
    \qquad
    \inference{}{e \ass  p,s \to e \ass \run{s}{p}}
    \qquad
    \inference{\oname{eev}_{\Trs}(e,s) = 0}{\mathsf{while}~e~p,s \to \ret_{s}}
    \\[1ex]
    \inference{\oname{eev}_{\Tr}(e,s) = n \quad n \not = 0}
    {\mathsf{while}~e~p,s \to
      s.\run{s}{p; \mathsf{while}~e~p}
    }
    \qquad
    \inference{\oname{eev}_{\Tr}(e,s) =
      0}{\mathsf{if}~e~\mathsf{then}~p~\mathsf{else}~q,s \to
      s.\run{s}{q}}
    \\[1ex]
    \inference{\oname{eev}_{\Tr}(e,s) = n \quad n \not = 0}
    {\mathsf{if}~e~\mathsf{then}~p~\mathsf{else}~q,s \to
      s.\run{s}{p}
    }
    \qquad
    \inference{\oname{eev}_{\Tr}(e,s) = v \quad v \in Loc + \mathbb{Z}}
    {\mathsf{expr}~e,s \to \ret_{v,s}}
    \qquad
    \inference{\oname{eev}_{\Tr}(e,s) = p}
    {\mathsf{expr}~e,s \to s.\run{s}{p}}
  \end{gather*}

\vspace{1ex}
~\dotfill~
\vspace{1ex}
  \begin{gather*}
    \inference{p,s\to c}{\run{s}{p} \to c}
    \qquad
    \inference{}{\ret_{v,s} \downarrow v,s}
    \qquad
\inference{}{\ret_{s} \downarrow s}
    \qquad
    \inference{}{s.c \xto{s} c}
    \\[1ex]
    \inference{c \xto{s} d}{\&c \xto{s} d}
    \qquad
    \inference{c \to d}{\&c \to \& d}
    \qquad
    \inference{c \downarrow v,s \quad l \not\in \mathrm{dom}(s)}{\&c \downarrow
      l,s[l \mapsto v]}
    \\[1ex]
    \inference{c\xto{s} d}{c\seqcomp q \xto{s} d \seqcomp q}
    \qquad
    \inference{c\to d}{c\seqcomp q \to d \seqcomp q}
    \qquad
    \inference{c\downarrow v,s}{c\seqcomp q \xto{s} \run{s}{q}}
 \qquad
    \inference{c\downarrow s}{c\seqcomp q \to \run{s}{q}}
    \\[1ex]
    \inference{c \xto{s} d}{e \ass  c\xto{s} e \ass  d}
    \qquad
    \inference{c \to d}{e \ass  c\to e \ass  d}
    \qquad
    \inference{\oname{eev}_{\Tr}(e,s) = l \quad c \downarrow v,s}{e \ass
      c \downarrow s[l \mapsto v]}
  \end{gather*}
  \caption{Reader semantics of \reflang (top) and writer semantics of \reflang (bottom).}
  \label{fig:comp-rules-ho}
\end{figure}

\begin{example}[Landin's knot]
  \label{ex:landin1}
  Landin's knot~\cite{landin64} is a programming gadget to encode general recursion into languages with higher-order store
  but without explicit recursion primitives. For \reflang, which is an
  untyped language, Landin's knot can be implemented via the simple pattern
  \[
    l \ass \mathsf{proc}~p \seqcomp \mathsf{expr}~{!}l,
  \]
  where $l$ is a location and $p$ is a reader that makes use of the expression
  $\mathsf{expr}~{!}l$. For instance, the reader $l \ass
  \mathsf{proc}(\mathsf{expr}~{!}l) \seqcomp \mathsf{expr}~{!}l$ diverges and the reader
  \[
    l \ass
    \mathsf{proc}(\mathsf{if}~l'~\mathsf{then}~q \seqcomp l' \ass \mathsf{expr}({!}l' \ominus 1) \seqcomp \mathsf{expr}~{!}l~\mathsf{else}~\mathsf{skip})
    \seqcomp l' \ass \mathsf{expr}~ 10 \seqcomp \mathsf{expr}~{!}l
  \]
  acts as an iterator. Parentheses have been added for readability.
\end{example}

\subsection{Program Equivalence in $\boldsymbol{\reflang}$}

Our treatment of program equivalence on \reflang proceeds along the lines of termination equivalence in \whiletwo (\Cref{sec:terminationsemantics}), adapted to higher-order stores.
Let us first recall the standard program equivalence for languages with higher-order store, which is
\emph{contextual equivalence}.

\begin{notation}\label{not:weak-trans-ref}
For $p\in \Tr$, $c,d\in \Co$, $s\in \store(\Tr)$ and $v\in V(\Tr)$ we write
\begin{itemize}
\item $c\To d$ and $c\xTo{s} d$ if there exist $n\geq 0$ and $c=c_0,\ldots,c_n=d\in \Co$ such that, for each $i<n$, either
$c_i\to c_{i+1}$ or $c_i\xto{s_{i+1}} c_{i+1}$ for some $s_{i+1}\in \store(\Tr)$;
\item $c\Downarrow v,s$ if there exists $c'\in \Co$ such that $c\To c' \downarrow v,s$;
\item $c\Downarrow s$ if there exists $c'\in \Co$ such that $c\To c' \downarrow s$;
\item $c\Downarrow$ if either $c\Downarrow v,s$ for some $v\in V(\Tr)$ and $s\in \S(\Tr)$, or $c\Downarrow s$ for some $s\in \S(\Tr)$;
\item $p,s\Downarrow$ if there exists $c\in \Co$ such that $p,s\to c$ and $c\Downarrow$.
\end{itemize}
\end{notation}

\begin{definition} Let $\Sigma$ be the two-sorted signature of \reflang.
\begin{enumerate}
\item A \emph{context} is a $\Sigma$-term $C$ in a single variable `$\cdot$' (its \emph{hole}) that occurs at most once in $C$. For $\mathsf{s},\mathsf{t}\in \{\Trs,\Cos\}$ we write $C_\mathsf{t}[\cdot_\mathsf{s}]$ for a context $C$ of output sort $\mathsf{t}$ with a hole of sort $\mathsf{s}$. Moreover, $C[t]$ denotes the term obtained by substituting a term $t\in \Tr\cup\Co$ of suitable sort for the hole.
\item \emph{Contextual equivalence} $\equiv^{\mathrm{ctx}}$ is the two-sorted equivalence relation on $(\Tr,\Co)$ defined by
\begin{itemize}
\item  $p\equiv^{\ctx}_\Trs q$\; if\; (i) $\forall C_\Trs[\cdot_\Trs].\,\forall s.\, C[p],s\,{\Downarrow}\iff C[q],s\,{\Downarrow}$ \;and\; (ii) $\forall C_\Cos[\cdot_\Trs].\, C[p] \,{\Downarrow}  \iff C[q]\, {\Downarrow}$;
\item $c\equiv^{\ctx}_\Cos d$ \;if\; (i) $\,\forall C_\Trs[\cdot_\Cos].\forall s.\, C[c],s\,{\Downarrow}\iff C[d],s\,{\Downarrow}$ \,and\, (ii) $\forall C_\Cos[\cdot_\Cos].\, C[c] \,{\Downarrow}  \iff C[d]\,{\Downarrow}$.
\end{itemize}
\end{enumerate}
\end{definition}
A standard observation is that contextual equivalence can be characterized alternatively as the greatest congruence relation on the set of programs that is adequate with respect to termination:

\begin{definition}
\begin{enumerate}
\item A (\reflang)-\emph{congruence} is a congruence on the initial algebra $(\Tr,\Co)$ for $\Sigma$, that is, a two-sorted relation $R = (R_{\Trs} \subseteq \Tr \times \Tr,R_{\Cos} \subseteq
  \Co \times \Co)$ compatible with all constructors of \reflang.
\item A two-sorted relation $R = (R_{\Trs} \subseteq \Tr \times \Tr,R_{\Cos} \subseteq
  \Co \times \Co)$ is (\emph{termination-})\emph{adequate} if
  \begin{enumerate}
  \item for all $p,q \in \Tr$ and $s \in \store(\Tr)$, if $R_{\Trs}(p,q)$ then $p,s
    \Downarrow \iff q,s \Downarrow$;
  \item for all $c,d \in \Co$, if $R_{\Cos}(c,d)$ then $c \Downarrow \iff d \Downarrow$.
  \end{enumerate}
\end{enumerate}
\end{definition}

\begin{proposition}\label{prop:cont-eq}
  Contextual equivalence is the greatest adequate congruence.
\end{proposition}

It is worth mentioning that mere termination equivalence is no longer a suitable
program equivalence in \reflang, as it only relates programs that return strictly
equal stores. This is too restrictive in the higher-order setting: the returned stores may contain readers that should themselves be related, but not necessarily be syntactically equal.
However, to reason about contextual equivalence in an efficient manner, we use
the same method as for \whiletwo, namely that of \emph{termination
  simulations}, which is a higher-order variant of the notion introduced in
\Cref{def:term-sim}.

\begin{notation}
  For every relation $R \subseteq X \times X$, we define the relation $V(R)
  \subseteq V(X) \times V(X)$ by
\[ V(R) = \Delta_{Loc}\cup \Delta_\bZ \cup R, \]
and the relation $\store(R) \subseteq \store(X) \times \store(X)$ by
\[\store(R) = \{(s_{1},s_{2}) \mid \mathrm{dom}(s_{1}) =
    \mathrm{dom}(s_{2}) \,\land\, \forall l\in \mathrm{dom}(s_1).\, V(R)(s_{1}(l),s_{2}(l))\}.
\]
\end{notation}

\begin{definition}
  \label{def:termbisimref}
  A two-sorted relation $R = (R_{\Trs} \subseteq \Tr \times \Tr,R_{\Cos} \subseteq
  \Co \times \Co)$ is a
  (\emph{higher-order}) \emph{termination
    simulation} if the following conditions hold for all $p,q\in \Tr$, $c,d\in \Co$, $s\in \S(\Tr)$ and $v\in V(\Tr)$:
  \begin{enumerate}
  \item \label{enum:whileweakbisim1ref} If $R_\Trs(p,q)$ and $p,s\to c$ then there exists $d$ such that $q,s\to d$ and $R_\Cos(c,d)$.
   \item If $R_\Cos(c,d)$ and $c\xto{s} c'$ or $c\to c'$, then there exists $d'$ such that $d\To d'$ and $R_\Cos(c',d')$.
  \item If $R_\Cos(c,d)$ and $c\downarrow s$, then there exists $s'$ such that $d\Downarrow s'$ and $\S(R_\Trs)(s,s')$.
\item If $R_\Cos(c,d)$ and $c\downarrow v,s$, then there exist $v'$, $s'$ such that $d\Downarrow v',s'$ and $V(R_\Trs)(v,v')$ and $\S(R_\Trs)(s,s')$.
  \end{enumerate}
  \emph{Termination similarity}, denoted by ${\preceq} = (\preceq_{\Trs},\preceq_{\Cos})$,
  is the greatest termination simulation.
\end{definition}

We denote the symmetrization of termination similarity by $\equiv \;= (\equiv_{\Trs},\equiv_{\Cos}) =
(\preceq_{\Trs} \cap \succeq_{\Trs},\preceq_{\Cos} \cap \succeq_{\Cos})$, where
$\succeq_{\Trs/\Cos}$ is the converse of $\preceq_{\Trs/\Cos}$. From the definition of similarity, we immediately obtain:

\begin{proposition}[Adequacy]
  \label{prop:refadeq}
  The relation $\equiv$ is adequate.
\end{proposition}

\begin{remark}
  The stricter notion of termination equivalence is recovered by requiring
  $d \Downarrow s$ and $d \Downarrow v,s$ in clauses (3) and (4) of
  \Cref{def:termbisimref}. Termination equivalence is included in $\equiv$.
\end{remark}

We give a few examples of termination-similar terms.

\begin{example}[Skip-ing]
  For any reader $p \in \Tr$, one has $\mathsf{skip} \seqcomp p \equiv_\Trs p$ because the two programs are termination equivalent. The relation $(T_{\Trs},T_{\Cos})$ below and its converse are termination simulations:
  \[
      T_{\Trs}
      =\; \Delta_{\Trs} %
          \cup \{(\mathsf{skip} \seqcomp p, p)\}, \qquad
      T_{\Cos}
      =\; \Delta_{\Cos} \cup \bigcup_{s\in \store(\Tr)}
         \{(\run{s}{\mathsf{skip}} \seqcomp p,c_{s}),(\ret_{s} \seqcomp
         p,c_{s}),(\run{s}{p},c_{s})\},
  \]
\end{example}

\begin{example}
  For any location $l$, one has $l \ass 2 \seqcomp l \ass
    \mathsf{expr}({!l} \oplus {2})
    \;\;\equiv_{\Trs}\;\;
    l \ass 2 \seqcomp l \ass
    \mathsf{expr}({!l} \oplus {!l})$.
\end{example}

\begin{example}
  Given a termination simulation $(T_{\Trs},T_{\Cos})$ and a location $l$,
  one can build a new termination simulation $(Q_{\Trs},Q_{\Cos})$ by setting
  \[
    \begin{aligned}
      Q_{\Trs} =&\; T_{\Trs} \cup \{(l \ass \mathsf{proc}~p, l
                  \ass \mathsf{proc}~q) \mid (p,q) \in T_{\Trs}\} , \\
      Q_{\Cos} =&\;  T_{\Cos} \cup \bigcup_{s \in \store(\Tr), (p,q)\in T_\Trs}\{(l \ass \run{s}{\mathsf{proc}~p},l \ass
                  \run{s}{\mathsf{proc}~q}), (l \ass \ret_{p,s},l \ass
                  \ret_{q,s})  \}.
    \end{aligned}
  \]
  Hence, given any pair $(p,q) \in T_{\Trs}$ of related readers, $l \ass \mathsf{proc}~p
  \leq_{\Trs} l \ass \mathsf{proc}~q$.\lsnote{Assignment wants a reader on the right hand side, right?}
\end{example}

\begin{example}[Landin's knot]
  Recall Landin's knot from \Cref{ex:landin1}. We can see that
  \[
    l \ass
    \mathsf{proc}(\mathsf{expr}~{!}l) ; \mathsf{expr}~{!}l
    \;\;\equiv_\Trs\;\;
    l \ass
    \mathsf{proc}(\mathsf{while}~1~\mathsf{skip}) ; \mathsf{expr}~{!}l
  \]
 because both programs diverge on every input state.
\end{example}

The following theorem is our main technical result on \reflang, and is proved in \Cref{subsec:reflangmodel,sec:proof-ref}.

\begin{theorem}[Compositionality]
  \label{th:finalcong}
  Termination similarity, hence also $\equiv$, forms a congruence.
\end{theorem}
As an important corollary of \Cref{prop:cont-eq}, \Cref{prop:refadeq}
and \Cref{th:finalcong}, we have thus established two-way termination similarity as a sound proof method for contextual equivalence:

\begin{corollary}[Soundness] If $p \equiv_{\Trs} q$ then $p
  \equiv^{\mathrm{ctx}}_{\Trs} q$, and if $c \equiv_{\Cos} d$ then $c
  \equiv^{\mathrm{ctx}}_{\Cos} d$.
\end{corollary}
The presence of a higher-order store complicates the proof of
\Cref{th:finalcong}, to the extent that a standard, language-specific approach
would likely require the development of a tailor-made Kripke/step-indexed logical relation
or a bisimulation-based method as in~\cite{KoutavasWand06}, tasks that are both
technically challenging and laborious. We will instead systematically derive the theorem using the abstract theory of congruence provided by (higher-order) abstract GSOS.

\subsection{Categorical Modelling of $\boldsymbol{\reflang}$}
\label{subsec:reflangmodel}

As for \whiletwo, the syntax and operational semantics \reflang are modelled
over the category $\impcat$. The syntax is given by the polynomial endofunctor $\Sigma$ on $\Set^2$ corresponding to the signature of \reflang:
\begin{equation*}
  \begin{aligned}
    \Sigma_{\Trs}(X)
    =&\;
       \underbrace{1}_{\mathsf{skip}}
       \;+\; \underbrace{\expr \times X_{\Trs}}_{\ass }
       \;+\; \underbrace{\expr \times X_{\Trs}}_{\mathsf{while}}
       \;+\; \underbrace{X_{\Trs} \times X_{\Trs}}_{-;-}
       \;+\; \underbrace{\expr \times X_{\Trs} \times X_{\Trs}}_{\mathsf{if}\text{-statement}}
       \;+\; \underbrace{\expr}_{\mathsf{expr}} \;+\;
       \underbrace{X_{\Trs}}_{\&} \;+\; \underbrace{X_{\Trs}}_{\mathsf{proc}}, \\
    \Sigma_{\Cos}(X)
    =&\; \underbrace{X_{\Trs}\times\store(X_{\Trs})}_{\mathsf{\run{-}{-}}}
       \;+\; \underbrace{\expr \times X_{\Cos}}_{\ass}
       \;+\; \underbrace{X_{\Cos} \times X_{\Trs}}_{-;-} \;+ \underbrace{X_{\Cos}}_{\&}
       +\; \underbrace{\store(X_{\Trs})}_{\ret}
       \;+\; \underbrace{V(X_{\Trs}) \times \store(X_{\Trs})}_{\ret_{-,-}}
       \;+\; \underbrace{\store(X_{\Trs}) \times X_{\Cos}}_{-.-}.
  \end{aligned}
\end{equation*}
The initial algebra for $\mu\Sigma$ is formed by the two-sorted set $(\Tr,\Co)$ of readers and writers.

Recall that the dynamics of \whiletwo corresponds to a coalgebra
\eqref{eq:while2opmodel}. Here, from a coalgebraic standpoint, the presence of a
higher-order store turns the dynamics of \reflang into a
\emph{higher-order} coalgebra~\cite{GoncharovMiliusEtAl23}, that is, a morphism of type $C \to B(C,C)$ where $B\colon \C^\op\times \C\to \C$
is a \emph{bifunctor} of mixed variance. Specifically, the dynamic behaviour of \reflang is modelled by the
bifunctor $B \c (\impcat)^{\opp} \times \impcat \to \impcat$ given as follows:
\begin{equation}
  \label{eq:behreflang}
  B_{\Trs}(X,Y) = (Y_{\Cos}+1)^{\store(X_{\Trs})} \qand
  B_{\Cos}(X,Y)
  = \mathcal{P}((V(Y_{\Trs}) + 1) \times \store(Y_{\Trs}) + Y_{\Cos} \times (\store(Y_{\Trs}) + 1)).
\end{equation}
We regard $B$ as ordered by equality in the reader sort and by pointwise inclusion in the writer sort. The transitions of \Cref{fig:comp-rules-ho} then induce a higher-order $B$-coalgebra
\begin{equation}\label{eq:opmodel-ho}
\gamma\colon (\Tr,\Co)\to ((\Co+1)^{\S(\Tr)}, \mathcal{P}((V(\Tr) + 1) \times \store(\Tr) + \Co \times (\store(\Tr) + 1))) = B((\Tr,\Co),(\Tr,\Co))
\end{equation}
where the reader component is given by $\gamma^\Trs(p)(s) = c$ if $p,s\to c$ and $\gamma^\Trs(p)(s)=*$ if no such transition exists, and the writer component is defined as follows for $d\in \Co$, $s\in \S(\Tr)$ and $v\in V(\Tr)$,
\[ (d,*)\;/\; (d,s)\;/\; (*,s) \;/\; (v,s) \in \gamma^\Cos(c) \quad\iff\quad c\to d\;/\; c\xto{s} d\;/\; c\downarrow s\;/\; c\downarrow v,s.   \]
The operational rules of \reflang can be modelled using a higher-order version of the notion of GSOS law. A (\emph{$0$-pointed}) \emph{higher-order
  GSOS law}~\cite{GoncharovMiliusEtAl23} of $\Sigma\colon \C\to \C$ over $B\colon \C^\opp\times \C\to \C$ is a family
\begin{align}\label{eq:ho-gsos-law}
    \rho_{X,Y} \c \Sigma (X \times B(X,Y))\to B(X, \Sigma^\star (X+Y))\qquad (X,Y\in \C)
  \end{align}
of morphisms in $\C$ dinatural in $X \in \C$ and natural in $Y\in \C$. For \reflang we consider the higher-order GSOS law \eqref{eq:ho-gsos-law} of the syntax functor $\Sigma$ of the language over the behaviour functor \eqref{eq:behreflang} that, as in the first-order case, encodes the rules of the language. For instance, for each $X\in \Set^2$ the  map
\[
\rho_{X,Y}^\Cos\colon \Sigma_\Cos(X\times B(X,Y)) \to \mathcal{P}((V(\Sigmas_\Trs(X+Y)) + 1) \times \store(\Sigmas_\Trs(X+Y)) + \Sigmas_\Cos(X+Y) \times (\store(\Sigmas_\Trs(X+Y)) + 1))
 \]
encodes the behaviour of the allocation operator $\&\colon \Cos\to \Cos$ as follows for $c\in X_\Cos$, $U\in B_\Cos(X,Y)$:
\[
 \rho^{\Cos}_{X,Y}(\&(c,U)) =
        \{(\&d,s) \mid (d,s) \in U\} \cup \{\&d \mid d \in U\} \cup
        \{(l,s[l \mapsto v]) \mid (v,s) \in U \land l \not\in \mathrm{dom}(s)\}.
\]
See \ifarx{Appendix~\ref{app:gsos-laws}}{\cite[App.~A]{stateful25-arxiv}} for the full definition of the law for \reflang. The (di)naturality of the higher-order GSOS law $\rho$ results from the operational
semantics being parametrically polymorphic w.r.t. the choices of $X$ and $Y$.
In particular, it is polymorphic on the set of readers $X_{\Trs}$ in $\store(X_{\Trs})$.

The definition of the operational model for first-order GSOS laws extends to the present
higher-order setting. Specifically, the \emph{operational model} of $\rho$
\eqref{eq:ho-gsos-law} is the higher-order coalgebra
$\gamma \c \mu\Sigma \to B(\mu\Sigma,\mu\Sigma)$
whose structure is the unique morphism making the diagram below commute:

\begin{equation}\label{eq:ho-gsos-opmodel}
\begin{tikzcd}[column sep=5em]
\Sigma(\mu\Sigma) \ar{rr}{\iota} \ar{d}[swap]{\Sigma\langle \id,\, \gamma\rangle} & & \mu\Sigma \ar[dashed]{d}{\gamma} \\
\Sigma(\mu\Sigma\times B(\mu\Sigma,\mu\Sigma)) \ar{r}{\rho_{\mu\Sigma,\mu\Sigma}} & B(\mu\Sigma,\Sigmas(\mu\Sigma+\mu\Sigma)) \ar{r}{B(\id,\,\hat\ini\comp \Sigmas \nabla)} &  B(\mu\Sigma,\mu\Sigma)
\end{tikzcd}
\end{equation}
In the case of \reflang this yields precisely the higher-order coalgebra \eqref{eq:opmodel-ho} that runs program terms.

\subsection{Program Equivalence in $\boldsymbol{\reflang}$ via Higher-Order Abstract GSOS}\label{sec:proof-ref}
Similar to first-order abstract GSOS, a general theory of congruence of
coalgebraic similarity is available in higher-order abstract GSOS. We adapt the
notion of AOS (\Cref{def:aos}) in the base category $\Set^{T}$ (for a set $T$) as follows:

\begin{definition}
A \emph{higher abstract operational setting} (\emph{HAOS})
$\O=(\Sigma,\ol{\Sigma},B,\barB ,\rho,\gamma,\tilde{\gamma})$ is given by
\begin{itemize}
\item a polynomial functor $\Sigma\colon \Set^T\to \Set^T$ with its {canonical} relation lifting $\ol{\Sigma}$;
\item an ordered bifunctor $(B,\preceq)\c (\Set^T)^{\opp} \times \Set^T\to \Set^T$ with a relation lifting $\barB$;
\item a higher-order GSOS law $\rho$ of $\Sigma$ over $B$;
\item the operational model $(\mu\Sigma,\gamma)$ of $\rho$ with a weakening $\tilde{\gamma}$.
\end{itemize}
\end{definition}
Here, a \emph{weakening} $\tilde{\gamma}\colon \mu\Sigma\to B(\mu\Sigma,\mu\Sigma)$ of the higher-order coalgebra $\gamma$ is meant to be one with respect to the {endo}functor $B(\mu\Sigma,-)\colon \Set^T\to \Set^T$ and its lifting $\barB(\Delta_{\mu\Sigma},-)\colon \Rel(\Set^T)\to \Rel(\Set^T)$ as in \Cref{def:coalg-sim}. The compositionality theorem for first-order abstract GSOS (\Cref{thm:congruence-abstract-gsos}) then generalizes accordingly, and is another special case of~\cite[Cor.~VIII.7]{UrbatTsampasEtAl23b}. Its proof in \emph{op.~cit.}~is based on a categorical abstraction of \emph{Howe's method}~\cite{Howe89,Howe96}, a standard technique for establishing congruence of (bi)similarity for higher-order languages.

\begin{theorem}[Compositionality]\label{thm:congruence-ho-abstract-gsos}
Suppose that $\O=(\Sigma,\ol{\Sigma},B,\barB ,\rho,\gamma,\tilde{\gamma})$ is a HAOS such that
\begin{enumerate}
\item\label{ho-comp-cond1} $\barB$ is up-closed and satisfies $\Delta_{B(X,Y)}\seq\barB(\Delta_X,\Delta_Y)$ for all $X,Y\in \Set^T$ and
\[ \barB(R,S)\cdot \barB(\Delta_X,T)\seq \barB(R,S\bullet T) \quad \text{for all $R\monoto X\times X$ and $S,T\monoto Y\times Y$};\]
\item\label{ho-comp-cond2} $\rho$ is liftable;
\item\label{ho-comp-cond3} $(\mu\Sigma,\ini,\tilde{\gamma})$ is a {higher-order} lax
  $\rho$-bialgebra.
\end{enumerate}
Then the $\barB(\Delta_{\mu\Sigma},-)$-similarity relation on $(\mu\Sigma,\gamma,\tilde{\gamma})$ is a $\ol{\Sigma}$-congruence on the initial algebra $(\mu\Sigma,\ini)$.
\end{theorem}
Here \ref{ho-comp-cond2} means that for each $R\monoto X\times X$ and $S\monoto Y\times Y$ the map $\rho_{X,Y}$ is a relation morphism from $\ol{\Sigma}(R\times \barB(R,S))$ to $\barB(R,\barSigmas(R+S))$, and \ref{ho-comp-cond3} means that the diagram \eqref{eq:ho-gsos-opmodel} commutes laxly when $\gamma$ is replaced with $\tilde{\gamma}$, that is, $B(\id,\hat\iota \comp \Sigmas\nabla) \comp
  \rho_{\mu\Sigma,\mu\Sigma} \comp \Sigma\langle \id,\tilde{\gamma}\rangle \preceq
  \tilde{\gamma} \comp \iota$.

The compositionality of termination similarity in \reflang (\Cref{th:finalcong}) is an instance of the compositionality result for higher-order abstract GSOS: we apply \Cref{thm:congruence-ho-abstract-gsos} to
$ \O=(\Sigma,\ol{\Sigma},B,\barB ,\rho,\gamma,\tilde{\gamma})$
where $\rho$ is the higher-order GSOS law for \reflang, the relation lifting $\barB$ of $B$ \eqref{eq:behreflang} is defined by
\begin{align*}
(\barB(R,S))_\Trs &= \{ (f_1,g_1) \mid \forall s_1,s_2.\, \store(R_{\Trs})(s_{1},s_{2}) \implies f_1(s_1)=f_2(s_2)=\star \vee S_\Cos(f_1(s_1),f_2(s_2)) \},\\
(\barB(R,S))_\Cos &= \vec{\Pow}((V(S_\Trs)+\Delta_1)\times \S(S_\Trs) + S_\Cos \times (\S(S_\Trs)+\Delta_1)),
\end{align*}
and $\tilde{\gamma}\colon (\Tr,\Co)\to B((\Tr,\Co),(\Tr,\Co))$ is the weakening of $\gamma$ given by \Cref{not:weak-trans-ref}:
\[ \tilde{\gamma}^\Trs = \gamma^\Trs \qand \tilde{\gamma}^\Cos(c)= \{ d \mid c\To d \} \cup \{ (d,s) \mid c\To d,s \} \cup \{ s\mid d\Downarrow s\} \cup \{ (v,s) \mid c\Downarrow v,s \}.  \]
Observe that $\barB(\Delta_{\mu\Sigma},-)$-similarity on the operational model $((\Tr,\Co),\gamma,\tilde{\gamma})$ is termination similarity (\Cref{def:termbisimref}). It is then a matter of routine calculations, following the definitions of the ingredients of $\O$, to verify that the conditions \ref{ho-comp-cond1}--\ref{ho-comp-cond3} of \Cref{thm:congruence-ho-abstract-gsos} hold. Once again, \ref{ho-comp-cond2} boils down to parametric polymorphism of the rules of \reflang, and \ref{ho-comp-cond3} to the rules being sound for weak transitions.

\subsection{Towards Variable Binding and Higher-Order Features}\label{sec:higher-higher-order}

The reader-writer approach to the operational semantics of stateful languages also
applies to higher-order languages that feature higher-order functions, variable binding and
substitution.
More so, at first glance, this combination of reader-writer semantics with
higher-order features seems to be compatible with higher-order
abstract GSOS (as its constituents demonstrably are). We shall now briefly sketch the
key ideas of a reader-writer, call-by-name, untyped $\lambda$-calculus with higher-order store (which we call $\lambda$-\reflang) in higher-order abstract GSOS, by extending \reflang accordingly.

The \emph{reader} syntax of $\lambda$-\reflang extends that of \reflang by
adding variables $x$, $\lambda$-abstractions $\lambda x.p$ and applications
$p\;q$. On the side of \emph{writers}, we add ``ongoing'' applications, i.e.
expressions of the form $c\;q$. The operational semantics of the new constructs
are as follows:
\vspace{-0.03cm}
\[
  \begin{gathered}
    \inference{}{p\;q,s \to \run{s}{p}\;q}\qquad
    \inference{c \to d}{c\;q \to d\;q}\qquad
\inference{c \xto{s} d}{c\;q \xto{s} d\;q}\qquad
    \inference{}{\run{s}{\lambda x.p}\;q \xto{s} \run{s}{p[q/x]}}
\end{gathered}
\]

The operational semantics of $\lambda$-\reflang live in $(\Set^{\fset})^{2}$,
the category of two-sorted (covariant) presheaves over the category $\fset$ of finite cardinals and functions (representing untyped variable
contexts). Modelling syntax for languages with variable binding in such presheaf categories is standard, following Fiore et
al.~\cite{FiorePlotkinEtAl99}. In this new setting, an object $X \in
(\Set^{\fset})^{2}$ is a pair of presheaves $(X_{\Trs},X_{\Cos})$ in $\Set^\fset$.
Specifically for $\lambda$-\reflang, the object of terms would be $(\Tr,\Co)$ where
 $\Tr(n)$ and $\Co(n)$ are the sets of readers and writers in the variable context $x_1,\ldots,x_n$.

Implementing $\lambda$-\reflang as a higher-order GSOS law
requires modelling $\lambda$-abstractions as functions on readers. The behaviour
corresponding to $\lambda$-abstractions should thus be given by an exponential
$Y_{\Trs}^{X_{\Trs}}$ (We
  omit some technical details
  of the behaviour, such as a reader-reader substitution structure of terms~\cite{GoncharovMiliusEtAl23}). The semantics of
$\lambda$-abstractions is given by a reader-labelled transition
system, e.g.
\[
  \inference{}{\lambda x.p \xto{q} p[q/x]}
  \qquad
  \inference{p \xto{r} q}{\run{s}{p} \xto{r} \run{s}{q},s} \qquad
  \inference{c \xto{q,s} p'}{c\;q \xto{s} p'}.
\]
In the above, $\lambda x.p$ behaves as a function mapping a
reader $q$ to $p[x/q]$. Under this perspective, which is standard in
higher-order abstract GSOS~\cite[\textsection 5]{GoncharovMiliusEtAl23},
$\beta$-reduction in $\lambda$-\reflang is implemented in terms of the two
rules for resp. $\run{s}{-}$ and application. The exact details of the above
semantics, as well as the development of a theory of compositionality for $\lambda$-\reflang and similar languages via the methods provided by higher-order abstract GSOS, are left as a prospect for future work.

\section{Conclusion and Future Work}

We have presented a systematic approach to the operational semantics of
stateful languages, based on the formal distinction
between readers and writers. This approach is capable of tapping into
the powerful theory of (higher-order) abstract GSOS, refuting the accepted
surmise that stateful languages are not compatible with the framework. Taking
advantage of this fact, we derive efficient reasoning techniques for
program equivalence for both first-order and higher-order store. The extension to fully fledged higher-order languages as outlined in \Cref{sec:higher-higher-order} is a natural next step.

Our approach notably does not treat the semantics of state in terms of the \emph{state
monad}, as is often the case in works about program equivalence on
the $\lambda$-calculus with algebraic
effects~\cite{Dal-LagoGavazzoEtAl17,PlotkinPower01,JohannSimpsonEtAl10}. Even
though the distinction between readers and writers alludes to the two components
of the state monad, it is currently unclear if there is some formal
connection to the latter. However, it is worth
pointing out that the divergence of
\citet{Abou-SalehPattinson11,Abou-SalehPattinson13} from the Turi-Plotkin
framework can be partly attributed to the non-commutativity of the state
monad. Now, with a satisfactory abstract GSOS based approach to state, revisiting abstract
GSOS in Kleisli categories of \emph{commutative} monads, e.g. the power set monad for
trace semantics, has become a possibility.

\begin{acks}
Sergey Goncharov and Stelios Tsampas acknowledge funding by the Deutsche Forschungsgemeinschaft (DFG, German Research Foundation) - project numbers 527481841.
Stefan Milius was supported by Deutsche Forschungsgemeinschaft (DFG, German Research Foundation) – project number 517924115. %
Lutz Schr\"oder was supported by Deutsche Forschungsgemeinschaft (DFG, German Research Foundation) – project number 531706730. %
Henning Urbat is supported by Deutsche Forschungsgemeinschaft (DFG, German Research Foundation) -- project number 470467389.
\end{acks}

\renewcommand{\c}{\cedilla}

\bibliography{mainBiblio}
\ifarx{
\clearpage
\appendix
\section{GSOS Laws}\label{app:gsos-laws}
We give the full GSOS laws for the languages \whiletwo, $\L^2$ and \reflang.

\subsection*{GSOS Law for $\boldsymbol{\whiletwo}$}
The language \whiletwo is modelled by a GSOS law $\rho$ of the signature $\Sigma$ (\Cref{sec:whiletwo}) over the behaviour functor $B_0$ \eqref{eq:beh-whiletwo}. The component $\rho_{0,X}=(\rho_0)_X$ at $X\in \Set^2$ is given by the maps
\[
\rho^\Trs_{0,X} \colon \Sigma_\Trs(X\times B_0X)\to (\Sigmas_\Cos X)^\store \qqand
\rho^\Cos_{0,X} \colon \Sigma_\Cos(X\times B_0X)\to \Sigmas_\Cos X\times \S + \Sigmas_\Cos X +\store
\]
defined as follows for $p,q\in X_\Trs$, $f,g\in X_\Cos^\S = (B_0)_\Trs X$, $c\in X_\Cos$, $u\in X_\Cos\times \S + X_\Cos + \S = (B_0)_\Cos X$, $s\in \S$:
\begin{align*}
\rho^{\Trs}_{0,X}(\mathsf{skip}) =&\; \lambda s.\ret_s\\
\rho^{\Trs}_{0,X}(x\ass  e) =&\;
  \lambda s. \ret_{s[x \ass  \oname{eev}(e,s)]}\\
 \rho^{\Trs}_{0,X}(\mathsf{while}\,e\,(p,f)) =&\;
  \lambda s.\begin{cases}
    s.\run{s}{p\seqcomp \mathsf{while}\,e\,p}
    &\text{if $\oname{eev}(e,s) \not = 0$} \\
    \ret_s & \text{if $\oname{eev}(e,s) = 0$}
  \end{cases}\\
 \rho^{\Trs}_{0,X}((p,f)\seqcomp\,(q,g) =&\;
  \lambda s. \run{s}{p}\seqcomp q \\
 \rho^{\Cos}_{0,X}(\run{s}{p,f}) =&\; f(s) \\
 \rho^{\Cos}_{0,X}(\ret_s) =&\; s\\
 \rho^{\Cos}_{0,X}(s.(c,u)) =&\; (c,s)\\
 \rho^{\Cos}_{0,X}((c,u)\seqcomp\, (q,f)) =&\;
  \begin{cases}
c'\seqcomp q & \text{if } $u=c'$\\
((c' \seqcomp q),s) & \text{if } $u=(c',s)$\\
([q]_{s},s) & \text{if } $u=s$.
\end{cases}
\end{align*}

\subsection*{GSOS Law for $\L^2$}
The language $\L^2$ is modelled by a GSOS law $\rho=\rho^{\L^2}$ of $\Sigma$ (\Cref{not:sig-l2}) over $B$ (\Cref{not:stateful-sos}). (Since $\L^2$ is a deterministic language, we could also first model it via a GSOS law $\rho_0^{\L^2}$ of $\Sigma$ over $B_0$ and then extend it to one over $B$ as in \Cref{rem:gsos-law-nondet}, which would give the same $\rho^{\L^2}$.)

The component $\rho_X$ at $X\in \Set^2$ is given by the maps
\[
\rho^\Trs_{X} \colon \Sigma_\Trs(X\times BX)\to (\Sigmas_\Cos X)^\store \qqand
\rho^\Cos_{X} \colon \Sigma_\Cos(X\times BX)\to \Pow(\Sigmas_\Cos X\times \S + \Sigmas_\Cos X +\store)
\]
defined as follows for $p_i\in X_\Trs$, $f_i\in X_\Cos^\store$, $c_i\in X_\Cos$, $U_i\in \Pow (X_\Cos\times \store + X + S)$, $s\in \S$:

\medskip\noindent\underline{Definition of $\rho^\Trs_X( \f((p_1,f_1),\ldots,(p_n,f_n) ))$:}\\

\noindent Let $\f$ be an $n$-ary operator. If $\f$ is passive and $\L$ contains the rule
  \begin{equation*}\label{eq:sos-rule}
      \inference{}{\f(x_1,\ldots,x_n),s \to t,s'} \quad/\quad  \inference{}{\f(x_1,\ldots,x_n),s \downarrow s'}
\end{equation*}
then
\[ \rho^\Trs_X( \f((p_1,f_1),\ldots,(p_n,f_n) ))(s) = s'.[ t[p_1/x_1],\ldots, p_n/x_n] ]_{s'} \quad/\quad \ret_{s'}.  \]
If $\f$ is active with receiving position $j$, then
\[ \rho^\Trs_X( \f((p_1,f_1),\ldots,(p_n,f_n) ))(s) = \barf(p_1,\ldots, [p_j]_s,\ldots, p_n).  \]

\medskip\noindent\underline{Definition of $\rho^\Trs_X( \barf((p_1,f_1),\ldots, (c_j,U_j), \ldots (p_n,f_n) ))$:}\\

\noindent Let $\f$ be an $n$-ary active operator with receiving position $j$. Then
\[\rho^\Trs_X( \barf((p_1,f_1),\ldots, (c_j,U_j), \ldots (p_n,f_n) )) \seq \Sigmas_\Cos X \times \store + \Sigmas_\Cos X + \store \]
contains the following elements:
\begin{itemize}
\item all $\barf(p_1,\ldots,d_j,\ldots,p_n)\in \Sigmas_\Cos X$ where $d_j\in U_j$;
\item all $(\barf(p_1,\ldots,d_j,\ldots,p_n),s)\in \Sigmas_\Cos X\times \store $ where $(d_j,s)\in U_j$;
\item all $([t[p_1/x_1,\ldots,p_{j-1}/x_{j-1},p_{j+1}/x_{j+1},\ldots, p_n/x_n]]_{s''},s'')\in \Sigmas_\Cos X\times \store$ where
\[\inference{\rets{x_{j},s}{s'}}{\goes{\f(x_{1},\dots,x_{n}),s}{t,s''}} \]
is a rule of $\L$ and  $s'\in U_j$ (recall that $t\in \Sigmas(\{x_1,\ldots,x_n\}\smin \{x_j\})$);
\item all $s''\in \S$ where
\[\inference{\rets{x_{j},s}{s'}}{\rets{\f(x_{1},\dots,x_{n}),s}{s''}} \]
is a rule of $\L$ and  $s'\in U_j$;
\end{itemize}
\medskip\noindent\underline{Definition of $\rho^\Trs_X( [(p_1,f_1)]_s )$:}\\
\[ \rho^\Trs_X( [(p_1,f_1)]_s ) = \{ f_1(s) \}. \]
\medskip\noindent\underline{Definition of $\rho^\Trs_X( \ret_s )$:}\\
\[ \rho^\Trs_X( \ret_s ) = \{ s \}. \]
\medskip\noindent\underline{Definition of $\rho^\Trs_X( s.(c_1,U_1) )$:}\\
\[ \rho^\Trs_X( s.(c_1,U_1) ) = \{ (c_1,s) \}. \]

\subsection*{Higher-Order GSOS Law for \reflang}
The language \reflang is modelled by a higher-order GSOS law
\[
    \rho_{X,Y}  \colon \Sigma (X \times B(X,Y))\to B(X, \Sigmas (X+Y))\qquad (X,Y\in \Set^{2})
\]
of the syntax functor $\Sigma$ for \reflang over the behaviour bifunctor $B$ of \eqref{eq:behreflang}. The two maps
\begin{align*}
\rho_{X,Y}^\Trs\colon & \Sigma_\Trs(X\times B(X,Y)) \to (\Sigmas_\Cos(X+Y)+1)^{\S(X_\Trs)},\\
\rho_{X,Y}^\Cos\colon& \Sigma_\Cos(X\times B(X,Y)) \to \mathcal{P}((V(\Sigmas_\Trs(X+Y)) + 1) \times \store(\Sigmas_\Trs(X+Y)) + \Sigmas_\Cos(X+Y) \times (\store(\Sigmas_\Trs(X+Y)) + 1))
\end{align*}
are defined as follows, where $p,q\in X_\Trs$, $c,d\in X_\Cos$, $f,g\in Y_\Cos^{\S(X_\Trs)}$, $U\in B_\Cos(X,Y)$, $e\in \expr$, $s\in \S(X_\Trs)$, $v\in V(X_\Cos)$:

  \begin{equation*}
    \begin{aligned}
      & \rho_{X,Y}^{\Trs}(\mathsf{skip}) = \lambda s. \mathsf{ret}_{s} \\
      & \rho_{X,Y}^{\Trs}(e\ass (p,f)) = \lambda s.e \ass \run{s}{p} \\
      & \rho_{X,Y}^{\Trs}(\mathsf{while}~e~(p,f))(s) =
        \begin{cases}
        s.\run{s}{p; \mathsf{while}~e~p} & \text{if }
         \oname{eev}_{X_{\Trs}}(e,s) = n \not = 0,\\
        \ret_{s} & \text{if }  \oname{eev}_{X_{\Trs}}(e,s) = 0,\\
  	    * & \text{otherwise}
\end{cases}\\
      & \rho_{X,Y}^{\Trs}(\mathsf{if}~e~\mathsf{then}~(p,f)~\mathsf{else}~(q,g))(s) =
\begin{cases}
s.\run{s}{p} & \text{if } \oname{eev}_{X_{\Trs}}(e,s) = n \not = 0,\\
        s.\run{s}{q} & \text{if } \oname{eev}_{X_{\Trs}}(e,s) = 0 \\
        * & \text{otherwise}
\end{cases}\\
      & \rho_{X,Y}^{\Trs}((p,f) \seqcomp (q,g)) =
        \lambda s.\run{s}{p} \seqcomp q\\
      & \rho_{X,Y}^{\Trs}(\&(p,f)) =
        \lambda s.\&\run{s}{p}\\
      & \rho_{X,Y}^{\Trs}(\mathsf{expr}~e)(s) =
        \begin{cases}
        \ret_{v,s} & \text{if } \oname{eev}_{X_{\Trs}}(e,s) = v
        \land v \in Loc + \mathbb{Z} \\
        s.\run{s}{p} & \text{if } \oname{eev}_{X_{\Trs}}(e,s) = p\\
        * & \text{otherwise}
          \end{cases}\\
      & \rho_{X,Y}^{\Trs}(\mathsf{proc}~(p,f)) = \lambda s.\ret_{p,s} \\
      & \rho_{X,Y}^{\Cos}(\run{s}{p,f}) = f(s), \\
        &\rho_{X,Y}^{\Cos}(\ret_s) = \{s\}, \\
      &\rho_{X,Y}^{\Cos}(\ret_{v,s}) = \{(v,s)\}, \\
      & \rho_{X,Y}^{\Cos}(s.(c,U)) = \{(c,s)\} \\
      & \rho_{X,Y}^{\Cos}(e:= (c,U)) =
        \{(e \ass d,s) \mid (d,s) \in U\} \cup \{(e \ass d) \mid d \in U\} ~~{\cup}
      \\
      & \qquad \{(s[l \mapsto v]) \mid (v,s) \in U
        \land \oname{eev}_{X_{\Trs}}(e,s) = l\} \\
      & \rho_{X,Y}^{\Cos}((c,U) \seqcomp (q,g)) =
        \{((d \seqcomp q),s) \mid (d,s) \in U\} \cup \{(d \seqcomp q) \mid d \in U\}
        \cup \{\run{s}{q} \mid (v,s) \in U\} \cup \{\run{s}{q} \mid s \in U\} \\
      & \rho_{X,Y}^{\Cos}(\&(c,U)) =
        \{\&d,s \mid (d,s) \in U\} \cup \{\&d \mid d \in U\} \cup
        \{(l,s[l \mapsto v]) \mid (v,s) \in U \land l \not\in \mathrm{dom}(s)\} \\
    \end{aligned}
  \end{equation*}

\section{Omitted Proofs}

\begin{lemma}\label{lem:weakening}
For every $B$-coalgebra $(X,\chi)$ the coalgebra $\tilde{\chi}_\trc$ is a weakening of $\chi$ w.r.t.~$\barB_\trc$.
\end{lemma}

\begin{proof}
For every relation $R\monoto X\times X$ we need to show that $R$ satisfies the conditions (1)--(4) of \Cref{def:trace-sim} iff $R$ satisfies the conditions (1) and (2')--(4') where (i') is (i) with $\to$, $\downarrow$ replaced by $\xTo{1}$, $\Downarrow^1$. The `if' direction is clear since $c\to c'$ implies $c\xTo{1} c'$ (and similarly for the other types of weak transitions). Let us consider the `only if' direction:
\begin{enumerate}
\item[(2')] Let $R_\Cos(c,d)$ and $c\xTo{1} c'$. Then $c=c_0\to c_1\to \cdots \to c_n=c'$ for some $c_i$. By (2) we have $d=d_0\xTo{1} d_1\xTo{1} \cdots \xTo{1} d_n=d'$ for some $d_i$ such that $R_\Cos(c_i,d_i)$. Therefore $d\xTo{1} d'$ and $R_\Cos(c',d')$.
\item[(3')] Let $R_\Cos(c,d)$ and $c\xTo{1} c',s$. Then $c\xTo{1} c_1$ and $c_1\to c,s$ for some $c_1$. By (2') we have $d\xTo{1} d_1$ such that $R_\Cos(c_1,d_1)$. By (3) we have $d_1\xTo{1} d',s$ such that $R_\Cos(c',d')$. Thus $d\xTo{1} d',s$ and $R_\Cos(c',d')$.
\item[(4')] Let $R_\Cos(c,d)$ and $c\Downarrow^1 s$. Then $c\xTo{1} c'$ and $c'\downarrow s$ for some $c'$. By (2') we have $d\xTo{1} d'$ such that $R_\Cos(c',d')$. By (4) we have $d'\Downarrow^1 s$; thus $d\Downarrow^1 s$.\qedhere
\end{enumerate}
\end{proof}

\begin{lemma}\label{lem:weakening-scn}
For every $B$-coalgebra $(X,\chi)$ the coalgebra $\tilde{\chi}_\scn$ is a weakening of $\chi$ w.r.t.~$\barB_\scn$.
\end{lemma}

\begin{proof}
For every relation $R\monoto X\times X$ we need to show that $R$ satisfies the conditions (1)--(4) of \Cref{def:cost-sim} iff $R$ satisfies the conditions (1) and (2')--(4') where (i') is (i) with $\to$, $\downarrow$ replaced by $\xTo{1}$, $\Downarrow^1$.
This is completely analogous to the proof of \Cref{lem:weakening}: just replace the argument for (3') in the `only if' direction by
\begin{enumerate}
\item[(3')] Let $R_\Cos(c,d)$ and $c\xTo{1,s} c'$. Then $c\xTo{1} c_1$ and $c_1\xto{s} c$ for some $c_1$. By (2') we have $d\xTo{1} d_1$ such that $R_\Cos(c_1,d_1)$. By (3) we have $d_1\xTo{1,s'} d'$ for some $d'$ and $s'$ such that $R_\Cos(c',d')$. Thus $d\xTo{1,s'} d'$ and $R_\Cos(c',d')$. \qedhere
\end{enumerate}
\end{proof}

\begin{lemma}\label{lem:weakening-trm}
For every $B$-coalgebra $(X,\chi)$ the coalgebra $\tilde{\chi}_\trm$ is a weakening of $\chi$ w.r.t.~$\barB_\trm$.
\end{lemma}

\begin{proof}
For every relation $R\monoto X\times X$ we need to show that $R$ satisfies the conditions (1)--(3) of \Cref{def:term-sim} iff $R$ satisfies the conditions (1) and (2'),(3') where (i') is (i) with $\to$, $\downarrow$ replaced by $\xTo{2}$, $\Downarrow^2$. Again the `if' direction is trivial. Let us consider the `only if' direction:
\begin{enumerate}
\item[(2')] Let $R_\Cos(c,d)$ and $c\xTo{2} c'$ or $c\xTo{2,s} c'$ Then there exist $c=c_0,\ldots,c_n=c'\in X_\Cos$ such that for for each for $i<n$, either $c_i\to c_{i+1}$ or $c_i\xto{s_{i+1}} c_{i+1}$ for some $s_{i+1}$. By iterated application of (2) we get $d=d_0\xTo{2} d_1 \xTo{2} \cdots \xTo{2} d_n=d'$ where $R_\Cos(c_i,d_i)$, whence $d\xTo{2} d'$ and $R_\Cos(d,d')$.
\item[(3')] Let $R_\Cos(c,d)$ and $c\Downarrow^2 s$. Then $c\xTo{2} c'$ and $c'\downarrow s$ for some $c'$. By (2') we have $d\xTo{2} d'$ such that $R_\Cos(c',d')$. By (3) we have $d'\Downarrow^2 s$; thus $d\Downarrow^2 s$.\qedhere
\end{enumerate}
\end{proof}

\subsection*{Proof of \Cref{prop:trc-sim}}
If $R$ is a trace simulation on a $B_0$-coalgebra and $R_\Cos(c,d)$, then $\trc(c)=\trc(d)$ by definition of trace simulation, so trace similarity implies trace equivalence. Conversely, trace equivalence is readily verified to be trace simulation, so trace equivalence implies trace similarity. \qed

\subsection*{Proof of \Cref{prop:scn-sim}}
If $R$ is a cost simulation on a $B_0$-coalgebra and $R_\Cos(c,d)$, then $\scn(c)=\scn(d)$ by definition of cost simulation, so cost similarity implies cost equivalence. Conversely, cost equivalence is readily verified to be cost simulation, so cost equivalence implies cost similarity.
\qed

\subsection*{Proof of \Cref{prop:trm-sim}}
Let $(X,\chi)$ be a $B_0$-coalgebra. Put $\equiv_\trm \,=\, \preceq_\trm \cap \succeq_\trm$. If $c\equiv_\trm d$, then $c\preceq_\trm d$ and $d\preceq_\trm c$ and so $c\Downarrow^2 s$ iff $d\Downarrow^2 s$. Therefore, $\equiv_\trm$ is contained in termination equivalence. Conversely, termination equivalence $R\monoto X\times X$ is a symmetric termination simulation and therefore contained in $\equiv_\trm$.
\qed

\subsection*{Proof of \Cref{thm:whiletwo-cong}}
The theorem is a special case of \Cref{thm:l2-comp}.

\subsection*{Proof of \Cref{thm:semantics-pres}}
This theorem is a special case of \Cref{prop:trace-pres}.

\subsection*{Proof of \Cref{thm:l2-comp}}
We treat the cases of trace, cost, and termination semantics separately.

\paragraph*{Trace semantics} We consider the AOS
\[ \O_\trc = (\Sigma,\ol{\Sigma},B,\barB_\trc,\rho^{\L^2},\gamma^{\L_2},\tilde{\gamma}^{\L^2}_\trc) \]
where $\barB_\trc$ is the relation lifting of $B$ given by \eqref{eq:barB-trc}, and $\tilde{\gamma}^{\L^2}_\trc$ is the weakening of $\gamma^{\L^2}$ given by \eqref{eq:weakening-trc}. Recall from \Cref{prop:trc-sim} that $\barB_\trc$-similarity on $(\mu\Sigma,\gamma^{\L^2},\tilde{\gamma}^{\L^2}_\trc)$ is trace equivalence. By instantiating the congruence theorem for abstract GSOS (\Cref{thm:congruence-abstract-gsos}) to $\O_\trc$, we get:

\begin{theorem}\label{thm:cong-trc-l2}
Trace equivalence is a congruence on the operational model of $\L^2$.
\end{theorem}

\begin{proof}
We only need to verify that $\O_\trc$ satisfies the conditions (1)--(3) of \Cref{thm:congruence-abstract-gsos}. To simplify the notation, we put
\[ \barB=\barB_\trc,\qquad \rho=\rho^{\L^2},\qquad \gamma=\gamma^{\L^2},\qquad \tilde{\gamma}=\tilde{\gamma}^{\L^2}_\trc. \]
Moreover, we write $\To$, $\Downarrow$ for the weak transitions $\xTo{1}$, $\Downarrow^1$ from \Cref{not:weaktrans-1}.
\begin{enumerate}
\item $\barB$ is up-closed and laxly preserves composition and identities:
\begin{enumerate}
\item Up-closure of $\barB$ is immediate from up-closure of $\vec{\Pow}$.
\item $\barB$ laxly preserves composition: for every $R,S\monoto X\times X$ in $\Set^2$ we have
\begin{align*}
\barB(R\bullet S) &= \barB(R_\Trs\cdot S_\Trs, R_\Cos \bullet S_\Cos) \\
&= ((R_\Cos\bullet S_\Cos)^\S, \vec{\Pow}((R_\Cos\bullet S_\Cos)\times \Delta_\S + R_\Cos\bullet S_\Cos + \Delta_\S) ) \\
&= ((R_\Cos)^\S\bullet (S_\Cos)^\S, \vec{\Pow}((R_\Cos\times \Delta_\S + R_\Cos + \Delta_\S)\bullet (S_\Cos\times \Delta_\S + S_\Cos + \Delta_\S)  ) \\
&\supseteq ((R_\Cos)^\S\bullet (S_\Cos)^\S, \vec{\Pow}(R_\Cos\times \Delta_\S + R_\Cos + \Delta_\S)\bullet \vec{\Pow}(S_\Cos\times \Delta_\S + S_\Cos + \Delta_\S)  ) \\
&=\barB R \bullet \barB S.
\end{align*}
In the penultimate step, we use that $\vec{\Pow}$ laxly preserves composition, which is immediate.
\item $\barB$ laxly preserves identities: For every $X\in \Set^2$ we have
\begin{align*}
\barB(\Delta_X) &= \barB(\Delta_{X_\Trs}, \Delta_{X_\Cos}) \\
&= ((\Delta_{X_\Cos})^\S, \vec{\Pow}(\Delta_{X_\Cos}\times \Delta_{\S} + \Delta_{X_\Cos} + \Delta_{\S}) ) \\
&= (\Delta_{X_\Cos^\S}, \vec{\Pow}(\Delta_{X_\Cos\times \S + X_\Cos + \S} )) \\
&\supseteq (\Delta_{X_\Cos^\S}, \Delta_{\Pow(X_\Cos\times \S + X_\Cos + \S)} )\\
&= \Delta_{BX}
\end{align*}
In the penultimate step, we use that $\vec{\Pow}$ laxly preserves identity relations, which is immediate.
\end{enumerate}
\item The GSOS law $\rho$ of $\Sigma$ over $B$ is liftable, that is, for each relation $R\monoto X\times X$ in $\Set^2$ the map
\[ \rho_X\colon \ol{\Sigma}(R\times \barB R)\to \barB\ol{\Sigma}^{\star} R \]
is a relation morphism. This boils down to the observation that the rules of $\L^2$ (\Cref{fig:comp-rules-l2}) are parametrically polymorphic. We illustrate the argument for the operator $\barf\colon \Trs\times \cdots \times \Cos \times \cdots \times \Trs \to \Cos$ of $\Sigma$ corresponding to an $n$-ary active operator $\f\in \Theta$ with receiving position $j$.
Consider two elements $\barf((p_1,f_1),\ldots, (c_j,U_j),\ldots, (p_n,f_n))$ and $\barf((p_1',f_1'),\ldots, (c_j',U_j'),\ldots  (p_n',f_n'))$ of $\Sigma_\Cos(X\times BX)$ that are related in $\ol{\Sigma}(R\times \barB R)$; in particular, this means that
\begin{equation}\label{eq:liftable-1}
p_i,p_i'\in X_\Cos \text{ and } R_\Cos(p_i,p_i')\qquad \text{for $i=1,\ldots,n$, $i\neq j$},
\end{equation}
\begin{equation}\label{eq:liftable-2}
U_j,U_j'\in \Pow(X_\Cos\times \store + X_\Cos + \store)\text{ and } (U_j,U_j')\in \vec{\Pow}(R_\Cos\times \Delta_\S + R_\Cos + \Delta_\S).
\end{equation}
We need to show that
\begin{equation}\label{eq:proof-goal} (\rho_X^\Cos(\barf((p_1,f_1),\ldots, (c_j,U_j),\ldots, (p_n,f_n))),\, \rho_X^\Cos(\barf((p_1',f_1'),\ldots, (c_j',U_j'),\ldots, (p_n',f_n'))))\in (\barB \ol{\Sigma}^\star R)_\Cos. \end{equation}
Recall from definition of $\rho=\rho^{\L^2}$ (Appendix~\ref{app:gsos-laws}) that
\[
\rho^\Cos_{X} \colon \Sigma_\Cos(X\times BX)\to \Pow(\Sigmas_\Cos X\times \S + \Sigmas_\Cos X +\store)
\]
sends $\barf((p_1,f_1),\ldots, (c_j,U_j),\ldots,(p_n,f_n))$ to the set
\[\rho_X^\Cos(\barf((p_1,f_1),\ldots, (c_j,U_j),\ldots,(p_n,f_n)))\seq \Sigmas_\Cos X\times \S + \Sigmas_\Cos X +\store\]
containing the following elements:
\begin{enumerate}
\item all $\barf(p_1,\ldots,d_j,\ldots,p_n)\in \Sigmas_\Cos X$ where $d_j\in U_j$;
\item all $(\barf(p_1,\ldots,d_j,\ldots,p_n),s)\in \Sigmas_\Cos X\times \store $ where $(d_j,s)\in U_j$;
\item all $([t[p_1/x_1,\ldots,p_{j-1}/x_{j-1},p_{j+1}/x_{j+1},\ldots, p_n/x_n]]_{s''},s'')\in \Sigmas_\Cos X\times \store$ where
\[\inference{\rets{x_{j},s}{s'}}{\goes{\f(x_{1},\dots,x_{n}),s}{t,s''}} \]
is a rule of $\L$ and  $s'\in U_j$;
\item all $s''\in \S$ where
\[\inference{\rets{x_{j},s}{s'}}{\rets{\f(x_{1},\dots,x_{n}),s}{s''}} \]
is a rule of $\L$ and  $s'\in U_j$;
\end{enumerate}
Analogously for $\barf((p_1',f_1'),\ldots, (c_j',U_j'),\ldots,(p_n',f_n'))$. Now \eqref{eq:proof-goal} easily follows. For instance, if
\[(\barf(p_1,\ldots, d_j,\ldots p_n),s)\in \rho_X^\Cos(\barf((p_1,f_1),\ldots, (c_j,U_j),\ldots,(p_n,f_n))) \]
according to (b), that is, $(d_j,s)\in U_j$, then there exists $(d_j',s)\in U_j'$ with $R_\Cos(d_j,d_j')$ by \eqref{eq:liftable-2}, and so
\[(\barf(p_1',\ldots, d_j',\ldots p_n'),s)\in \rho_X^\Cos(\barf((p_1',f_1'),\ldots, (c_j',U_j'),\ldots,(p_n',f_n'))) \]
by definition of $\rho_X^\Cos$. Moreover, these two expressions are related in  $\ol{\Sigma}^\star_\Cos R \times \Delta_\S$ by \eqref{eq:liftable-1} and the definition of $\ol{\Sigma}^\star$.

Similarly, if
\[ ([t[p_1/x_1,\ldots,\ldots,p_{j-1}/x_{j-1},p_{j+1}/x_{j+1},\ldots,p_n/x_n]]_{s''},s'') \in \rho_X^\Cos(\barf((p_1,f_1),\ldots, (c_j,U_j),\ldots,(p_n,f_n))) \]
according to (c), then $s'\in U_j$ implies $s'\in U_j'$ by \eqref{eq:liftable-2}, and so
\[ ([t[p_1'/x_1,\ldots,\ldots,p_{j-1}'/x_{j-1},p_{j+1}'/x_{j+1},\ldots,p_n'/x_n]]_{s''},s'') \in \rho_X^\Cos(\barf((p_1',f_1'),\ldots, (c_j',U_j'),\ldots,(p_n',f_n')))  \]
by definition of $\rho_X^\Cos$. Moreover, these expressions are related in $\ol{\Sigma}^\star_\Cos R \times \Delta_\S$ by \eqref{eq:liftable-1} and the definition of $\ol{\Sigma}^\star$. The remaining cases are treated analogously.
\item $(\mu\Sigma,\gamma,\tilde{\gamma})$ forms a lax $\rho$-bialgebra. To show this, let us first introduce some terminology. A \emph{weak instance} of a rule $\L^2$ emerges by substituting the variables by terms of suitable sort in $\mu\Sigma$ and the transitions $\to$, $\downarrow$ for writers by corresponding weak transitions $\To$, $\Downarrow$. For example, the weak instances of the rule ($\barf$1) in \Cref{fig:comp-rules-l2} are given by
\begin{equation}\label{eq:l2-rule3-weak}
      \inference{c_j\To d_j}{\barf(p_1,\ldots,c_j,\ldots,p_n) \To \barf(p_1,\ldots, d_j,\ldots, p_n)}.
\end{equation}
where $p_i\in (\mu\Sigma)_\Trs$ for $i\neq j$ and $c_j,d_j\in (\mu\Sigma)_\Cos$. We say that a rule of $\L^2$ is \emph{sound for weak transitions} if all its weak instances are sound in the operational model $\gamma$, that is, if the premise holds then the conclusion holds. Note that every premise-free rule is trivially sound for weak transitions.

Lax commutativity of the diagram \eqref{eq:lax-bialgebra-fo} states precisely that all rules of $\L^2$ are sound for weak transitions. This property of the rules is easily verified:
\begin{itemize}
\item The rule ($\barf$1) is sound for weak transitions: every weak instance \eqref{eq:l2-rule3-weak} is sound by repeated application of ($\barf$1).
\item The rule ($\barf$2) is sound for weak transitions: every weak instance
\[       \inference{c_j\xTo{s} d_j}{\barf(p_1,\ldots,c_j,\ldots,p_n)\xTo{s} \barf(p_1,\ldots, d_j,\ldots, p_n)}  \]
is sound by repeated application of ($\barf$1) and one application of ($\barf$2).
\item The rule ($\barf$3) is sound for weak transitions: every weak instance
\[\inference{c_j\Downarrow s'}{\barf(p_1,\ldots,c_j,\ldots,p_n) \xTo{s''} [t[p_1/x_1,\ldots,p_{j-1}/x_{j-1},p_{j+1}/x_{j+1},\ldots,p_n/x_n]  ]_{s''}} \]
is sound by repeated application of ($\barf$1) and one application of ($\barf$3). (Here, it is crucial that the variable $x_j$ does not appear in $t$, as ensured by the cool format.) Similarly for the rule ($\barf$4).
\item The rule for $[-]_s$ is sound for weak transitions because the weakening does not affect the premise (as it involves a reader term), only the conclusion.
\item All remaining rules are premise-free, hence sound for weak transitions.
\end{itemize}
\end{enumerate}
This concludes the proof.
\end{proof}

\paragraph*{Cost Semantics.} We consider the AOS
\[ \O_\scn = (\Sigma,\ol{\Sigma},B,\barB_\scn,\rho^{\L^2},\gamma^{\L_2},\tilde{\gamma}^{\L^2}_\scn) \]
where $\barB_\scn$ is the relation lifting of $B$ given by \eqref{eq:barB-scn}, and $\tilde{\gamma}^{\L^2}_\scn$ is the weakening of $\gamma^{\L^2}$ given by \eqref{eq:weakening-trc}. Recall from \Cref{prop:scn-sim} that $\barB_\scn$-similarity on $(\mu\Sigma,\gamma^{\L^2},\tilde{\gamma}^{\L^2}_\scn)$ is cost equivalence. By instantiating the congruence theorem for abstract GSOS (\Cref{thm:congruence-abstract-gsos}) to $\O_\scn$, we get:

\begin{theorem}\label{thm:cong-scn-l2}
Cost equivalence is a congruence on the operational model of $\L^2$.
\end{theorem}

\begin{proof}
We verify that $\O_\scn$ satisfies the conditions (1)--(3) of \Cref{thm:congruence-abstract-gsos}. Since $\tilde{\gamma}^{\L^2}_\scn=\tilde{\gamma}^{\L^2}_\trc$, condition (3) has already been shown in the proof of \Cref{thm:cong-trc-l2}. The arguments for (1) and (2) are essentially identical to those in the proof of \Cref{thm:cong-trc-l2}.
\end{proof}

\paragraph*{Termination Semantics.} We consider the AOS
\[ \O_\trm = (\Sigma,\ol{\Sigma},B,\barB_\trm,\rho^{\L^2},\gamma^{\L_2},\tilde{\gamma}^{\L^2}_\trm) \]
where $\barB_\trm$ is the relation lifting of $B$ given by \eqref{eq:barB-trm}, and $\tilde{\gamma}^{\L^2}_\trm$ is the weakening of $\gamma^{\L^2}$ given by \eqref{eq:weakening-trm}. Recall that the $\barB_\trm$-similarity relation on $(\mu\Sigma,\gamma^{\L^2},\tilde{\gamma}^{\L^2}_\scn)$ is given by \emph{termination similarity} $\preceq_\trc$. By instantiating the congruence theorem for abstract GSOS (\Cref{thm:congruence-abstract-gsos}) to $\O_\trm$, we get:

\begin{theorem}\label{thm:cong-trm-l2-sim}
Termination similarity $\preceq_\trm$ is a congruence on the operational model of $\L^2$.
\end{theorem}

\begin{proof}
We verify that $\O_\trm$ satisfies the conditions (1)--(3) of \Cref{thm:congruence-abstract-gsos}. Since $\barB_\trm = \barB_\scn$, conditions (1)  and (2) have already been shown in the proof of \Cref{thm:cong-scn-l2}. For (3), we write
\[ \barB=\barB_\trm,\qquad \rho=\rho^{\L^2},\qquad \gamma=\gamma^{\L^2},\qquad \tilde{\gamma}=\tilde{\gamma}^{\L^2}_\trm, \]
and we let $\To$, $\Downarrow$ denote the weak transitions $\xTo{2}$, $\Downarrow^2$ from \Cref{not:weak-trans-2}. Again, our task is to show that all rules of $\L^2$ are sound for weak transitions. The argument is much analogous to the proof of \Cref{thm:cong-trc-l2}. For example, the rule ($\barf$3) is sound for weak transitions: every weak instance
\[\inference{c_j\Downarrow s'}{\barf(p_1,\ldots,c_j,\ldots,p_n) \xTo{s''} [t[p_1/x_1,\ldots,p_{j-1}/x_{j-1},p_{j+1}/x_{j+1},\ldots,p_n/x_n]  ]_{s''}} \]
is sound by repeated application of ($\barf$1) and ($\barf$2) followed by one application of ($\barf$3). Similar reasoning applies to the other rules.
\end{proof}

Recall from \Cref{prop:trm-sim} that $\preceq_\trm \cap \succeq_\trm$ is termination equivalence. Since congruences are closed under reversal and intersection, the above theorem entails:
\begin{theorem}\label{thm:cong-trm-l2}
Termination equivalence is a congruence on the operational model of $\L^2$.
\end{theorem}

\subsection*{Proof of \Cref{prop:trace-pres}}
It suffices to prove the first equality; the other two are then immediate because cost and termination semantics in both \while and \whiletwo are derived from trace semantics. The key to the proof lies in relating transitions in the operational model \eqref{eq:while-opmodel} of \while to (weak) transitions $\to$, $\xTo{1}$, $\Downarrow^1$ (\Cref{not:weaktrans-1}) in the operational model \eqref{eq:while2opmodel} of \whiletwo. We write $\To$, $\Downarrow$ for $\xTo{1}$, $\Downarrow^1$.
We will show that the following holds for all $p,p'\in \mS$ and $s,s'\in \store$:

\begin{enumerate}
\item\label{statement-1} If $p,s\to p',s'$ in $\gamma^\L$, then there exist $c_0,c_1,c_2,c_3\in (\mu\Sigma)_\Cos$ with $\gamma^{\L^2}$-transitions \[
p,s\to c_0\xTo{s'} c_1,\qquad c_1\To c_2, \qquad p',s'\to c_3\To c_2.\]
\item\label{statement-2} If $p,s\downarrow s'$ in $\gamma^\L$, then there exists $c_0\in (\mu\Sigma)_\Cos$ with $\gamma^{\L^2}$-transitions
\[p,s\to c_0, \qquad c_0\Downarrow s'.\]
\end{enumerate}
Statement \ref{statement-1} is depicted in the diagram below.  The upper arrow is an \while-transition, the remaining arrows are \whiletwo-transitions.
\[
  \begin{tikzcd}[ampersand replacement=\&]
p,s \ar{rr} \ar{dd} \&\& p',s' \ar{d} \\
\&\& c_3 \ar[Rightarrow]{d} \\
c_0 \ar[Rightarrow]{r}{s'} \& c_1 \ar[Rightarrow]{r} \& c_2
\end{tikzcd}
\]
The equality $\trc_\L = \trc^\Trs_{\L^2}$ is an easy consequence of the above statements. Indeed, suppose that $\trc_\L(p)(s)=(s_1,\ldots,s_n,s')$, that is, there exist $p_1,\ldots,p_n\in \mS$ and $\gamma^\L$-transitions
\[p,s\to  p_1,s_1 \to p_2,s_2 \to \cdots \to p_n,s_n\downarrow s'. \]
Applying \ref{statement-1} to the first $n$ transitions and \ref{statement-2} to the last one yields a situation as shown below, where the $\bullet$'s are elements of $(\mu\Sigma)_\Cos$, the transitions in the upper row are $\gamma^\L$-transitions and the remaining transitions are (weak) $\gamma_{\L^2}$-transitions.
\[
  \begin{tikzcd}[column sep=15, row sep=15,ampersand replacement=\&]
p,s \ar{rr} \ar{dd} \&\& p_1,s_1 \ar{d} \ar{rr} \&\& p_2,s_2 \ar{d} \& \cdots\quad \ar{r} \& p_n,s_n \ar{d} \&[-1.5em] \downarrow  s' \\
\&\& \bullet \ar[Rightarrow]{d} \&\& \bullet \ar[Rightarrow]{d} \&\& \bullet \ar[Rightarrow]{d}\&  \\
\bullet \ar[Rightarrow]{r}{s_1} \& \bullet \ar[Rightarrow]{r} \& \bullet \ar[Rightarrow]{r}{s_2} \& \bullet \ar[Rightarrow]{r} \& \bullet \& \quad \bullet \ar[Rightarrow]{r} \& \bullet \& \Downarrow  s'
\end{tikzcd}
\]
Thus $\trc_{\L^2}^\Trs(p)(s)=(s_1,\ldots,s_n,s')=\trc_{\L}(p)(s)$. The case where $\trc_\L(p)(s)$ is an infinite trace is treated analogously. It follows that $\trc_\L = \trc_{\L^2}^\Trs$ as required.

It remains to prove the above claims \ref{statement-1} and \ref{statement-2}. The proof is by structural induction on $p$.

\medskip\noindent \emph{Proof of \ref{statement-1}}. Let $p=\f(p_1,\ldots,p_n)$, and suppose that $p,s\to p',s'$ in $\gamma^\L$. We consider the cases of passive and active operators separately:

\medskip\noindent \underline{Case:} $\f$ is passive.\\
Let $c'\in (\mu\Sigma)_\Cos$ be the (unique) term such that $p',s'\to c'$ in $\gamma^{\L^2}$. Then we have the following $\gamma^{\L_2}$-transitions (we label transitions with the name of the rules of \Cref{fig:comp-rules-l2} that apply):
\[ p,s = \f(p_1,\ldots,p_n),s \xrightarrow[\text{($\f$1)}]{} s'.[p']_{s'} \xrightarrow[\text{($s'.-$)}]{s'} [p']_{s'},\qquad [p]_{s'} \xrightarrow[\text{($[-]$)}]{} c'. \]
Thus
\[ p,s \to s'.[p']_{s'} \To [p']_{s'},s',\qquad [p]_{s'}\To c',\qquad p',s'\to c'\To c', \]
which are transitions of the form required by (1).

\medskip\noindent \underline{Case:} $\f$ is active.\\
Let $j$ be the receiving position of $\f$. We consider two subcases:

\medskip\noindent \underline{Subcase:} The $\gamma^\L$-transition $p,s\to p',s'$ is induced by the rule
  \[  \inference{\goes{x_{j},s}{y_j,s'}}{\goes{\f(x_{1},\dots,x_j,\ldots,x_{n}),s}{\f(x_{1},\dots, y_j,\ldots x_{n})},s'},\]
in $\L$, that is,
\[  p_j,s\to p_j',s' \qand p,s =\f(p_1,\ldots,p_j,\ldots,p_n),s \to \f(p_1,\ldots,p_j',\ldots,p_n),s' = p',s'.  \]
By induction applied to the $\gamma^\L$-transition $p_j,s\to p_j',s'$, there exist $c_0,c_1,c_2,c_3\in (\mu\Sigma)_\Cos$ such that
\[ p_j,s\to c_0\xTo{s'} c_1,\qquad c_1\To c_2, \qquad p_j',s'\to c_3\To c_2.\]
Then
\begin{align*}
 p,s=\f(p_1,\ldots,p_j,\ldots,p_n),s & \xrightarrow[\text{($\f$3)}]{} \barf(p_1,\ldots,[p_j]_s,\ldots, p_n)\\ &\xrightarrow[\text{($\barf$1)},\,\text{($[-]$)}]{} \barf(p_1,\ldots,c_0,\ldots, p_n)\\
 & \xrightarrow[\text{($\barf$1),\,($\barf$2)}]{} \barf(p_1,\ldots,c_1,\ldots, p_n),s'
\end{align*}
and
\[ \barf(p_1,\ldots,c_1,\ldots, p_n)\xRightarrow[\text{($\barf$1)}]{} \barf(p_1,\ldots,c_2,\ldots, p_n)  \]
and
\begin{align*}
 p',s'=\f(p_1,\ldots,p_j',\ldots,p_n),s' & \xrightarrow[\text{($\f$3)}]{} \barf(p_1,\ldots,[p_j']_{s'},\ldots,p_n) \\ &\xrightarrow[\text{($\barf1$),\,($[-]$)}]{} \barf(p_1,\ldots,c_3,\ldots,p_n) \\ & \xRightarrow[\text{($\barf$1)}]{} \barf(p_1,\ldots,c_2,\ldots,p_n).
\end{align*}
We have thus shown that
\begin{align*}
& p,s\to \barf(p_1,\ldots,[p_j]_s,\ldots p_n) \To \barf(p_1,\ldots,c_1,\ldots, p_n),s',\\
&\barf(p_1,\ldots,c_1,\ldots p_n)\To \barf(p_1,\ldots,c_2,\ldots, p_n), \\
& p',s' \to  \barf(p_1,\ldots,[p_j']_{s'},\ldots,p_n) \To \barf(p_1,\ldots,c_2,\ldots,p_n),
\end{align*}
which are transitions of the required form.

\medskip\noindent \underline{Subcase:} The $\gamma^\L$-transition $p,s\to p',s'$ is induced by the rule
\[ \inference{\rets{x_{j},s}{s''}}{\goes{\f(x_{1},\dots,x_{n}),s}{t,s'}},\]
that is,
\[ p_j,s\downarrow s'',\quad p,s =\f(p_1,\ldots,p_j,\ldots,p_n),s \to t[p_1/x_1,\ldots,p_{j-1}/x_{j-1},p_{j+1}/x_{j+1},p_n/x_n],s'=p',s'. \]
By induction applied to the $\gamma^\L$-transition $p_j,s\downarrow s''$, there exist $c_0,c_1\in (\mu\Sigma)_\Cos$ such that
\[ p_j,s\to c_0 \To c_1 \downarrow s''. \]
Then
\begin{align*}
p,s=\f(p_1,\ldots,p_j,\ldots,p_n) & \xrightarrow[\text{($\f$3)}]{} \barf(p_1,\ldots,[p_j]_s,\ldots,p_n) \\
& \xrightarrow[\text{($\barf1$),\, ($[-]$)}]{} \barf(p_1,\ldots,c_0,\ldots,p_n) \\
& \xRightarrow[\text{($\barf$1)}]{} \barf(p_1,\ldots,c_1,\ldots,p_n) \\
& \xRightarrow[\text{($\barf$3)}]{s'} [p']_{s'}.
\end{align*}
Moreover, for $c'\in (\mu\Sigma)_\Cos$ such that $p',s'\to c'$ we have $[p']_{s'}\to c'$. We thus have shown that
\[ p,s\to \barf(p_1,\ldots,[p_j]_s,\ldots,p_n)\xTo{s'} [p']_{s'},\qquad [p']_{s'}\To c',\qquad p',s'\to c'\To c',   \]
which are transitions of the required form.

\medskip\noindent \emph{Proof of \ref{statement-2}}. Let $p=\f(p_1,\ldots,p_n)$, and suppose that $p,s\downarrow s'$ in $\gamma^\L$. Again we distinguish between passive and active operators:

\medskip\noindent \underline{Case:} $\f$ is passive.\\
Then there are $\gamma^{\L^2}$-transitions
\[p,s\to \ret_{s'}\downarrow s'\] by {\text{($\barf$2)}} and ($\ret$), respectively, which are transitions of form required by (2).

\medskip\noindent \underline{Case:} $\f$ is active.\\
Let $j$ be the receiving position of $\f$. Consider the rule
\[ \inference{\rets{x_{j},s}{s''}}{\rets{\f(x_{1},\dots,x_{n}),s}{s'}}\]
that induces the $\gamma^\L$-transition $p,s\downarrow s'$, that is,
\[ p_j,s\downarrow s'' \qqand p,s =\f(p_1,\ldots,p_j,\ldots,p_n),s \downarrow s'. \]
By induction applied to the $\gamma^\L$-transition $p_j,s\downarrow s''$, there exist $c_0,c_1\in (\mu\Sigma)_\Cos$ such that
\[ p_j,s\to c_0 \To c_1 \downarrow s''. \]
Then
\begin{align*}
p,s=\f(p_1,\ldots,p_j,\ldots,p_n),s & \xrightarrow[\text{($\f$3)}]{} \barf(p_1,\ldots,[p_j]_s,\ldots,p_n) \\
& \xrightarrow[\text{($\barf1$),\, ($[-]$)}]{} \barf(p_1,\ldots,c_0,\ldots,p_n) \\
& \xRightarrow[\text{($\barf$1)}]{} \barf(p_1,\ldots,c_1,\ldots,p_n) \\
\end{align*}
and moreover $\barf(p_1,\ldots,c_1,\ldots,p_n)\downarrow s'$ by ($\barf$4), so overall
\[ p,s \to  \barf(p_1,\ldots,[p_j]_s,\ldots,p_n) \Downarrow s', \]
which are transitions of the required form.
\qed

\subsection*{Proof of \Cref{thm:cong-l}}
We prove the statement for trace equivalence, the two other cases being completely analogous. Let $\f$ be an $n$-ary operator of $\Theta$ and let $p_i,p_i'\in \mS$ such that $\trc_\L(p_i)=\trc_\L(p_i')$ for $i=1,\ldots,n$.  Then $\trc_{\L^2}^\Trs(p_i)=\trc_{\L^2}^\Trs(p_i)$ for $i=1,\ldots,n$ by \Cref{prop:trace-pres}, and so
\[ \trc_{\L^2}^\Trs(\f(p_1,\ldots,p_n))=\trc_{\L^2}^\Trs(\f(p_1',\ldots,p_n')) \]
because trace equivalence is a congruence on the operational model of $\L^2$ (\Cref{thm:l2-comp}). Thus
\[ \trc_{\L}(\f(p_1,\ldots,p_n))=\trc_{\L}(\f(p_1',\ldots,p_n')) \]
by \Cref{prop:trace-pres} again, so trace equivalence is a congruence on the operational model of $\L$.
\qed

\subsection*{Proof of \Cref{prop:cont-eq}}
Clearly contextual equivalence is adequate (consider the empty context `$\cdot$') and a congruence. Conversely, if $R\monoto (\Tr,\Co)$ is an adequate congruence, then $R$ is contained in $\equiv^\ctx$. Indeed, if $R_\Trs(p,q)$, then for every context $C_\Trs[\cdot_\Trs]$ we have $R_\Trs(C[p],C[q])$ because $R$ is a congruence, whence $C[p],s\Downarrow \iff C[q],s\Downarrow$ for all $s\in \S(\Tr)$ because $R$ is adequate. Similar for contexts $C_\Cos[\cdot_\Trs]$. This proves $p\equiv^\ctx_\Trs q$. The case of writers is analogous.

\qed

\subsection*{Proof of \Cref{th:finalcong}}
\begin{notation}\label{not:power-relations}
 Given single-sorted relations $R\seq A\times A$ and $S\in B\times B$, we let $S^R\seq B^A\times B^A$ denote the relation on the function space given by
\[ S^R(f,g) \iff \forall a,a'\in A.\, (R(a,a')\implies S(f(a),g(a'))).   \]
Note that if also $R'\seq A\times A$ and $S'\seq B\times B$, then $S^R\bullet (S')^{R'}\seq (S\bullet S')^{R\bullet R'}$.
\end{notation}

We apply \Cref{thm:congruence-ho-abstract-gsos} to the HAOS
\[ \O=(\Sigma,\ol{\Sigma},B,\barB ,\rho,\gamma,\tilde{\gamma}) \]
where $\rho$ is the higher-order GSOS law for \reflang, the relation lifting $\barB$ of $B$ \eqref{eq:behreflang} is defined by
\[
\barB(R,S) =  ((S_\Cos+\Delta_1)^{\S(R_\Trs)},\, \vec{\Pow}((V(S_\Trs)+\Delta_1)\times \S(S_\Trs) + S_\Cos \times (\S(S_\Trs)+\Delta_1)),
\]
using \Cref{not:power-relations}. Moreover $\tilde{\gamma}\colon (\Tr,\Co)\to B((\Tr,\Co),(\Tr,\Co))$ is the weakening of $\gamma$ given by \Cref{not:weak-trans-ref}:
\[ \tilde{\gamma}^\Trs = \gamma^\Trs \qand \tilde{\gamma}^\Cos(c)= \{ d \mid c\To d \} \cup \{ (d,s) \mid c\To d,s \} \cup \{ s\mid d\Downarrow s\} \cup \{ (v,s) \mid c\Downarrow v,s \}.  \]
Note that $\barB(\Delta_\Trs,-)$-similarity on the operational model $((\Tr,\Co),\gamma,\tilde{\gamma})$ is precisely termination similarity (\Cref{def:termbisimref}). We now only need to verify the conditions \ref{ho-comp-cond1}--\ref{ho-comp-cond3} of \Cref{thm:congruence-ho-abstract-gsos}.
\begin{enumerate}
\item Up-closure of $\barB$ is immediate from up-closure of $\vec{\Pow}$.  For every $X,Y\in \Set^2$ we have
\begin{align*}
\barB(\Delta_X,\Delta_Y) &= ((\Delta_{Y_\Cos}+\Delta_1)^{\S(X_\Trs)}, \vec{\Pow}((V(\Delta_{Y_\Trs})+\Delta_1)\times \S(\Delta_{Y_\Trs}) + \Delta_{Y_\Cos} \times (\S(\Delta_{Y_\Trs})+\Delta_1))) \\
&= (\Delta_{(Y_\Cos+1)^{\S(X_\Trs)}}, \vec{\Pow}(\Delta_{(V(Y_\Trs)+1)\times \S(Y_\Trs)+Y_\Cos \times (\S(Y_\Trs)+1)}))   \\
&\supseteq (\Delta_{(Y_\Cos+1)^{\S(X_\Trs)}}, \Delta_{\Pow((V(Y_\Trs)+1)\times \S(Y_\Trs)+Y_\Cos \times (\S(Y_\Trs)+1))}))   \\
&= \Delta_{B(X,Y)}
\end{align*}
 For $R\monoto X\times X$ and $S,T\monoto Y\times Y$ in $\Set^2$ we have
\begin{align*}
&\barB(R,S\bullet T) \\
& =  (((S\bullet T)_\Cos +\Delta_1)^{\S(R_\Trs)},\, \vec{\Pow}((V((S\bullet T)_\Trs)+\Delta_1)\times \S((S\bullet T)_\Trs) + (S\bullet T)_\Cos \times (\S((S\bullet T)_\Trs)+\Delta_1)),\\
&\supseteq (S_\Cos+\Delta_1)^{\S(R_\Trs)},\, \vec{\Pow}((V(S_\Trs)+\Delta_1)\times \S(S_\Trs) + S_\Cos \times (\S(S_\Trs)+\Delta_1)) \\
&\quad \bullet ((T_\Cos+\Delta_1)^{\S(\Delta_{X_\Trs})},\, \vec{\Pow}((V(T_\Trs)+\Delta_1)\times \S(T_\Trs) + T_\Cos \times (\S(T_\Trs)+\Delta_1))\\
&= \barB(R,S)\bullet \bar(\Delta_X,T)
\end{align*}
using the property noted in \Cref{not:power-relations} in the penultimate step.
\item We have to show that for each $R\monoto X\times X$ and $Q
  \monoto Y \times Y$ in $\Set^2$ the map
  $\rho_{X,Y}\colon \ol{\Sigma}(R\times \barB(R,Q))\to
  \barB\ol{\Sigma}^{\star}(R+Q)$ is a morphism of relations. Again, this is in essence just a formal version of the observation that the rules of \reflang are parametrically polymorphic. Most of the cases
  are identical to those of \whiletwo. The more interesting ones are those of
  the allocation $\&$ and assignment $e \ass c$. For the former, the liftability
  condition on $\&$ suggests that if two writers $c$ and $c'$ return values
  values $v,v'$ (resp.) such that $V(Q)(v,v')$ and stores $s, s'$ (resp.) such
  that $\store(Q)(s,s')$, then $\store(Q)(s[l \mapsto v],s'[l \mapsto v'])$,
  which is true. For the latter, the important observation is that the
  evaluation $\oname{eev}_{X} : \store(X) \times \expr \rightharpoonup V(X)$ is
  itself liftable: evaluating the same expression in related stores produces
  related values. In particular, if $\store(Q)(s,s')$ then if
  $\oname{eev}_{\Trs}(e,s) = l$, $\oname{eev}_{\Trs}(e,s)= l$.
\item The lax bialgebra condition is once again a matter of observing that the rule of \reflang remain sound in the operational model of \reflang when strong transitions are replaced with weak ones according to \Cref{not:weak-trans-ref}. The argument is much analogous to \whiletwo in the case of termination semantics. For example, the weak instance
\[
\inference{c \Downarrow v,s \quad l \not\in \mathrm{dom}(s)}{\&c \Downarrow
      l,s[l \mapsto v]}
\]
of the third rule for $\&$ is sound because it emerges by repeated application of the first and second rule for $\&$, followed by a single application of the third one.
\end{enumerate}

\qed
}{}
\end{document}

%% file: catprog.tex
\providecommand{\catname}{\mathbf} 
\providecommand{\clsname}{\mathcal}
\providecommand{\oname}[1]{{\operatorname{\mathsf{#1}}}}

\def\defcatname#1{\expandafter\def\csname B#1\endcsname{\catname{#1}}}
\def\defcatnames#1{\ifx#1\defcatnames\else\defcatname#1\expandafter\defcatnames\fi}
\defcatnames ABCDEFGHIJKLMNOPQRSTUVWXYZ\defcatnames

\def\defclsname#1{\expandafter\def\csname C#1\endcsname{\clsname{#1}}}
\def\defclsnames#1{\ifx#1\defclsnames\else\defclsname#1\expandafter\defclsnames\fi}
\defclsnames ABCDEFGHIJKLMNOPQRSTUVWXYZ\defclsnames

\def\defbbname#1{\expandafter\def\csname BB#1\endcsname{{\bm{\mathsf{#1}}}}}
\def\defbbnames#1{\ifx#1\defbbnames\else\defbbname#1\expandafter\defbbnames\fi}
\defbbnames ABCDEFGHIJKLMNOPQRSTUVWXYZ\defbbnames

\def\Set{\catname{Set}}

\providecommand{\argument}{\operatorname{-\!-}}

\DeclareOldFontCommand{\bf}{\normalfont\bfseries}{\mathbf}
\providecommand{\Id}{\operatorname{Id}}

\providecommand{\id}{\mathsf{id}}
\providecommand{\op}{\mathsf{op}}
\providecommand{\comp}{\mathbin{\circ}}

\providecommand{\xto}[1]{\,\xrightarrow{#1}\,}

\providecommand{\To}{\mathrel{\Rightarrow}}			           %

\providecommand{\dar}{\kern-1.2pt\operatorname{\downarrow}}	
\providecommand{\uar}{\kern-1.2pt\operatorname{\uparrow}}	

\providecommand{\fst}{\oname{fst}}

\providecommand{\brks}[1]{\langle #1\rangle}

\usepackage{stmaryrd}

\providecommand{\lsem}{\llbracket}
\providecommand{\rsem}{\rrbracket}
\providecommand{\sem}[1]{\lsem #1 \rsem}

\providecommand{\pacman}[1]{}					                     %

\newcommand{\undefine}[1]{\let #1\relax}					                       %

\providecommand{\mone}{{\text{\kern.5pt\rmfamily-}\mathsf{\kern-.5pt1}}}

\providecommand{\smin}{\smallsetminus}

\makeatletter
\def\mfix#1{\oname{#1}\@ifnextchar\bgroup\@mfix{}}	       %
\def\@mfix#1{#1\@ifnextchar\bgroup\mfix{}}			           %
\makeatother

\providecommand{\case}[3]{\mfix{case}{\mathbin{}#1}{of}{#2}{\kern-1pt;}{\mathbin{}#3}}